\RequirePackage{fix-cm}

\documentclass[smallextended]{svjour3}       % onecolumn (second format)

\smartqed  % flush right qed marks, e.g. at end of proof

\usepackage{graphicx}

\usepackage{mathptmx}      % use Times fonts if available on your TeX system

\usepackage{changepage}   % for the adjustwidth environment
\usepackage{sidecap}
\usepackage{frcursive}
\usepackage{setspace}
\usepackage{caption}
\usepackage[multiple]{footmisc}
\usepackage{multirow}

\usepackage{wrapfig}
\usepackage{colortbl}
\usepackage{etex}
\usepackage{bbding}
\usepackage{dingbat}
\usepackage[ruled]{algorithm2e}
\usepackage{amsmath}
\usepackage{algorithmic}
\usepackage{amsfonts}
\usepackage{amssymb}
\usepackage{latexsym}
\usepackage{graphics}
\usepackage{color}
\usepackage[normalem]{ulem}
\usepackage{balance}
\usepackage{textcomp}
\usepackage{nicefrac}
\usepackage{wasysym}
\usepackage[inline]{enumitem}
\usepackage{gensymb}
\usepackage{longtable}
\usepackage[numbers,sort&compress]{natbib}

\usepackage{pifont,epsfig,rotating}

\usepackage{etoolbox}

\DeclareMathAlphabet{\mathpzc}{OT1}{pzc}{m}{it}

\usepackage{etaremune}

\usepackage{mathtools}

\usepackage[final,authormarkup=none,deletedmarkup=sout]{changes}
\setdeletedmarkup{{\color{red}\sout{#1}}}
\usepackage{ctable}

\definecolor{dgreyblue}{rgb}{0.26,0.3,0.46}             %89.25 102 158.10

\newcommand{\bee}{{\mathbf{e}}}
\newcommand{\cA}{\mathcal{A}}
\newcommand{\cB}{\mathcal{B}}
\newcommand{\cC}{\mathcal{C}}
\newcommand{\cD}{\mathcal{D}}
\newcommand{\cE}{\mathcal{E}}

\newcommand{\cK}{\mathcal{K}}
\newcommand{\cL}{\mathcal{L}}

\newcommand{\cR}{\mathcal{R}}

\newcommand{\cT}{\mathcal{T}}
\newcommand{\cU}{\mathcal{U}}

\newcommand{\R}{{\mathbb R}}  %ams bold
\renewcommand{\text}[1]{\hbox{\rm \ #1\ \/}}

\newcommand{\be}[1]{\begin{equation}\label{#1}}
\newcommand{\ee}{\end{equation}}
\newcommand{\beqn}{\begin{eqnarray*}}
\newcommand{\eeqn}{\end{eqnarray*}}
\newcommand{\beq}{\begin{eqnarray}}
\newcommand{\eeq}{\end{eqnarray}}
\newcommand{\ben}{\begin{enumerate}}
\newcommand{\een}{\end{enumerate}}
\newcommand{\bi}{\begin{itemize}}
\newcommand{\ei}{\end{itemize}}
\newcommand{\eps}{\varepsilon}

\newcommand{\IE}{{\em i.e.}\xspace}
\newcommand{\tx}{^{\mathrm{th}}}
\newcommand{\Ave}[1]{{\mathbb E}[ #1]}

\newcommand{\prob}[1]{\Pr\left[ #1\right]}

\newtheorem{fact}{Fact}

\renewenvironment{proof}{{\noindent\bf Proof.\ }}{\hfill{\Pisymbol{pzd}{113}}\vspace{0.1in}}

\newcommand{\NP}{\mathsf{NP}}
\newcommand{\LP}{{\mathsf{LP}}}
\newcommand{\SDP}{{\mathsf{SDP}}}
\renewcommand{\deg}{\mathsf{deg}}

\newcommand{\cS}{\mathcal{S}}

\newcommand{\EA}{{\em et al.}\xspace}
\newcommand{\TB}{\vspace{-0.1ex}}\newcommand{\TiE}{\setlength{\itemsep}{-1ex}}

\newcommand{\junk}[1]{}

\newcommand{\EG}{{\it e.g.}\xspace}
\newcommand{\FI}[1]{Fig.~\ref{#1}\xspace}

\newcommand{\h}{{\mathfrak{h}}}

\newcommand{\wwhile}{{\bf{while}}}
\newcommand{\rreturn}{{\bf{return}}}

\newcommand{\ddo}{{\bf{do}}}

\newcommand{\opt}{ {\mathsf{OPT}} }
\newcommand{\nopt}{ {\mathsf{OPT}_{\#}} }

\newcommand{\vopt}{ V_{\mathrm{opt}} }

\newcommand{\vhat}{\widehat{V}}

\newcommand{\eqdef}{\stackrel{\mathrm{def}}{=}}

\definecolor{columbiablue}{rgb}{0.61, 0.87, 1.0}

\newcommand{\fmc}{\textsc{Fmc}\xspace}
\newcommand{\nfmc}{\textsc{Node-fmc}\xspace}
\newcommand{\sfmc}{\textsc{Segr-fmc}\xspace}
\newcommand{\bfmc}{\ensuremath{\Delta}-\textsc{bal-fmc}\xspace}
\newcommand{\gfmc}{\textsc{Geom-fmc}\xspace}
\newcommand{\fmcg}[2]{\textsc{Fmc}(\ensuremath{#1, #2})}
\newcommand{\mindisc}{\textsc{Min-disc}\xspace}
\newcommand{\algmed}{\textsc{Alg}-\textsc{medium}-\textsc{opt}\ensuremath{_\#}\xspace}
\newcommand{\algsm}{\textsc{Alg}-\textsc{small}-\textsc{opt}\ensuremath{_\#}\xspace}
\newcommand{\algla}{\textsc{Alg}-\textsc{large}-\textsc{opt}\ensuremath{_\#}\xspace}
\newcommand{\algiter}{\textsc{Alg}-\textsc{iter}-\textsc{round}\xspace}
\newcommand{\alggreed}{\textsc{Alg}-\textsc{greed}-\textsc{plus}\xspace}
\newcommand{\alggreedsmpl}{\textsc{Alg}-\textsc{greedy}\xspace}
\newcommand{\alggeom}{\textsc{Alg}-\textsc{geom}\xspace}
\newcommand{\optf}{\ensuremath{\opt_{\mathrm{frac}}}}

\newcommand{\tf}{\ensuremath{t^{\mathrm{final}}}}
\newcommand{\x}{\mathbf{x}}

\newcommand{\khat}{\ensuremath{\widehat{k}}}
\newcommand{\Ghat}{\ensuremath{\widehat{G}}}
\newcommand{\Vhat}{\ensuremath{\widehat{V}}}
\newcommand{\Ehat}{\ensuremath{\widehat{E}}}
\newcommand{\Xhat}{\ensuremath{\widehat{X}}}

\newcommand{\phat}{\ensuremath{{\widehat{p}}}}
\newcommand{\qhat}{\ensuremath{{\widehat{q}}}}
\newcommand{\nhat}{\ensuremath{\widehat{n}}}
\newcommand{\mhat}{\ensuremath{\widehat{m}}}
\newcommand{\rhat}{\ensuremath{\widehat{r}}}
\newcommand{\Vsol}{\ensuremath{\Vhat_{\mathrm{sol}}}}

\newcommand{\Xsol}{\ensuremath{\Xhat_{\mathrm{sol}}}}

\allowdisplaybreaks

\makeatletter
\def\namedlabel#1#2{\begingroup
    #2%
    \def\@currentlabel{#2}%
    \label{#1}\endgroup
}
\makeatother

\usepackage[scr=boondoxo]{mathalfa}

\usepackage{hyperref}
\hypersetup{
    colorlinks=true,
    linkcolor=blue,
    filecolor=magenta,      
    urlcolor=cyan,
		citecolor=blue,
}
\journalname{}

\begin{document} 

\title{Maximizing coverage while ensuring fairness: a tale of conflicting objectives}

\titlerunning{Maximizing coverage while ensuring fairness}

\author{
Abolfazl Asudeh
\and
Tanya Berger-Wolf
\and
Bhaskar DasGupta$^\ast$
\and
Anastasios Sidiropoulos
}

\institute{
Abolfazl Asudeh \at
           Department of Computer Science, University of Illinois at Chicago, Chicago, IL 60607, USA \\
           \email{asudeh@uic.edu} 
\and
Tanya Berger-Wolf \at
           Department of Computer Science and Engineering, Ohio State University, Columbus, OH 43210, USA  \\
           \email{berger-wolf.1@osu.edu}
\and
$^\ast$Bhaskar DasGupta (corresponding author) \at
           Department of Computer Science, University of Illinois at Chicago, Chicago, IL 60607, USA \\
           Tel.: +312-355-1319\\
           Fax: +312-413-0024\\
           \email{bdasgup@uic.edu} 
\and
Anastasios Sidiropoulos \at
           Department of Computer Science, University of Illinois at Chicago, Chicago, IL 60607, USA \\
           \email{sidiropo@uic.edu}
}

\date{Received: date / Accepted: date}
% The correct dates will be entered by the editor

\maketitle

\begin{abstract}
Ensuring fairness in computational problems has emerged as a \emph{key} topic during recent years, 
buoyed by considerations for equitable resource distributions and social justice. 
It \emph{is} possible to incorporate fairness in computational problems from several perspectives, 
such as using optimization, game-theoretic
or machine learning frameworks. 
In this paper we address the problem of incorporation of fairness from a 
\emph{combinatorial optimization} perspective.
We formulate a combinatorial optimization framework, suitable for analysis by 
researchers in approximation algorithms and related areas, that incorporates fairness in maximum coverage 
problems as an interplay between \emph{two} conflicting objectives. 
Fairness is imposed in coverage by using coloring constraints that \emph{minimizes} the discrepancies between 
number of elements of different colors covered by selected sets; this is in contrast to the usual 
discrepancy minimization problems studied extensively in the literature where (usually two) 
colors are \emph{not} given \emph{a priori} 
but need to be selected to minimize the maximum color discrepancy of \emph{each} individual set.
Our main results are a set of randomized and deterministic approximation algorithms that attempts to 
\emph{simultaneously} approximate both fairness and coverage in this framework.
\end{abstract}

%\PACS{02.10.Ox \and 89.20.Ff \and 02.40.Pc}
%\subclass{MSC 68Q25 \and MSC 68W25 \and MSC 68W40 \and MSC 05C85}

%%%%%%%%%%%%%%%%%%%%%%%%%%%%%%%%%%%%%%%%%%%%%%%%%%%%%%%%%%%%%%%%%%%%%%%%%%%%%%%
%%%%%%%%%%%%%%%%%%%%%%%%%%%%%%%%%%%%%%%%%%%%%%%%%%%%%%%%%%%%%%%%%%%%%%%%%%%%%%%
%%%%%%%%%%%%%%%%%%%%%%%%%%%%%%%%%%%%%%%%%%%%%%%%%%%%%%%%%%%%%%%%%%%%%%%%%%%%%%%

\section{Introduction}
\label{sec-intro}

In this paper we introduce and analyze a combinatorial optimization framework capturing two \emph{conflicting} objectives: 
\emph{optimize} the main objective while trying to ensure that the selected solution is as \emph{fair} as possible. 
We illustrate the framework with the following \emph{simple} graph-theoretic illustration. 
Consider the graph $G$ of $10$ nodes and $18$ edges as shown in  
\FI{fig1} where each edge is colored from one of $\chi=3$ colors (red, blue or green) 
representing three different attributes.
Suppose that we want to select \emph{exactly} $k=3$ nodes that maximizes the number of edges they ``cover'' 
subject to the ``fairness'' constraint that the proportion of red, blue and green edges in the selected edges 
are the \emph{same}. 
%%%
An optimal solution is shown in \FI{fig1} by the solid black nodes $u_1, u_2, u_3$ covering $6$ edges;
\FI{fig1} also shows that the solution 
is quite different from what it would have been 
(the yellow corner nodes $v_1, v_2, v_3$ covering $11$ edges) 
\emph{if} the fairness constraint
was absent.
%%%
A simple consequence of the analysis of our algorithms for a more general setting is that, 
assuming that there exists \emph{at least} one feasible solution and assuming 
$k$ is large enough, 
we can find a randomized solution to this fair coverage problem for graphs 
where we select \emph{exactly} $k$ nodes, 
cover at least $63\%$ of the optimal number of edges on an average 
and, for \emph{every} pair of colors, 
with high probability 
the ratio of the number of edges of these two colors among the selected edges 
is $O(1)$.

%%%%%%%%%%%%%%%%%%%%%%%%%%%%%%%%%%%%%%%%%%%%%%%%%%%%%%%%%%%%%%%%%%%%%%%%%%%%%%%
%%%%%%%%%%%%%%%%%%%%%%%%%%%%%%%%%%%%%%%%%%%%%%%%%%%%%%%%%%%%%%%%%%%%%%%%%%%%%%%

\begin{figure}[ht]
\centerline{\includegraphics[width=0.5\textwidth]{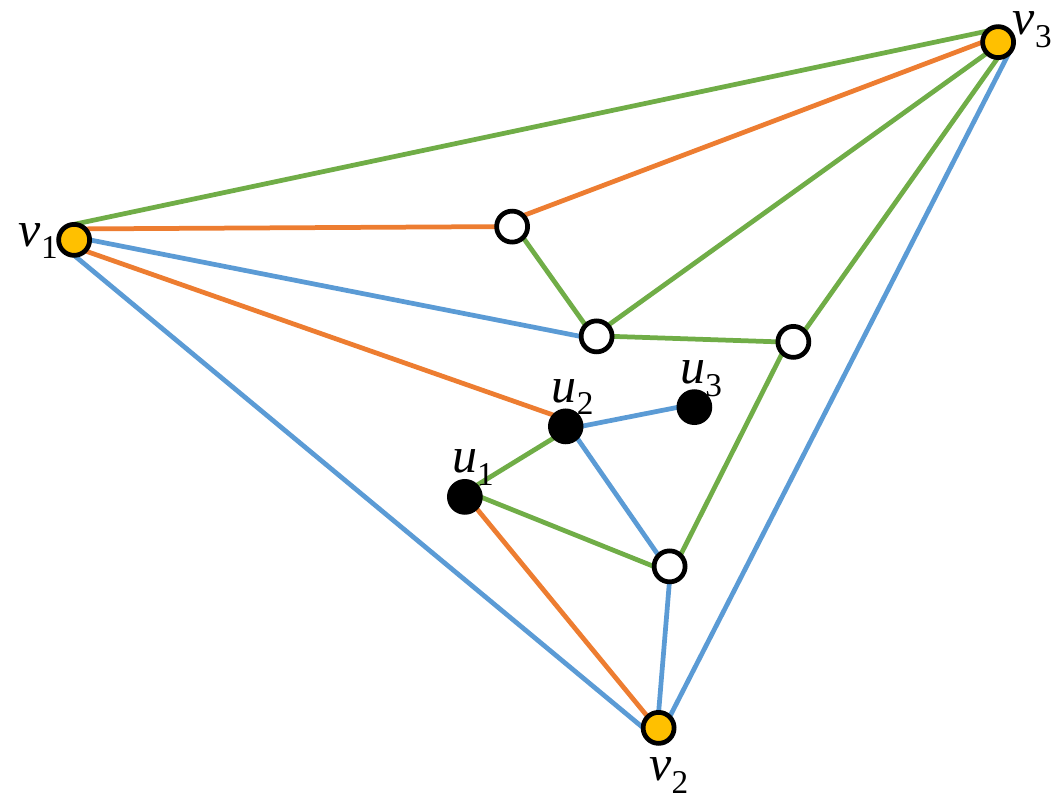}}
\caption{\label{fig1}A simple illustration of fairness in maximum coverage problems for graphs.}
\end{figure}
%%%%%%%%%%%%%%%%%%%%%%%%%%%%%%%%%%%%%%%%%%%%%%%%%%%%%%%%%%%%%%%%%%%%%%%%%%%%%%%
%%%%%%%%%%%%%%%%%%%%%%%%%%%%%%%%%%%%%%%%%%%%%%%%%%%%%%%%%%%%%%%%%%%%%%%%%%%%%%%

In this paper we consider this type of problem in more general settings. 
Of course, in the example in \FI{fig1} (and in general) there is nothing special about requiring that 
the proportion of red, blue and green edges in the covered edges 
should be \emph{exactly} equal as opposed to a \emph{pre-specified} unequal proportion. 
For example, we may also require that the proportion of edges of different colors in our solution 
should mimic that in the entire graph, \IE, in \FI{fig1}
among the covered edges 
the proportion of red, blue and green edges should be $q_1,q_2$ and $q_3$ 
where 
$q_1=\nicefrac{1}{6}$, 
$q_2=\nicefrac{1}{3}$, 
and 
$q_3=\nicefrac{1}{2}$. 
Our algorithms will work with \emph{easy} modifications for any \emph{constant} values of 
$q_1$, $q_2$ and $q_3$.

%%%%%%%%%%%%%%%%%%%%%%%%%%%%%%%%%%%%%%%%%%%%%%%%%%%%%%%%%%%%%%%%%%%%%%%%%%%%%%%
\subsection{Different research perspectives in ensuring fairness}
%%%%%%%%%%%%%%%%%%%%%%%%%%%%%%%%%%%%%%%%%%%%%%%%%%%%%%%%%%%%%%%%%%%%%%%%%%%%%%%

Theoretical investigations 
of ensuring fairness in computation can be pursued from
many perspectives. We briefly comment on a few of them.
	
One line of research dealing with the goal of ensuring fairness uses the \emph{optimization} framework, \IE, 
we model the problem as an optimization problem with precisely defined fairness constraints. 
This is a common framework used by researchers in combinatorial and graph-theoretic algorithms, 
such as research works that involve designing exact or approximation algorithms, 
investigating fixed-parameter tractability issues
or proving inapproximability results.
In this paper we use such a framework.
Fairness is imposed in coverage by using coloring constraints that minimizes the discrepancies between 
different colors among elements covered by selected sets; this is in contrast to the usual 
discrepancy minimization problems studied extensively in the literature~\cite{C00}
where the (usually two) colors are \emph{not} given a priori 
but need to be selected to \emph{minimize} the maximum color discrepancy of \emph{each} individual set.

A second line of research dealing with fairness involves machine learning frameworks.
Even though it is a relatively new research area, 
there is already a large body of research dealing with ensuring fairness in machine learning algorithms by 
preprocessing the data used in the algorithms, optimization of statistical outcomes 
with appropriate fairness criteria and metric 
during the training, or by post-processing the answers of the algorithms~\cite{Ze13,HPS16,ZLM18}.

A third line of research dealing with fairness involves game theoretic frameworks.
For example, developments of solutions for fair ways of sharing 
transferable utilities in cooperative game-theoretic environments have 
given rise to interesting concepts such as Shapley values and Rabin's fairness model.
We refer the reader to the excellent textbook in algorithmic game theory by 
Nisan \EA~\cite{NRTV07} for further details on these research topics.

Yet another more recent line of research dealing with fairness involve applying 
fairness criteria in the context of clustering of points in a \emph{metric space}
under $k$-means objective, $k$-median or other $\ell_p$-norm objectives~\cite{BCFN19,CKLV17}.
The assumption of an underlying metric space 
allows the development of efficient algorithms for these frameworks.

%%%%%%%%%%%%%%%%%%%%%%%%%%%%%%%%%%%%%%%%%%%%%%%%%%%%%%%%%%%%%%%%%%%%%%%%%%%%%%%
\section{Fair maximum coverage: notations, definitions and related concepts}
%%%%%%%%%%%%%%%%%%%%%%%%%%%%%%%%%%%%%%%%%%%%%%%%%%%%%%%%%%%%%%%%%%%%%%%%%%%%%%%

The \emph{Fair Maximum Coverage} problem with $\chi$ colors 
is defined as follows.
We are given an universe $\cU=\{u_1,\dots,u_n\}$ of $n$ elements, 
a weight function $w: \cU\mapsto \R$ assigning a non-negative weight to every element,  
a color function $\cC: \cU\mapsto \{1,\dots,\chi\}$ assigning a color to every element,  
a collection of $m$ sets $\cS_1,\dots,\cS_m\subseteq \cU$, 
and
a positive integer $k$. 
A collection of $k$ distinct subsets,
say $\cS_{i_1},\dots,\cS_{i_k}$, 
with the set of ``covered'' elements
$\bigcup_{j=1}^k \cS_{i_j}$ containing $p_i$ elements of color $i$ 
is considered a 
\emph{valid solution}\footnote{For a more general version of the problem 
we are given 
$\chi$ ``color-proportionality constants'' $q_1,\dots,q_\chi\in(0,1]$ with $q_1+\dots+q_\chi=1$, 
and a valid solution must satisfy 
$\nicefrac{p_i}{p_j}=\nicefrac{q_i}{q_j}$ for all $i$ and $j$. 
As we mentioned already, 
with suitable modifications 
our algorithms will work with similar asymptotic performance guarantee 
for any \emph{constant} values of 
$q_1,\dots,q_\chi$, but to simplify exposition we will assume the simple requirement 
of $q_1=\dots=q_\chi$ in the sequel.}
provided 
$p_i=p_j$ for all $i$ and $j$.
The objective is to \emph{maximize} 
the sum of weights of the covered elements. 
\added[comment={added}]{More explicitly, our problem is defined as follows:}

\medskip
%%%%%%%%%%%%%%%%%%%%%%%%%%%%%%%%%%%%%%%%%%%%%%%%%%%%%%%%%%%%%%%%%%%%%%%%%
\begin{center}
\begin{tabular}{r l}
\toprule
\textbf{Problem name:} & {Fair Maximum Coverage with $\chi$ colors (\fmcg{\chi}{k})}
\\
\\
\textbf{Input:} & 
   {$\bullet$ universe $\cU=\{u_1,\dots,u_n\}$}
\\
& {$\bullet$ (element) weight function $w: \cU\mapsto \R^+\cup\{0\}$}
\\
& {$\bullet$ (element) color function $\cC: \cU\mapsto \{1,\dots,\chi\}$}
\\
& {$\bullet$ sets $\cS_1,\dots,\cS_m\subseteq \cU$}
\\
& {$\bullet$ integer $k>0$}
\\
\\
\textbf{Valid solution:} & 
      {collection of $k$ distinct subsets $\cS_{i_1},\dots,\cS_{i_k}$ satsifying}
\\
& 
\begin{tabular}{l}
{$\forall i,j\in\{1,\dots,\chi\}:$} 
\\
\hspace*{0.5in}
   {$p_i \eqdef \big| \, \{u_\ell \,|\, u_\ell \in \bigcup_{j=1}^k \cS_{i_j} \text{ and } \cC(u_\ell)=i \} \,\big|$}
\\
\hspace*{0.5in}
     {$=$}
\\
\hspace*{0.5in}
    {$p_j \eqdef \big| \, \{u_\ell \,|\, u_\ell \in \bigcup_{j=1}^k \cS_{i_j} \text{ and } \cC(u_\ell)=j \} \,\big|$}
\end{tabular}
\\
\\
\textbf{Objective:} & 
    {\emph{maximize} $\sum_{u_\ell\in\bigcup_{j=1}^k \cS_{i_j}} w(u_\ell)$}
%%%%%%%%%%%%%
\\
\bottomrule
\end{tabular}
\end{center}
%%%%%%%%%%%%%%%%%%%%%%%%%%%%%%%%%%%%%%%%%%%%%%%%%%%%%%%%%%%%%%%%%%%%%%%%%
\medskip

\deleted{We denote this problem by \fmcg{\chi}{k}, or just \fmc when $\chi$ and $k$ are clear from the context.}
\added{We denote \fmcg{\chi}{k} by just \fmc when $\chi$ and $k$ are clear from the context.}
In the sequel, we will distinguish between the following two versions of the problem: 
%%%%%%%%%%%%%%%%%%%%%%%%%%%%%%%%%%%%%%%%%%%%%%%%%%%%%%%%%%%%%%%%%%%%%%%%%
\begin{enumerate}[label=(\emph{\roman*})]
\item
\emph{unweighted} \fmc
in which $w(u_\ell)=1$ for all $\ell\in\{1,\dots,n\}$ and thus the objective is to maximize
the \emph{number} of elements covered, 
and 
\item
\emph{weighted} \fmc
in which $w(u_\ell)\geq 0$ for all $\ell\in\{1,\dots,n\}$.
\end{enumerate}
%%%%%%%%%%%%%%%%%%%%%%%%%%%%%%%%%%%%%%%%%%%%%%%%%%%%%%%%%%%%%%%%%%%%%%%%%
For the purpose of stating and analyzing algorithmic performances, we define the following 
notations and natural parameters associated with an instance of 
\fmcg{\chi}{k}: 
%%%%%%%%%%%%%%%%%%%%%%%%%%%%%%%%%%%%%%%%%%%%%%%%%%%%%%%%%%%%%%%%%%%%%%%%%%%%%%%%%%%%%%%%%%%%%%%%%%%%%%%
\begin{enumerate}[label=$\triangleright$]
%%%%%%%%%%%%%%%%%%%%%%%%%%%%%%%%%%%%%%%%%%%%%%%%%%%%%%%%%%%%%%%%%%%%%%%%%%%%%%%%%%%%%%%%%%%%%%%%%%%%%%%
\item
$a\in\{2,3,\dots,n\}$
denotes the maximum of the \emph{cardinalities} (number of elements) of all sets.
%%%%%%%%%%%%%%%%%%%%%%%%%%%%%%%%%%%%%%%%%%%%%%%%%%%%%%%%%%%%%%%%%%%%%%%%%%%%%%%%%%%%%%%%%%%%%%%%%%%%%%%
\item
$f\in\{1,2,\dots,m\}$ 
denotes the \emph{maximum} of the frequencies of all elements, where the \emph{frequency} 
of an element is 
the number of sets in which it belongs.
%%%%%%%%%%%%%%%%%%%%%%%%%%%%%%%%%%%%%%%%%%%%%%%%%%%%%%%%%%%%%%%%%%%%%%%%%%%%%%%%%%%%%%%%%%%%%%%%%%%%%%%
\item
$\opt$ denotes the optimal objective value of the given instance of \fmc.
%%%%%%%%%%%%%%%%%%%%%%%%%%%%%%%%%%%%%%%%%%%%%%%%%%%%%%%%%%%%%%%%%%%%%%%%%%%%%%%%%%%%%%%%%%%%%%%%%%%%%%%
\item
$\opt_{\#}$ denotes the \emph{number of covered elements} in an optimal solution of the given instance of \fmc.
For weighted \fmc, if there are multiple optimal solutions then $\opt_\#$ will the {maximum} number of elements 
covered \emph{among} these optimal solutions.
Note that $\opt=\opt_\#$
for unweighted \fmc. 
The reason we need to consider $\opt_\#$ separately from 
$\opt$ for weighted \fmc is because the coloring constraints are tied to $\opt_\#$ 
whereas the optimization objective is tied to $\opt$. 
%%%%%%%%%%%%%%%%%%%%%%%%%%%%%%%%%%%%%%%%%%%%%%%%%%%%%%%%%%%%%%%%%%%%%%%%%%%%%%%%%%%%%%%%%%%%%%%%%%%%%%%
\item
The performance ratios of many of our algorithms are expressed using the function $\varrho(\cdot)$:
\[
{\varrho(x) \eqdef {\big( 1 - \nicefrac{1}{x} \big)}^x}
\]
Note that 
$\varrho(x)<\varrho(y)$ for $x>y>0$ and  
$\varrho(x)> 1 - \bee^{-1}$ for all $x>0$.
%%%%%%%%%%%%%%%%%%%%%%%%%%%%%%%%%%%%%%%%%%%%%%%%%%%%%%%%%%%%%%%%%%%%%%%%%%%%%%%%%%%%%%%%%%%%%%%%%%%%%%%
\end{enumerate}
%%%%%%%%%%%%%%%%%%%%%%%%%%%%%%%%%%%%%%%%%%%%%%%%%%%%%%%%%%%%%%%%%%%%%%%%%%%%%%%%%%%%%%%%%%%%%%%%%%%%%%%
For $\NP$-completeness results, if the 
problem is trivially in $\NP$ then we will \emph{not} mention it.
To analyze our algorithms in this paper, we have used several standard mathematical equalities or inequalities which 
are listed explicitly below for the convenience of the reader: 
%%%%%%%%%%%%%%%%%%%%%%%%%%%%%%%%%%%%%%%%%%%%%%%%%%%%%%%%%%%%%%%%%%%%%%%%%%%%%%%%%%%%%%%%%%%%%%%%%%%%%%%
\begin{align}
& 
\forall \, x\in[0,1]\,:\,
\bee^{-x} \geq 1-x 
\label{stdeq1}
\\
& 
\forall \, x \,:\, \bee^{-x} = 1 - x + ({x^2}/{2}) \bee^{-\xi} \text{ for some } \xi\in[0,x] 
\label{stdeq3}
\\
& 
\textstyle
\forall \, \alpha_1,\dots,\alpha_q \geq 0 \,:\, \left( \frac{1}{q} {\sum_{j=1}^q \alpha_j} \right)^{\!q}
        \geq \prod_{j=1}^q \alpha_j 
\label{stdeq-new}
\\
& 
\textstyle
\forall \, x\in[0,1]\,\,
\forall \, y\geq 1\,:\,
1- \left( 1 - \frac{x}{y} \right)^{y}
\geq
\left( 1- \left( 1 - \frac{1}{y} \right)^{y} \right) x 
\label{stdeq-new2}
\end{align}
%%%%%%%%%%%%%%%%%%%%%%%%%%%%%%%%%%%%%%%%%%%%%%%%%%%%%%%%%%%%%%%%%%%%%%%%%%%%%%%%%%%%%%%%%%%%%%%%%%%%%%%

%%%%%%%%%%%%%%%%%%%%%%%%%%%%%%%%%%%%%%%%%%%%%%%%%%%%%%%%%%%%%%%%%%%%%%%%%%%%%%%
\subsection{Three special cases of the general version of \fmc}
%%%%%%%%%%%%%%%%%%%%%%%%%%%%%%%%%%%%%%%%%%%%%%%%%%%%%%%%%%%%%%%%%%%%%%%%%%%%%%%

In this subsection we state three important special cases of the general framework of \fmc. 

%%%%%%%%%%%%%%%%%%%%%%%%%%%%%%%%%%%%%%%%%%%%%%%%%%%%%%%%%%%%%%%%%%%%%%%%%%%%%%%
\medskip 
\noindent
\textbf{Fair maximum $k$-node coverage or \nfmc} 
\medskip 
%%%%%%%%%%%%%%%%%%%%%%%%%%%%%%%%%%%%%%%%%%%%%%%%%%%%%%%%%%%%%%%%%%%%%%%%%%%%%%%

This captures the scenario posed by the example in \FI{fig1}. 
We are given a connected undirected edge-weighted graph $G=(V,E)$ where $w(e)\geq 0$ denotes 
the weight assigned to edge $e\in E$, 
a color function $\cC: E\mapsto \{1,\dots,\chi\}$ assigning a color to every edge,  
and a positive integer $k$. 
A node $v$ is said to \emph{cover} an edge $e$ if $e$ is incident on $v$. 
A collection of $k$ nodes
$v_{i_1},\dots,v_{i_k}$
covering 
$p_i$ edges of color $i$ for each $i$ 
is considered a 
valid solution
provided 
$p_i=p_j$ for all $i$ and $j$. 
The objective is to \emph{maximize} 
the sum of weights of the covered edges. 
It can be easily seen that this is a special case of \fmc by using the \emph{standard} translation 
from node cover to set cover, \IE, 
the edges are the set of elements, and corresponding to every node $v$ there is a set containing the 
edges incident on $v$. Note that for this special case $f=2$ and 
$a$ is equal to the maximum node-degree in the graph.

%%%%%%%%%%%%%%%%%%%%%%%%%%%%%%%%%%%%%%%%%%%%%%%%%%%%%%%%%%%%%%%%%%%%%%%%%%%%%%%
\medskip 
\noindent
\textbf{Segregated \fmc or \sfmc} 
\medskip 
%%%%%%%%%%%%%%%%%%%%%%%%%%%%%%%%%%%%%%%%%%%%%%%%%%%%%%%%%%%%%%%%%%%%%%%%%%%%%%%

Segregated \fmc
is the special case of \fmc when all the elements in any set have the \emph{same} color, \IE, 
\[
\forall \, j\in\{1,\dots,m\} \, 
\forall \, u_p,u_q\in \cS_j \,:\, 
\cC(u_p)=\cC(u_q)
\]
\added[comment={added}]{
Another equivalent way of describing 
\sfmc is as follows.
Let $C_j$ is the set of all elements colored $j$ for $j\in\{1,\dotsc,\chi\}$ in the given instance of \sfmc.
Let the notation $2^A$ denote the power set for any set $A$. Then, \sfmc is the special case when 
$\cS_j \in \bigcup_{j=1}^\chi 2^{C_j}$ holds for all $j\in\{1,\dots,m\}$.
A simple example of an instance of \sfmc 
with $n=6$, $m=10$ and $\chi=2$ 
is shown below: 
}
\added{
%%%%%%%%%%%%%%%%%%%%%%%%%%%%%%%%%%%%%%%%%%%%%%%%%%%%%%%%%%%%%%%%%%%%%%%%%%%%%%%
\begin{gather*}
\cU=\{u_1,u_2,u_3,u_4,u_5,u_6\}
\\
\cC(u_1)=\cC(u_2)=\cC(u_3)=\cC(u_4)=1
,\,\,\, 
\cC(u_5)=\cC(u_6))=2
\\
\cS_1=\{u_1,u_2,u_3\}
,\,\,\, 
\cS_2=\{u_2,u_3,u_4\}
,\,\,\, 
\cS_3=\{u_5,u_6\}
,\,\,\, 
\cS_4=\{u_6\}
\\
w(u_1)= w(u_2)=7,
,\,\,\, 
w(u_3)= 9, 
,\,\,\, 
w(u_4)= w(u_5)= w(u_6)= 1
\end{gather*}
%%%%%%%%%%%%%%%%%%%%%%%%%%%%%%%%%%%%%%%%%%%%%%%%%%%%%%%%%%%%%%%%%%%%%%%%%%%%%%%
}
Even though computing an exact solution of \sfmc is still $\NP$-complete, it is 
much easier to approximate 
(see Section~\ref{sec-segr}).
From our application point of view as discussed in Section~\ref{sec-appl},
this may for example model cases in which city neighborhoods are segregated in some manner, \EG, 
racially or based on income.

%%%%%%%%%%%%%%%%%%%%%%%%%%%%%%%%%%%%%%%%%%%%%%%%%%%%%%%%%%%%%%%%%%%%%%%%%%%%%%%
\medskip 
\noindent
\textbf{$\Delta$-balanced \fmc or \bfmc} 
\medskip 
%%%%%%%%%%%%%%%%%%%%%%%%%%%%%%%%%%%%%%%%%%%%%%%%%%%%%%%%%%%%%%%%%%%%%%%%%%%%%%%

$\Delta$-balanced \fmc
is the special case of \fmc when the number of elements of each color in a set are within 
an \emph{additive} range of $\Delta$, \IE, 
%%%%%%%%%%%%%%%%%%%%%%%%%%%%%%%%%%%%%%%%%%%%%%%%%%%%%%%%%%%%%%%%%%%%%%%%%%%%%%%
\begin{multline*}
\forall \, j\in\{1,\dots,m\} \, 
\forall \, p\in \{1,\dots,\chi\} \,:\, 
\\
\max \left\{ 1,\, \left\lfloor \frac{|\cS_j|}{\chi} \right\rfloor
-\Delta
\right\}
\leq
\left| u_\ell \,|\, ( u_\ell\in\cS_j ) \wedge ( \cC(u_\ell)=p ) \right|
\leq
\left\lceil \frac{|\cS_j|}{\chi} \right\rceil 
+\Delta
\end{multline*}
%%%%%%%%%%%%%%%%%%%%%%%%%%%%%%%%%%%%%%%%%%%%%%%%%%%%%%%%%%%%%%%%%%%%%%%%%%%%%%%
Similar to \sfmc, it is 
much easier to approximate \bfmc
for small $\Delta$
(see Section~\ref{sec-bala}).

%%%%%%%%%%%%%%%%%%%%%%%%%%%%%%%%%%%%%%%%%%%%%%%%%%%%%%%%%%%%%%%%%%%%%%%%%%%%%%%
\medskip 
\noindent
\textbf{Geometric \fmc or \gfmc} 
\medskip 
%%%%%%%%%%%%%%%%%%%%%%%%%%%%%%%%%%%%%%%%%%%%%%%%%%%%%%%%%%%%%%%%%%%%%%%%%%%%%%%

In this \emph{unweighted} ``geometric'' version of \fmc, 
the elements are points in $[0,\Delta]^d$ for some $\Delta$ and some constant $d\geq 2$, the 
sets are unit radius balls in $\R^d$, and the distributions of points of different colors 
are given by $\chi$
Lipschitz-bounded measures. 
More precisely, the distribution of points of color $i$ is given by 
a probability measure $\mu_i$
supported on $[0,\Delta]^d$ with a $C$-Lipschitz density 
function\footnote{A function 
$f:\Delta\mapsto \R$
for some subset $\Delta$ of real numbers is 
$C$-Lipschitz provided 
$|f(x)-f(y)| \leq C \,|x-y|$
for all real numbers $x,y\in \Delta$.} 
for some $C>0$ that is upper-bounded by $1$.
Given 
a set of $k$ unit balls 
$\cB_1,\dots,\cB_k \subset \R^d$,
the number of points $p_i$ of color $i$ covered by these balls is given by 
$
\mu_i \big(\bigcup_{i=1}^k \cB_i\big)
$, and 
the total number of points covered by these balls 
is given by 
$
\sum_{i=1}^\chi \mu_i\big(\bigcup_{i=1}^k \cB_i\big)
$.
%%%
This variant has an \emph{almost} optimal polynomial-time approximation algorithm for \emph{fixed} $d$ and 
under some \emph{mild} assumption on $\opt$ (see Section~\ref{sec-gfmc}).

%%%%%%%%%%%%%%%%%%%%%%%%%%%%%%%%%%%%%%%%%%%%%%%%%%%%%%%%%%%%%%%%%%%%%%%%%%%%%%%
\section{Sketch of application scenarios}
\label{sec-appl}
%%%%%%%%%%%%%%%%%%%%%%%%%%%%%%%%%%%%%%%%%%%%%%%%%%%%%%%%%%%%%%%%%%%%%%%%%%%%%%%

\fmc and its variants are \emph{core} abstractions of many data-driven societal domain applications. 
We present three \emph{diverse} categories of application and highlight 
the \emph{real-world} fairness issues addressed by our 
problem formulations (leaving other applications in the cited references).

%%%%%%%%%%%%%%%%%%%%%%%%%%%%%%%%%%%%%%%%%%%%%%%%%%%%%%%%%%%%%%%%%%%%%%%%%%%%%%%
\medskip
\noindent
\textbf{Service/Facility Allocation}
\medskip
%%%%%%%%%%%%%%%%%%%%%%%%%%%%%%%%%%%%%%%%%%%%%%%%%%%%%%%%%%%%%%%%%%%%%%%%%%%%%%%

One of the most common data based policy decisions is assigning services/facilities across different places, \EG, 
placing schools~\cite{ChicagoSchoolClosure}, bus stops, or police/fire stations, choosing a few hospitals for 
specific medical facilities or services, or deciding where to put cell-phone towers.
Of course, a major objective in such assignments is to serve the \emph{maximum} number of people 
(\IE, maximize the \emph{coverage}).
Unfortunately, historical discriminations, such as \emph{redlining}~\cite{redlining}, through their 
long drawn-out effects of manifestations 
in different aspects of public policy are \emph{still} hurting the minorities.
As a result, blindly optimizing for maximum coverage biases the assignment against \emph{equitable} distribution of services.
Below are examples of two \emph{real} cases that further underline the importance of fairness while maximizing coverage:
%%%%%%%%%%%%%%%%%%%%%%%%%%%%%%%%%%%%%%%%%%%%%%%%%%%%%%%%%%%%%%%%%%%%%%%%%%%%%%%
\begin{enumerate}[label=$\triangleright$,leftmargin=*]
%%%%%%%%%%%%%%%%%%%%%%%%%%%%%%%%%%%%%%%%%%%%%%%%%%%%%%%%%%%%%%%%%%%%%%%%%%%%%%%
\item {\em Bike sharing:}
As more and more cities adopt advanced transportation systems such as bike-sharing, concerns such as equity and 
fairness arise with them~\cite{yan2019fairness}. For instance, according to~\cite{BikeShare1} 
the bike-sharing network at NYC neglects many low-income neighborhoods and communities of color while giving the 
priority to well-to-do neighborhoods. Here the location of bike stations (or bikes) determines the set of people 
that will have access to the service, perpetuating the unhealthy cycle of lack of transportation, movement, \emph{etc}.
%%%%%%%%%%%%%%%%%%%%%%%%%%%%%%%%%%%%%%%%%%%%%%%%%%%%%%%%%%%%%%%%%%%%%%%%%%%%%%%
\item {\em Delivery services for online shopping:}
Online shopping has by now gained a 
major share of the shopping market.
Platforms such as Amazon provide services such as same-day delivery %% (or Prime Air delivery drone) 
to make e-shopping even more convenient to their customers.
While Amazon's main aim is to maximize the number of customers 
covered by this service, 
by not considering fairness it demonstrably \emph{failed} to provide such services for predominantly 
black communities~\cite{AmazonDeliveryBias1,AmazonDeliveryBias2}.
%%%%%%%%%%%%%%%%%%%%%%%%%%%%%%%%%%%%%%%%%%%%%%%%%%%%%%%%%%%%%%%%%%%%%%%%%%%%%%%
\end{enumerate}
%%%%%%%%%%%%%%%%%%%%%%%%%%%%%%%%%%%%%%%%%%%%%%%%%%%%%%%%%%%%%%%%%%%%%%%%%%%%%%%

%%%%%%%%%%%%%%%%%%%%%%%%%%%%%%%%%%%%%%%%%%%%%%%%%%%%%%%%%%%%%%%%%%%%%%%%%%%%%%%
\smallskip
\noindent
\textbf{Data Integration}
\medskip
%%%%%%%%%%%%%%%%%%%%%%%%%%%%%%%%%%%%%%%%%%%%%%%%%%%%%%%%%%%%%%%%%%%%%%%%%%%%%%%

Combining multiple data sources to augment 
the power of any individual data source
is a popular method for data collection. 
Naturally the main objective of data integration is to collect (``cover'') a maximum number of data points.
However, failing to include an adequate number of instances from minorities, known as {\em population bias}, 
in datasets used for training machine learning models is a major reason for model 
unfairness~\cite{olteanu2019social,asudeh2019assessing}.
For example, image recognition and motion detection services by Google~\cite{google-gorilla} and 
HP~\cite{hp1} with a reasonable overall performance failed to tag/detect African Americans since their 
training datasets did not include enough instances from this minority group.
While solely optimizing for coverage may result in biased datasets, considering fairness for integration may 
help remove population bias.

%%%%%%%%%%%%%%%%%%%%%%%%%%%%%%%%%%%%%%%%%%%%%%%%%%%%%%%%%%%%%%%%%%%%%%%%%%%%%%%
\medskip
\noindent
\textbf{Targeted advertisement}
\medskip
%%%%%%%%%%%%%%%%%%%%%%%%%%%%%%%%%%%%%%%%%%%%%%%%%%%%%%%%%%%%%%%%%%%%%%%%%%%%%%%

Targeted advertising is popular in social media.
Consider a company that wants to target its ``potential customers''.
To do so, the company needs to select a set of features (such as ``single'' or ``college student'') 
that specify the groups of users to be targeted.
Of course, the company wants to maximize coverage over its customers.
However, solely optimizing for coverage may result in incidents such as racism in the Facebook advertisements~\cite{AdRacism}
or sexism in the job advertisements~\cite{adSexism}.
Thus, a desirable goal for the company would be to select the keywords such that it 
provides fair coverage over users of diverse demographic groups.

%%%%%%%%%%%%%%%%%%%%%%%%%%%%%%%%%%%%%%%%%%%%%%%%%%%%%%%%%%%%%%%%%%%%%%%%%%%%%%%
\section{Review of prior related works}
\label{sec-prior-coverage}
%%%%%%%%%%%%%%%%%%%%%%%%%%%%%%%%%%%%%%%%%%%%%%%%%%%%%%%%%%%%%%%%%%%%%%%%%%%%%%%

To the best of our knowledge, \fmc in its full generalities has \emph{not} been separately 
investigated before. However, there are several prior lines of research that conceptually 
intersect with \fmc.

%%%%%%%%%%%%%%%%%%%%%%%%%%%%%%%%%%%%%%%%%%%%%%%%%%%%%%%%%%%%%%%%%%%%%%%%%%%%%%%
\medskip
\noindent
\textbf{Maximum $k$-set coverage and $k$-node coverage problems} 
\medskip
%%%%%%%%%%%%%%%%%%%%%%%%%%%%%%%%%%%%%%%%%%%%%%%%%%%%%%%%%%%%%%%%%%%%%%%%%%%%%%%

The maximum $k$-set coverage and $k$-node coverage 
problems are the same as the \fmc and \nfmc problems, respectively,
without element colors and without coloring constraints. 
These problems have been extensively studied in 
the algorithmic literature, \EG, 
see~\cite{AS04,F98,Ho97} for $k$-set coverage 
and~\cite{GLL18a,GLL18b,GNW07,M18,FL01,AS14,HYZZ02} for $k$-node coverage.
A summary of these results are as follows:

%%%%%%%%%%%%%%%%%%%%%%%%%%%%%%%%%%%%%%%%%%%%%%%%%%%%%%%%%%%%%%%%%%%%%%%%%%%%%%%
\begin{description}[leftmargin=15pt]
\item[\em $k$-set coverage:]
The best approximation algorithm 
for $k$-set coverage
is a deterministic algorithm that 
has an approximation ratio of 
$\max \big\{ \varrho(f) , \, \varrho(k) \big\} > 1-\nicefrac{1}{\bee}$~\cite{AS04,Ho97}. 
On the inapproximability side, 
assuming $\mathsf{P}\!\neq\!\NP$ an asymptotically optimal inapproximability ratio of $1-\nicefrac{1}{\bee}+\eps$ (for any $\eps>0$) 
is known for any polynomial-time algorithm~\cite{F98}.
%%%%%%%%%%%%%%%%%%%%%%%%%%%%%%%%%%%%%%%%%%%%%%%%%%%%%%%%%%%%%%%%%%%%%%%%%%%%%%%
\item[\em $k$-node coverage:]
The best approximation for 
$k$-node coverage 
is a randomized algorithm that has an approximation ratio of 
$0.7504$ with high probability~\cite{FL01,HYZZ02}.
On the inapproximability side, 
$k$-node coverage 
is $\NP$-complete even for bipartite graphs~\cite{AS14},
and 
cannot be approximated within a ratio of $1-\eps$ for some (small) constant $\eps>0$~\cite{L98,P94}.
More recently, 
Manurangsi~\cite{M18} 
provided an semidefinite programming based approximation algorithm with an approximation ratio of $0.92$, 
and Austrin and Stankovic~\cite{AS19} used the results in~\cite{AKS11} 
to provide an \emph{almost} matching upper bound of $0.929+\eps$ (for any $\eps>0$) 
on the approximation ratio of any polynomial time algorithm
assuming the \emph{unique games conjecture} is true.
There is also a significant body of prior research on the fixed parameter tractability issues 
for the $k$-node coverage problem: for example, 
$k$-node coverage 
is \emph{unlikely} to allow an FPT algorithm as it is $W[1]$-hard~\cite{GNW07}, but 
Marx designed an FPT \emph{approximation scheme} in~\cite{M08}
whose running times were subsequently improved in Gupta, Lee and Li in~\cite{GLL18a,GLL18b}.
\end{description}
%%%%%%%%%%%%%%%%%%%%%%%%%%%%%%%%%%%%%%%%%%%%%%%%%%%%%%%%%%%%%%%%%%%%%%%%%%%%%%%
\added{However, the coloring constraints make \fmc \emph{fundamentally different} from the maximum set or node coverage problems}.
Below we point out some of the significant aspects of these differences.
For comparison purposes, 
for an instance of \fmc let 
$\opt_{\mathrm{coverage}}$ 
denote the objective value of an optimal solution for the corresponding maximum $k$-set coverage problem 
for this instance by ignoring element colors and coloring constraints.
%%%%%%%%%%%%%%%%%%%%%%%%%%%%%%%%%%%%%%%%%%%%%%%%%%%%%%%%%%%%%%%%%%%%%%%%%%%%%%%%%%%%%%%%%%%%%%%%%%%%%%%
\begin{description}[leftmargin=15pt]
%%%%%%%%%%%%%%%%%%%%%%%%%%%%%%%%%%%%%%%%%%%%%%%%%%%%%%%%%%%%%%%%%%%%%%%%%%%%%%%%%%%%%%%%%%%%%%%%%%%%%%%
\item[\em Existence of a feasible solution:]
For the maximum $k$-set coverage problem, a feasible solution trivially exists for any $k$. However, 
a valid solution for 
\fmcg{\chi}{k}
may \emph{not} exist for some or all $k$ even if $\chi=2$ and in fact 
our results (Lemma~\ref{lem1}) show that even deciding if there exists a valid solution is $\NP$-complete.
The $\NP$-completeness result holds even if $f=1$ (\IE, the sets are mutually \emph{disjoint}); note that
if $f=1$ then 
it is \emph{trivial} to compute an optimal solution to the maximum $k$-set coverage problem. 
That is why for algorithmic purposes we will assume the existence of at least \emph{one} feasible 
solution\footnote{Actually, our $\LP$-relaxation based algorithms require only the existence of a feasible
\emph{fractional} solution but we cannot say anything about the approximation ratio in the absence of a 
feasible integral solution.} 
and for showing computational hardness results we will show the existence of at least one trivial 
feasible solution.
%%%%%%%%%%%%%%%%%%%%%%%%%%%%%%%%%%%%%%%%%%%%%%%%%%%%%%%%%%%%%%%%%%%%%%%%%%%%%%%%%%%%%%%%%%%%%%%%%%%%%%%
\item[\em Number of covered elements:]
The number of covered elements 
and the corresponding selected sets 
in an optimal solution in \fmc
can differ vastly from that in the 
maximum $k$-set coverage problem on the same instance.
The reason for the discrepancy is because
in \fmc one may need to select fewer covered elements to satisfy the
coloring constraints. 
%%%%%%%%%%%%%%%%%%%%%%%%%%%%%%%%%%%%%%%%%%%%%%%%%%%%%%%%%%%%%%%%%%%%%%%%%%%%%%%%%%%%%%%%%%%%%%%%%%%%%%%
\item[\em Exactly $k$ sets vs.\ at most $k$ sets:]
For the maximum $k$-set coverage problem
any solution trivially can use 
exactly $k$ sets and therefore there is \emph{no} change to the solution space whether the problem formulation 
requires exactly $k$ sets or at most $k$ sets. However, the corresponding situation for \fmc is different 
since it may be non-trivial to
convert a feasible solution containing $k'<k$ sets to one containing exactly $k$ sets because of the coloring 
constraints.
%%%%%%%%%%%%%%%%%%%%%%%%%%%%%%%%%%%%%%%%%%%%%%%%%%%%%%%%%%%%%%%%%%%%%%%%%%%%%%%%%%%%%%%%%%%%%%%%%%%%%%%
\end{description}
%%%%%%%%%%%%%%%%%%%%%%%%%%%%%%%%%%%%%%%%%%%%%%%%%%%%%%%%%%%%%%%%%%%%%%%%%%%%%%%%%%%%%%%%%%%%%%%%%%%%%%%

%%%%%%%%%%%%%%%%%%%%%%%%%%%%%%%%%%%%%%%%%%%%%%%%%%%%%%%%%%%%%%%%%%%%%%%%%%%%%%%
\smallskip
\noindent
\textbf{Discrepancy minimization problems} 
\medskip
%%%%%%%%%%%%%%%%%%%%%%%%%%%%%%%%%%%%%%%%%%%%%%%%%%%%%%%%%%%%%%%%%%%%%%%%%%%%%%%

Informally, the discrepancy minimization problem for set systems (\mindisc) is orthogonal to
unweighted \fmc.
Often \mindisc is studied in the context of \emph{two} colors, say red and blue, and is defined as follows.
Like unweighted \fmc we are given $m$ sets over $n$ elements. However, unlike \fmc element colors are \emph{not} given 
\emph{a priori} but the goal to color every element red or blue to minimize the maximum \emph{discrepancy}
over all sets, where the discrepancy of a set is the absolute difference of the number of red and blue elements 
it contains.
The Beck-Fiala theorem~\cite{BF81}
shows that the discrepancy of any set system is at most $2f$,  
Spencer showed in~\cite{S85} 
that the discrepancy of any set system is 
$O(\sqrt{n \log (2m/n)}\,)$, 
Bansal provided a randomized polynomial time algorithm achieving Spencer's bound in~\cite{B10}, and 
a deterministic algorithm with similar bounds were provided in~\cite{LRR17}.
On the lower bound side, it is possible to construct set systems such that the discrepancy is 
$\Omega(\sqrt{n}\,)$~\cite{C00}. 
For generalization of the formulation to \emph{more} than two colors and corresponding results, see
for example~\cite{D02,DS99,S03}.

%%%%%%%%%%%%%%%%%%%%%%%%%%%%%%%%%%%%%%%%%%%%%%%%%%%%%%%%%%%%%%%%%%%%%%%%%%%%%%%
\smallskip
\noindent
\added[comment={added}]{\textbf{Maximization of non-decreasing submodular set functions with linear inequality constraints}} 
\medskip

{
Kulik \EA in~\cite{KST09}
provided approximation algorithms for maximizing a non-decreasing submodular set function 
subject to multiple linear \textbf{inequality} constraints over the elements.
Unfortunately, because the linear constraints in 
\sfmc 
are \textbf{equality} constraints, \sfmc \emph{cannot} be put in the framework of~\cite{KST09}
and the approximation algorithms in~\cite{KST09}
do not directly apply to 
\sfmc. 
}

%%%%%%%%%%%%%%%%%%%%%%%%%%%%%%%%%%%%%%%%%%%%%%%%%%%%%%%%%%%%%%%%%%%%%%%%%%%%%%%%%%%%%%%%%%%%%%%%%%%%%%%
\section{\replaced{Summary of our contribution}{Summary of our contribution and proof techniques}}
%%%%%%%%%%%%%%%%%%%%%%%%%%%%%%%%%%%%%%%%%%%%%%%%%%%%%%%%%%%%%%%%%%%%%%%%%%%%%%%%%%%%%%%%%%%%%%%%%%%%%%%

%%%%%%%%%%%%%%%%%%%%%%%%%%%%%%%%%%%%%%%%%%%%%%%%%%%%%%%%%%%%%%%%%%%%%%%%%%%%%%%%%%%%%%%%%%%%%%%%%%%%%%%
\subsection{Feasibility hardness results}
%%%%%%%%%%%%%%%%%%%%%%%%%%%%%%%%%%%%%%%%%%%%%%%%%%%%%%%%%%%%%%%%%%%%%%%%%%%%%%%%%%%%%%%%%%%%%%%%%%%%%%%

Obviously \fmc (\emph{resp.}, \nfmc) obeys all the inapproximability results for the maximum 
$k$-set coverage (\emph{resp.}, $k$-node coverage) problem.
We show in Lemma~\ref{lem1} that determining feasibility of \fmc instances is $\NP$-complete 
even under very restricted parameter values; the proofs cover (or can be easily modified to cover) 
all the spacial cases of \fmc investigated in this paper. However, our subsequent algorithmic results show that 
even the existence of \emph{one} feasible solution gives rise to non-trivial approximation bounds for the 
objective and the coloring constraints.

%%%%%%%%%%%%%%%%%%%%%%%%%%%%%%%%%%%%%%%%%%%%%%%%%%%%%%%%%%%%%%%%%%%%%%%%%%%%%%%%%%%%%%%%%%%%%%%%%%%%%%%
\subsection{Algorithmic results}
%%%%%%%%%%%%%%%%%%%%%%%%%%%%%%%%%%%%%%%%%%%%%%%%%%%%%%%%%%%%%%%%%%%%%%%%%%%%%%%%%%%%%%%%%%%%%%%%%%%%%%%

A summary of our algorithmic results is shown in Table~\ref{tab-sum}.
Based on the discussion in the previous section, all of our algorithms assume that
at least one feasible solution for the \fmc instance exists.

%%%%%%%%%%%%%%%%%%%%%%%%%%%%%%%%%%%%%%%%%%%%%%%%%%%%%%%%%%%%%%%%%%%%%%%%%%%%%%%%%%%%%%%%%%%%%%%%%%%%%%%
\addtolength{\tabcolsep}{-5pt} 
%%%%%%%%%%%%%%%%%%%%%%%%%%%%%%%%%%%%%%%%%%%%%%%%%%%%%%%%%%%%%%%%%%%%%%%%%%%%%%%%%%%%%%%%%%%%%%%%%%%%%%%
\begin{table}[h]
\hspace*{-0.9in}
\scalebox{0.85}[0.85]{
\begin{tabular}{ c c c c c c c c c c}
\toprule
       & 
       &  
		   & 
			 \multicolumn{3}{|c|}{
		   {
        \begin{tabular}{c}
        coloring constraints \\ approximation 
				\end{tabular}
				  }
		     } 
			 & \multicolumn{3}{c|}{
		   {
        \begin{tabular}{c}
        parameter and other restrictions \\ (if any) 
				\end{tabular}
				  }
		     } 
       & { 
        theorem 
				 }
%%%%%%%%%%%%%%%%%%%%%%%%%%%%%%%%%%%%%%%%%%%%%%%%%%%%%%%%%%%%%%%%%%%%%%%%%%%%%%%%%%%%%%%%%%%%%%%%%%%%%%%
\\
\cmidrule{4-6}
\cmidrule{7-9}
        \begin{tabular}{c}
          problem \\ name 
				\end{tabular}
				& 
        \begin{tabular}{c}
           algorithm \\ name \\ \& type
				\end{tabular}
         & 
        \begin{tabular}{c}
        approximation \\ ratio
				\end{tabular}
				 &  %% coloring approximation 3 columns 
	    \multicolumn{1}{|c}{
			        \begin{tabular}{c}
               $\eps$-approximation \\ $\eps=$
			        \end{tabular}
										     } 
			         & 
			        \begin{tabular}{c}
			         randomized \\ $\eps$-approximation \\ $\eps=$ 
			        \end{tabular}
							 & 
					 \multicolumn{1}{c|}{
			        \begin{tabular}{c}
							 strong \\ randomized \\ $\eps$-approximation \\ $\eps=$
			        \end{tabular} 
							} &
		  %% parameter restriction 3 columns 	
       $\nopt$ & $\chi$ & \multicolumn{1}{c|}{other} & 
%%%%%%%%%%%%%%%%%%%%%%%%%%%%%%%%%%%%%%%%%%%%%%%%%%%%%%%%%%%%%%%%%%%%%%%%%%%%%%%%%%%%%%%%%%%%%%%%%%%%%%%
\\
[4pt]
\midrule
%%%%%%%%%%%%%%%%%%%%%%%%%%%%%%%%%%%%%%%%%%%%%%%%%%%%%%%%%%%%%%%%%%%%%%%%%%%%%%%%%%%%%%%%%%%%%%%%%%%%%%%
%%%%%%%%%%%%%%%%%%%%%%%%%%%%%%%%%%%%%%%%%%%%%%%%%%%%%%%%%%%%%%%%%%%%%%%%%%%%%%%%%%%%%%%%%%%%%%%%%%%%%%%
%%%%%%%%%%%%%%%%%%%%%%%%%%%%%%%%%%%%%%%%%%%%%%%%%%%%%%%%%%%%%%%%%%%%%%%%%%%%%%%%%%%%%%%%%%%%%%%%%%%%%%%
%%%%%%%%%%%%%%%%%%%%%%%%%%%%%%%%%%%%%%%%%%%%%%%%%%%%%%%%%%%%%%%%%%%%%%%%%%%%%%%%%%%%%%%%%%%%%%%%%%%%%%%
\multirow{4}{*}{
			        \begin{tabular}{c}
                   \\ \\ \\ \\ \\ \fmc
			        \end{tabular} 
									 } 
       & 
			        \begin{tabular}{c}
			            \algla \\ ${\cR\cA\mathfrak{l}\mathfrak{g}}$ 
			        \end{tabular} 
       &  
              $\varrho(f)$ 
       &  
			        $\nicefrac{\mathrm{N}}{\mathrm{A}}$
       &  
							$3.16f$
       &  
			        $O(f)$ 
       &  
							$\Omega(\chi \sqrt{n} \log\chi )\,\,\,\,$
       &  
			        ---$\,\,\,$
       &  
			        $\nicefrac{\mathrm{N}}{\mathrm{A}}$
       &  
			        Theorem~\ref{thm-main}
%%%%%%%%%%%%%%%%%%%%%%%%%%%%%%%%%%%%%%%%%%%%%%%%%%%%%%%%%%%%%%%%%%%%%%%%%%%%%%%%%%%%%%%%%%%%%%%%%%%%%%%
\\
\cmidrule{2-10}
%%%%%%%%%%%%%%%%%%%%%%%%%%%%%%%%%%%%%%%%%%%%%%%%%%%%%%%%%%%%%%%%%%%%%%%%%%%%%%%%%%%%%%%%%%%%%%%%%%%%%%%
       & 
			        \begin{tabular}{c}
			            \algmed \\ ${\cR\cA\mathfrak{l}\mathfrak{g}}$ 
			        \end{tabular} 
       &  
              $\varrho(f)$ 
       &  
			        $\nicefrac{\mathrm{N}}{\mathrm{A}}$
       &  
							$3.16f$
       &  
			        $O(f^2)$ 
       &  
			        \begin{tabular}{c}
			              $\Omega(a\chi\log\chi)$ \\ and \\ $O(\chi \sqrt{n \log\chi} )$
			        \end{tabular} 
       &  
			        ---$\,\,\,$
       &  
			        $\nicefrac{\mathrm{N}}{\mathrm{A}}$
       &  
			        Theorem~\ref{thm-main}
%%%%%%%%%%%%%%%%%%%%%%%%%%%%%%%%%%%%%%%%%%%%%%%%%%%%%%%%%%%%%%%%%%%%%%%%%%%%%%%%%%%%%%%%%%%%%%%%%%%%%%%
\\
\cmidrule{2-10}
%%%%%%%%%%%%%%%%%%%%%%%%%%%%%%%%%%%%%%%%%%%%%%%%%%%%%%%%%%%%%%%%%%%%%%%%%%%%%%%%%%%%%%%%%%%%%%%%%%%%%%%
       & 
			        \begin{tabular}{c}
			            \algsm \\ ${\cR\cA\mathfrak{l}\mathfrak{g}}$ 
			        \end{tabular} 
       &  
              $\varrho(f)$ 
       &  
			        $\nicefrac{\mathrm{N}}{\mathrm{A}}$
       &  
							$3.16f$
       &  
			        $O( f^2 \sqrt{a \,\chi \opt_\#}\,)$ 
       &  
			        ---$\,\,\,$
       &  
			        $O(\max\{1,\,\frac{\log n}{\log m}\})$ 
       &  
			        $\nicefrac{\mathrm{N}}{\mathrm{A}}$
       &  
			        Theorem~\ref{thm-main}
%%%%%%%%%%%%%%%%%%%%%%%%%%%%%%%%%%%%%%%%%%%%%%%%%%%%%%%%%%%%%%%%%%%%%%%%%%%%%%%%%%%%%%%%%%%%%%%%%%%%%%%
\\
\cmidrule{2-10}
%%%%%%%%%%%%%%%%%%%%%%%%%%%%%%%%%%%%%%%%%%%%%%%%%%%%%%%%%%%%%%%%%%%%%%%%%%%%%%%%%%%%%%%%%%%%%%%%%%%%%%%
& 
\multirow{2}{*}{
			        \begin{tabular}{c}
			            \\[-5pt] \algiter \\ ${\cD\cA\mathfrak{l}\mathfrak{g}}$ 
			        \end{tabular} 
							}
       &  
              $\nicefrac{1}{f}$ 
       &  
			        $O(f^2)$
       &  
			        $\nicefrac{\mathrm{N}}{\mathrm{A}}$
       &  
			        $\nicefrac{\mathrm{N}}{\mathrm{A}}$
       &  
			        ---$\,\,\,$
       &  
			        $O(1)$ 
       &  
              \begin{tabular}{c}
			             {at most} \\ $k+\frac{\chi-1}{2}$ sets
			        \end{tabular} $\,\,$ 
       &  
			        Theorem~\ref{thm-fmc-con}
%%%%%%%%%%%%%%%%%%%%%%%%%%%%%%%%%%%%%%%%%%%%%%%%%%%%%%%%%%%%%%%%%%%%%%%%%%%%%%%%%%%%%%%%%%%%%%%%%%%%%%%
\\
\cmidrule{3-10}
%%%%%%%%%%%%%%%%%%%%%%%%%%%%%%%%%%%%%%%%%%%%%%%%%%%%%%%%%%%%%%%%%%%%%%%%%%%%%%%%%%%%%%%%%%%%%%%%%%%%%%%
       &  
       &  
              $\nicefrac{1}{f}$ 
       &  
			        $O(f^2+\chi^2 f)$
       &  
			        $\nicefrac{\mathrm{N}}{\mathrm{A}}$
       &  
			        $\nicefrac{\mathrm{N}}{\mathrm{A}}$
       &  
			        ---$\,\,\,$
       &  
			        ---$\,\,\,$
       &  
              \begin{tabular}{c}
			             {at most} \\ $k+\chi-1$ sets
			        \end{tabular} $\,\,$ 
       &  
			        Theorem~\ref{thm-fmc-con}
%%%%%%%%%%%%%%%%%%%%%%%%%%%%%%%%%%%%%%%%%%%%%%%%%%%%%%%%%%%%%%%%%%%%%%%%%%%%%%%%%%%%%%%%%%%%%%%%%%%%%%%
%%%%%%%%%%%%%%%%%%%%%%%%%%%%%%%%%%%%%%%%%%%%%%%%%%%%%%%%%%%%%%%%%%%%%%%%%%%%%%%%%%%%%%%%%%%%%%%%%%%%%%%
\\
[4pt]
\midrule
%%%%%%%%%%%%%%%%%%%%%%%%%%%%%%%%%%%%%%%%%%%%%%%%%%%%%%%%%%%%%%%%%%%%%%%%%%%%%%%%%%%%%%%%%%%%%%%%%%%%%%%
%%%%%%%%%%%%%%%%%%%%%%%%%%%%%%%%%%%%%%%%%%%%%%%%%%%%%%%%%%%%%%%%%%%%%%%%%%%%%%%%%%%%%%%%%%%%%%%%%%%%%%%
\multirow{3}{*}{
			        \begin{tabular}{c}
                 \\[4pt] \nfmc
			        \end{tabular} 
								 }         
       & 
			 \multirow{2}{*}{
			        \begin{tabular}{c}
			            \\[-5pt] \algiter \\ ${\cD\cA\mathfrak{l}\mathfrak{g}}$ 
			        \end{tabular} 
							}
       &  
              $\nicefrac{1}{2}$ 
       &  
			        $4+4\chi$
       &  
			        $\nicefrac{\mathrm{N}}{\mathrm{A}}$
       &  
			        $\nicefrac{\mathrm{N}}{\mathrm{A}}$
       &  
			        ---$\,\,\,$
       &  
			        $O(1)$ 
       &  
              \begin{tabular}{c}
			             {at most} \\ $k+\frac{\chi-1}{2}$ sets
			        \end{tabular} $\,\,$ 
       &  
			        Theorem~\ref{thm-nfmc-con}
%%%%%%%%%%%%%%%%%%%%%%%%%%%%%%%%%%%%%%%%%%%%%%%%%%%%%%%%%%%%%%%%%%%%%%%%%%%%%%%%%%%%%%%%%%%%%%%%%%%%%%%
\\
\cmidrule{3-10}
%%%%%%%%%%%%%%%%%%%%%%%%%%%%%%%%%%%%%%%%%%%%%%%%%%%%%%%%%%%%%%%%%%%%%%%%%%%%%%%%%%%%%%%%%%%%%%%%%%%%%%%
& 
       &  
              $\nicefrac{1}{2}$ 
       &  
			        $4+2\chi+4\chi^2$
       &  
			        $\nicefrac{\mathrm{N}}{\mathrm{A}}$
       &  
			        $\nicefrac{\mathrm{N}}{\mathrm{A}}$
       &  
			        ---$\,\,\,$
       &  
			        ---$\,\,\,$
       &  
              \begin{tabular}{c}
			             {at most} \\ $k+\chi-1$ sets
			        \end{tabular} $\,\,$ 
       &  
			        Theorem~\ref{thm-nfmc-con}
%%%%%%%%%%%%%%%%%%%%%%%%%%%%%%%%%%%%%%%%%%%%%%%%%%%%%%%%%%%%%%%%%%%%%%%%%%%%%%%%%%%%%%%%%%%%%%%%%%%%%%%
\\
\cmidrule{2-10}
%%%%%%%%%%%%%%%%%%%%%%%%%%%%%%%%%%%%%%%%%%%%%%%%%%%%%%%%%%%%%%%%%%%%%%%%%%%%%%%%%%%%%%%%%%%%%%%%%%%%%%%
\multicolumn{9}{c}{Other results: same as \fmc with $f=2$}
%%%%%%%%%%%%%%%%%%%%%%%%%%%%%%%%%%%%%%%%%%%%%%%%%%%%%%%%%%%%%%%%%%%%%%%%%%%%%%%%%%%%%%%%%%%%%%%%%%%%%%%
\\
[4pt]
\midrule
%%%%%%%%%%%%%%%%%%%%%%%%%%%%%%%%%%%%%%%%%%%%%%%%%%%%%%%%%%%%%%%%%%%%%%%%%%%%%%%%%%%%%%%%%%%%%%%%%%%%%%%
%%%%%%%%%%%%%%%%%%%%%%%%%%%%%%%%%%%%%%%%%%%%%%%%%%%%%%%%%%%%%%%%%%%%%%%%%%%%%%%%%%%%%%%%%%%%%%%%%%%%%%%
\sfmc & 
              \begin{tabular}{c}
			            \alggreed \\ ${\cD\cA\mathfrak{l}\mathfrak{g}}$ 
			        \end{tabular} 
       &  
              $\varrho$ 
       &  
			        $2$
       &  
			        $\nicefrac{\mathrm{N}}{\mathrm{A}}$
       &  
			        $\nicefrac{\mathrm{N}}{\mathrm{A}}$
       &  
			        ---$\,\,\,$
       &  
			        ---$\,\,\,$
       &  
              \begin{tabular}{c}
			             {at most} \\ $k$ sets
			        \end{tabular} $\,\,$ 
       &  
			        Theorem~\ref{thm-segr-app}
%%%%%%%%%%%%%%%%%%%%%%%%%%%%%%%%%%%%%%%%%%%%%%%%%%%%%%%%%%%%%%%%%%%%%%%%%%%%%%%%%%%%%%%%%%%%%%%%%%%%%%%
\\
[4pt]
\midrule
%%%%%%%%%%%%%%%%%%%%%%%%%%%%%%%%%%%%%%%%%%%%%%%%%%%%%%%%%%%%%%%%%%%%%%%%%%%%%%%%%%%%%%%%%%%%%%%%%%%%%%%
\bfmc & 
              \begin{tabular}{c}
			            \alggreedsmpl \\ ${\cD\cA\mathfrak{l}\mathfrak{g}}$ 
			        \end{tabular} 
       &  
              $\varrho$ 
       &  
			        $O(\Delta f)$
       &  
			        $\nicefrac{\mathrm{N}}{\mathrm{A}}$
       &  
			        $\nicefrac{\mathrm{N}}{\mathrm{A}}$
       &  
			        ---$\,\,\,$
       &  
			        ---$\,\,\,$
       &  
			        $\nicefrac{\mathrm{N}}{\mathrm{A}}$
       &  
			        Proposition~\ref{prop-bala-app}
%%%%%%%%%%%%%%%%%%%%%%%%%%%%%%%%%%%%%%%%%%%%%%%%%%%%%%%%%%%%%%%%%%%%%%%%%%%%%%%%%%%%%%%%%%%%%%%%%%%%%%%
\\
[4pt]
\midrule
%%%%%%%%%%%%%%%%%%%%%%%%%%%%%%%%%%%%%%%%%%%%%%%%%%%%%%%%%%%%%%%%%%%%%%%%%%%%%%%%%%%%%%%%%%%%%%%%%%%%%%%
%%%%%%%%%%%%%%%%%%%%%%%%%%%%%%%%%%%%%%%%%%%%%%%%%%%%%%%%%%%%%%%%%%%%%%%%%%%%%%%%%%%%%%%%%%%%%%%%%%%%%%%
\gfmc & 
              \begin{tabular}{c}
			            \alggeom \\ ${\cR\cA\mathfrak{l}\mathfrak{g}}$ 
			        \end{tabular} 
       &  
              $1-O(\delta)$ 
       &  
			        $\nicefrac{\mathrm{N}}{\mathrm{A}}$
       &  
			        $\nicefrac{\mathrm{N}}{\mathrm{A}}$
       &  
			        $1+\delta$
       &  
			        ---$\,\,\,$
       &  
			        ---$\,\,\,$
       &  
              \begin{tabular}{c}
                   $d=O(1)$ 
			             \\ 
									 [3pt]
									 $\frac{\opt}{\chi}\geq \eta$
			        \end{tabular} $\,\,$ 
       &  
			        Theorem~\ref{thm-gfmc}
%%%%%%%%%%%%%%%%%%%%%%%%%%%%%%%%%%%%%%%%%%%%%%%%%%%%%%%%%%%%%%%%%%%%%%%%%%%%%%%%%%%%%%%%%%%%%%%%%%%%%%%
%%%%%%%%%%%%%%%%%%%%%%%%%%%%%%%%%%%%%%%%%%%%%%%%%%%%%%%%%%%%%%%%%%%%%%%%%%%%%%%%%%%%%%%%%%%%%%%%%%%%%%%
\\
\bottomrule
%%%%%%%%%%%%%%%%%%%%%%%%%%%%%%%%%%%%%%%%%%%%%%%%%%%%%%%%%%%%%%%%%%%%%%%%%%%%%%%%%%%%%%%%%%%%%%%%%%%%%%%
\\
%%%%%%%%%%%%%%%%%%%%%%%%%%%%%%%%%%%%%%%%%%%%%%%%%%%%%%%%%%%%%%%%%%%%%%%%%%%%%%%%%%%%%%%%%%%%%%%%%%%%%%%
\multicolumn{10}{c}{
		 \hspace*{-1.6in}
     ${\cD\cA\mathfrak{l}\mathfrak{g}}$: deterministic algorithm 
		 \hspace*{0.3in}
     ${\cR\cA\mathfrak{l}\mathfrak{g}}$: randomized algorithm 
}
%%%%%%%%%%%%%%%%%%%%%%%%%%%%%%%%%%%%%%%%%%%%%%%%%%%%%%%%%%%%%%%%%%%%%%%%%%%%%%%%%%%%%%%%%%%%%%%%%%%%%%%
\\
[2pt]
%%%%%%%%%%%%%%%%%%%%%%%%%%%%%%%%%%%%%%%%%%%%%%%%%%%%%%%%%%%%%%%%%%%%%%%%%%%%%%%%%%%%%%%%%%%%%%%%%%%%%%%
\multicolumn{10}{c}{
		 \hspace*{-1.6in}
		 $\nicefrac{\mathrm{N}}{\mathrm{A}}$: not applicable
		 \hspace*{0.3in}
		 ---: unrestricted
   }
%%%%%%%%%%%%%%%%%%%%%%%%%%%%%%%%%%%%%%%%%%%%%%%%%%%%%%%%%%%%%%%%%%%%%%%%%%%%%%%%%%%%%%%%%%%%%%%%%%%%%%%
\\
[2pt]
%%%%%%%%%%%%%%%%%%%%%%%%%%%%%%%%%%%%%%%%%%%%%%%%%%%%%%%%%%%%%%%%%%%%%%%%%%%%%%%%%%%%%%%%%%%%%%%%%%%%%%%
\multicolumn{10}{c}{
		 \hspace*{-1.3in}
     ${\varrho(x) = {\big( 1 - \nicefrac{1}{x} \big)}^x}>1-\nicefrac{1}{\bee}$
		 \hspace*{0.3in}
		 $\varrho=\max\{ \varrho(f), \, \varrho(k) \}>1-\nicefrac{1}{\bee}$
}
%%%%%%%%%%%%%%%%%%%%%%%%%%%%%%%%%%%%%%%%%%%%%%%%%%%%%%%%%%%%%%%%%%%%%%%%%%%%%%%%%%%%%%%%%%%%%%%%%%%%%%%
\\
[2pt]
%%%%%%%%%%%%%%%%%%%%%%%%%%%%%%%%%%%%%%%%%%%%%%%%%%%%%%%%%%%%%%%%%%%%%%%%%%%%%%%%%%%%%%%%%%%%%%%%%%%%%%%
\multicolumn{10}{c}{
		 \hspace*{-1.3in}
		 $\delta$: any constant in the range (0,1] 
		 \hspace*{0.3in}
		 $\eta$: any constant strictly greater than $1$
   }
%%%%%%%%%%%%%%%%%%%%%%%%%%%%%%%%%%%%%%%%%%%%%%%%%%%%%%%%%%%%%%%%%%%%%%%%%%%%%%%%%%%%%%%%%%%%%%%%%%%%%%%
%%%%%%%%%%%%%%%%%%%%%%%%%%%%%%%%%%%%%%%%%%%%%%%%%%%%%%%%%%%%%%%%%%%%%%%%%%%%%%%%%%%%%%%%%%%%%%%%%%%%%%%
\end{tabular}
}
\caption{\label{tab-sum}A summary of our algorithmic results. The $O(\cdot)$ notation is used when constants 
are irrelevant or not precisely calculated.
Precise definitions of 
(deterministic) $\eps$-approximation, 
randomized $\eps$-approximation and 
strong randomized $\eps$-approximation
appear in Section~\ref{sec-relax}.
} 
\end{table}
%%%%%%%%%%%%%%%%%%%%%%%%%%%%%%%%%%%%%%%%%%%%%%%%%%%%%%%%%%%%%%%%%%%%%%%%%%%%%%%%%%%%%%%%%%%%%%%%%%%%%%%
\addtolength{\tabcolsep}{7pt} 
%%%%%%%%%%%%%%%%%%%%%%%%%%%%%%%%%%%%%%%%%%%%%%%%%%%%%%%%%%%%%%%%%%%%%%%%%%%%%%%%%%%%%%%%%%%%%%%%%%%%%%%

%%\DeclareRobustCommand{\hsout}[1]{\texorpdfstring{\sout{#1}}{#1}} 

%%%%%%%%%%%%%%%%%%%%%%%%%%%%%%%%%%%%%%%%%%%%%%%%%%%%%%%%%%%%%%%%%%%%%%%%%%%%%%%%%%%%%%%%%%%%%%%%%%%%%%%
\subsection{{{Remarks on proof techniques}}}
%%\hsout{\subsection{Remarks on proof techniques}}
%%%%%%%%%%%%%%%%%%%%%%%%%%%%%%%%%%%%%%%%%%%%%%%%%%%%%%%%%%%%%%%%%%%%%%%%%%%%%%%%%%%%%%%%%%%%%%%%%%%%%%%

{
{
%%%%%%%%%%%%%%%%%%%%%%%%%%%%%%%%%%%%%%%%%%%%%%%%%%%%%%%%%%%%%%%%%%%%%%%%%%%%%%%%%%%%%%%%%%%%%%%%%%%%%%%
\smallskip
\noindent
\textbf{Distributions on level sets with negative correlations}
}
}
\smallskip
%%%%%%%%%%%%%%%%%%%%%%%%%%%%%%%%%%%%%%%%%%%%%%%%%%%%%%%%%%%%%%%%%%%%%%%%%%%%%%%%%%%%%%%%%%%%%%%%%%%%%%%

{
{
Our randomized algorithms use the sampling result by 
Srinivasan~\cite{S01} which allows one to sample variables satisfying an equality \emph{precisely} while 
still ensuring that the variables are \emph{negatively} correlated and therefore the tail bounds 
by Panconesi and Srinivasan~\cite{PS97}
can be applied. 
This allows the randomized algorithms in 
Theorem~\ref{thm-main}
to select precisely $k$ sets while still preserving the properties of the distribution
of variables that are needed for the proof. 
}
}

%%%%%%%%%%%%%%%%%%%%%%%%%%%%%%%%%%%%%%%%%%%%%%%%%%%%%%%%%%%%%%%%%%%%%%%%%%%%%%%%%%%%%%%%%%%%%%%%%%%%%%%
\medskip
\noindent
{
{
\textbf{Strengthening $\LP$-relaxation via additional inequalities}
}
}
\smallskip
%%%%%%%%%%%%%%%%%%%%%%%%%%%%%%%%%%%%%%%%%%%%%%%%%%%%%%%%%%%%%%%%%%%%%%%%%%%%%%%%%%%%%%%%%%%%%%%%%%%%%%%

{
{
As we show in Section~\ref{sec-str}, 
a straightforward $\LP$-relaxation of \fmc based on a corresponding known $\LP$-relaxation of maximum 
$k$-set coverage problems does \emph{not} have an finite integrality gap and therefore unsuitable for further 
analysis. To get around this, we use an approach similar to what was used by 
prior researchers (\EG, see the works by Carnes and Shmoys~\cite{CD08} and Carr~\EA~\cite{CFLP00})
by introducing extra $O(fn)$ \emph{covering inequalities}
which brings down the integrality gap 
and allows the results in Theorem~\ref{thm-main} to go through. 
}
}

{
{
Moreover, we had to separately modify existing constraints of or add new constraints to the basic $\LP$-relaxation
for the three algorithms, namely algorithms 
\algsm, \algmed and \algiter.
Modifications for 
\algsm and \algmed
in Theorem~\ref{thm-main}
are done to 
encode the coloring constraints suitably to 
optimize their coloring constraint approximation bounds for the corresponding parameter ranges.
The modifications for \algiter in Theorem~\ref{thm-nfmc-con} and Theorem~\ref{thm-fmc-con}
are necessary for the iterated rounding approach to go through.
}
}

%%%%%%%%%%%%%%%%%%%%%%%%%%%%%%%%%%%%%%%%%%%%%%%%%%%%%%%%%%%%%%%%%%%%%%%%%%%%%%%%%%%%%%%%%%%%%%%%%%%%%%%
\medskip
\noindent
{
{
\textbf{Doob martingales and Azuma's inequality}
}
}
\smallskip
%%%%%%%%%%%%%%%%%%%%%%%%%%%%%%%%%%%%%%%%%%%%%%%%%%%%%%%%%%%%%%%%%%%%%%%%%%%%%%%%%%%%%%%%%%%%%%%%%%%%%%%

{
{
The analysis of the rounding step of our various $\LP$-relaxations are further complicated by the fact 
that the random element-selection variables
may \emph{not} be pairwise independent;  
in fact, it is easy to construct examples in which 
each element-selection variable may be correlated to about $af$ other element-selection variables, 
thereby ruling out straightforward use of Chernoff-type tail bounds.
For sufficient large $\opt_\#$, 
we remedy this situation by using Doob martinagales and Azuma's inequality in the analysis of 
\algla in Theorem~\ref{thm-main}.
}
}

%%%%%%%%%%%%%%%%%%%%%%%%%%%%%%%%%%%%%%%%%%%%%%%%%%%%%%%%%%%%%%%%%%%%%%%%%%%%%%%%%%%%%%%%%%%%%%%%%%%%%%%
\medskip
\noindent
{
{
\textbf{Iterated rounding of $\LP$-relaxation}
}
}
\smallskip
%%%%%%%%%%%%%%%%%%%%%%%%%%%%%%%%%%%%%%%%%%%%%%%%%%%%%%%%%%%%%%%%%%%%%%%%%%%%%%%%%%%%%%%%%%%%%%%%%%%%%%%

{
{
The analysis of our deterministic algorithm 
\algiter uses the iterated rounding approach originally introduced by Jain in~\cite{J01} and subsequently used
by many researchers (the book by Lau, Ravi and Singh~\cite{LRS11} provides an excellent overview of the topic). 
A crucial ingredient of this technique used in our proof is the rank lemma.
In order to use this technique, we had to modify the $\LP$-relaxation again.
When $\chi=O(1)$, we can do two exhaustive enumeration steps in polynomial time, giving rise 
to a somewhat better approximation of the coloring constraints. 
}
}

%%%%%%%%%%%%%%%%%%%%%%%%%%%%%%%%%%%%%%%%%%%%%%%%%%%%%%%%%%%%%%%%%%%%%%%%%%%%%%%%%%%%%%%%%%%%%%%%%%%%%%%
\medskip
\noindent
{
{
\textbf{Random shifting technique}
}
}
\smallskip
%%%%%%%%%%%%%%%%%%%%%%%%%%%%%%%%%%%%%%%%%%%%%%%%%%%%%%%%%%%%%%%%%%%%%%%%%%%%%%%%%%%%%%%%%%%%%%%%%%%%%%%

\medskip
{
{
The analysis of our deterministic algorithm \alggeom 
uses the random shifting technique that has been used by prior researchers such as~\cite{A98,M99}.
}
}

%%%%%%%%%%%%%%%%%%%%%%%%%%%%%%%%%%%%%%%%%%%%%%%%%%%%%%%%%%%%%%%%%%%%%%%%%%%%%%%%%%%%%%%%%%%%%%%%%%%%%%%
\section{Organization of the paper and proof structures}
%%%%%%%%%%%%%%%%%%%%%%%%%%%%%%%%%%%%%%%%%%%%%%%%%%%%%%%%%%%%%%%%%%%%%%%%%%%%%%%%%%%%%%%%%%%%%%%%%%%%%%%

The rest of the paper is organized as follows.
%%%%%%%%%%%%%%%%%%%%%%%%%%%%%%%%%%%%%%%%%%%%%%%%%%%%%%%%%%%%%%%%%%%%%%%%%%%%%%%%%%%%%%%%%%%%%%%%%%%%%%%
\begin{enumerate}[label=$\triangleright$,leftmargin=*]
%%%%%%%%%%%%%%%%%%%%%%%%%%%%%%%%%%%%%%%%%%%%%%%%%%%%%%%%%%%%%%%%%%%%%%%%%%%%%%%%%%%%%%%%%%%%%%%%%%%%%%%
\item 
In Section~\ref{sec-hard-feas} we present our result
in Lemma~\ref{lem1}
on the computational hardness of finding a feasible solution of \fmc.
%%%%%%%%%%%%%%%%%%%%%%%%%%%%%%%%%%%%%%%%%%%%%%%%%%%%%%%%%%%%%%%%%%%%%%%%%%%%%%%%%%%%%%%%%%%%%%%%%%%%%%%
\item 
Based on the results in Section~\ref{sec-hard-feas}, we need to make some minimal assumptions and 
need to consider appropriate approximate variants of the coloring constraints. They are discussed in 
Section~\ref{sec-relax}
for the purpose of designing (deterministic or randomized) approximation algorithms.
%%%%%%%%%%%%%%%%%%%%%%%%%%%%%%%%%%%%%%%%%%%%%%%%%%%%%%%%%%%%%%%%%%%%%%%%%%%%%%%%%%%%%%%%%%%%%%%%%%%%%%%
\item 
In Section~\ref{sec-lp-rand}
we design and analyze our $\LP$-relaxation based randomized approximation algorithms for \fmc.
In particular, in Theorem~\ref{thm-main}
we employ two different $\LP$-relaxation of \fmc 
and combine three randomized rounding analysis on them 
to get an approximation algorithm whose approximation qualities depend on the range of relevant 
parameters.
%%%%%%%%%%%%%%%%%%%%%%%%%%%%%%%%%%%%%%%%%%%%%%%%%%%%%%%%%%%%%%%%%%%%%%%%%%%%%%%%%%%%%%%%%%%%%%%%%%%%%%%
\begin{enumerate}[label=$\triangleright$]
%%%%%%%%%%%%%%%%%%%%%%%%%%%%%%%%%%%%%%%%%%%%%%%%%%%%%%%%%%%%%%%%%%%%%%%%%%%%%%%%%%%%%%%%%%%%%%%%%%%%%%%
\item
Parts of the algorithm and analysis specific to the three algorithms 
\algla, \algmed and \algsm
are discussed in 
Section~\ref{sec-algla1}, 
Section~\ref{sec-algmed1} and 
Section~\ref{sec-algsm1}, respectively.
%%%%%%%%%%%%%%%%%%%%%%%%%%%%%%%%%%%%%%%%%%%%%%%%%%%%%%%%%%%%%%%%%%%%%%%%%%%%%%%%%%%%%%%%%%%%%%%%%%%%%%%
\item
Proposition~\ref{prop1} in Section~\ref{sec-limit-lp}
shows that 
the dependence of the coloring constraint bounds in Theorem~\ref{thm-main}(\emph{e})(\emph{i})--(\emph{ii})
on $f$ cannot be completely eliminated by better analysis of our $\LP$-relaxations
even for $\chi=2$. 
%%%%%%%%%%%%%%%%%%%%%%%%%%%%%%%%%%%%%%%%%%%%%%%%%%%%%%%%%%%%%%%%%%%%%%%%%%%%%%%%%%%%%%%%%%%%%%%%%%%%%%%
\end{enumerate}
%%%%%%%%%%%%%%%%%%%%%%%%%%%%%%%%%%%%%%%%%%%%%%%%%%%%%%%%%%%%%%%%%%%%%%%%%%%%%%%%%%%%%%%%%%%%%%%%%%%%%%%
\item 
In Section~\ref{sec-nfmc-det} 
we provide polynomial-time deterministic approximations of \fmc via iterated rounding
of a new $\LP$-relaxation. 
Our approximation qualities depend on the parameters $f$ and $\chi$.
%%%%%%%%%%%%%%%%%%%%%%%%%%%%%%%%%%%%%%%%%%%%%%%%%%%%%%%%%%%%%%%%%%%%%%%%%%%%%%%%%%%%%%%%%%%%%%%%%%%%%%%
\begin{enumerate}[label=$\triangleright$]
%%%%%%%%%%%%%%%%%%%%%%%%%%%%%%%%%%%%%%%%%%%%%%%%%%%%%%%%%%%%%%%%%%%%%%%%%%%%%%%%%%%%%%%%%%%%%%%%%%%%%%%
\item
For better understanding, we first prove our result for the special case 
\nfmc of \fmc 
in Theorem~\ref{thm-nfmc-con} (Section~\ref{sec-nfmc-con})
and later on describe how to adopt the same approach for \fmc
in Theorem~\ref{thm-fmc-con} (Section~\ref{sec-fmc-con}).
%%%%%%%%%%%%%%%%%%%%%%%%%%%%%%%%%%%%%%%%%%%%%%%%%%%%%%%%%%%%%%%%%%%%%%%%%%%%%%%%%%%%%%%%%%%%%%%%%%%%%%%
\item
The proofs for both 
Theorem~\ref{thm-nfmc-con} and
Theorem~\ref{thm-fmc-con} 
are themselves divided into two parts depending on whether $\chi=O(1)$ or not.
%%%%%%%%%%%%%%%%%%%%%%%%%%%%%%%%%%%%%%%%%%%%%%%%%%%%%%%%%%%%%%%%%%%%%%%%%%%%%%%%%%%%%%%%%%%%%%%%%%%%%%%
\end{enumerate}
%%%%%%%%%%%%%%%%%%%%%%%%%%%%%%%%%%%%%%%%%%%%%%%%%%%%%%%%%%%%%%%%%%%%%%%%%%%%%%%%%%%%%%%%%%%%%%%%%%%%%%%
\item 
In Section~\ref{sec-special-fmc} 
we provide deterministic 
approximation algorithms for two special cases of \fmc, namely 
\sfmc and \bfmc.
%%%%%%%%%%%%%%%%%%%%%%%%%%%%%%%%%%%%%%%%%%%%%%%%%%%%%%%%%%%%%%%%%%%%%%%%%%%%%%%%%%%%%%%%%%%%%%%%%%%%%%%
\item 
Section~\ref{sec-gfmc}
provides the deterministic 
approximation for \gfmc.
%%%%%%%%%%%%%%%%%%%%%%%%%%%%%%%%%%%%%%%%%%%%%%%%%%%%%%%%%%%%%%%%%%%%%%%%%%%%%%%%%%%%%%%%%%%%%%%%%%%%%%%
\end{enumerate}
%%%%%%%%%%%%%%%%%%%%%%%%%%%%%%%%%%%%%%%%%%%%%%%%%%%%%%%%%%%%%%%%%%%%%%%%%%%%%%%%%%%%%%%%%%%%%%%%%%%%%%%
Our proofs are structured as follows. A complex proof is divided into subsections corresponding to 
logical sub-divisions of the proofs and the algorithms therein. Often we provide some informal intuitions
behind the proofs (including some intuition about why other approaches may not work, if appropriate) 
before describing the actual proofs.

%%%%%%%%%%%%%%%%%%%%%%%%%%%%%%%%%%%%%%%%%%%%%%%%%%%%%%%%%%%%%%%%%%%%%%%%%%%%%%%%%%%%%%%%%%%%%%%%%%%%%%%
\section{Computational hardness of finding a feasible solution of \fmc}
\label{sec-hard-feas}
%%%%%%%%%%%%%%%%%%%%%%%%%%%%%%%%%%%%%%%%%%%%%%%%%%%%%%%%%%%%%%%%%%%%%%%%%%%%%%%%%%%%%%%%%%%%%%%%%%%%%%%

We show that determining if a given instance of 
\fmc
has even one feasible solution is $\NP$-complete even in very restricted parameter settings.
The relevant parameters of importance for \fmc is $a$, $f$ and $\chi$; Lemma~\ref{lem1} shows that 
the $\NP$-completeness result holds even for very small values of these parameters.

%%%%%%%%%%%%%%%%%%%%%%%%%%%%%%%%%%%%%%%%%%%%%%%%%%%%%%%%%%%%%%%%%%%%%%%%%%%%%%%%%%%%%%%%%%%%%%%%%%%%%%%
\begin{lemma}\label{lem1}
Determining feasibility of 
an instance of 
\fmc 
of $n$ elements is $\NP$-complete
even with the following restrictions:
%%%%%%%%%%%%%%%%%%%%%%%%%%%%%%%%%%%%%%%%%%%%%%%%%%%%%%%%%%%%%%%%%%%%%%%%%%%%%%%%%%%%%%%%%%%%%%%%%%%%%%%
\begin{enumerate}[label=$\triangleright$]
%%%%%%%%%%%%%%%%%%%%%%%%%%%%%%%%%%%%%%%%%%%%%%%%%%%%%%%%%%%%%%%%%%%%%%%%%%%%%%%%%%%%%%%%%%%%%%%%%%%%%%%
\item 
the instances correspond to the unweighted version, 
%%%%%%%%%%%%%%%%%%%%%%%%%%%%%%%%%%%%%%%%%%%%%%%%%%%%%%%%%%%%%%%%%%%%%%%%%%%%%%%%%%%%%%%%%%%%%%%%%%%%%%%
\item 
the following combinations of maximum set-size $a$, frequency $f$ and number of colors $\chi$ are satisfied:
%%%%%%%%%%%%%%%%%%%%%%%%%%%%%%%%%%%%%%%%%%%%%%%%%%%%%%%%%%%%%%%%%%%%%%%%%%%%%%%%%%%%%%%%%%%%%%%%%%%%%%%
\begin{enumerate}[label=\emph{(}\alph*\emph{)}]
%%%%%%%%%%%%%%%%%%%%%%%%%%%%%%%%%%%%%%%%%%%%%%%%%%%%%%%%%%%%%%%%%%%%%%%%%%%%%%%%%%%%%%%%%%%%%%%%%%%%%%%
\item
$f\in \{1,3\}$, all but one set contains exactly $3$ elements and all $\chi\geq 2$,
%%%%%%%%%%%%%%%%%%%%%%%%%%%%%%%%%%%%%%%%%%%%%%%%%%%%%%%%%%%%%%%%%%%%%%%%%%%%%%%%%%%%%%%%%%%%%%%%%%%%%%%
\item
the instances correspond to \nfmc (which implies $f=2$), 
$a=O(\sqrt{n}\,)$, and all $\chi\geq 2$, or 
%%%%%%%%%%%%%%%%%%%%%%%%%%%%%%%%%%%%%%%%%%%%%%%%%%%%%%%%%%%%%%%%%%%%%%%%%%%%%%%%%%%%%%%%%%%%%%%%%%%%%%%
\item
$f=1$, $a=3$ and $\chi=\nicefrac{n}{3}$.
%%%%%%%%%%%%%%%%%%%%%%%%%%%%%%%%%%%%%%%%%%%%%%%%%%%%%%%%%%%%%%%%%%%%%%%%%%%%%%%%%%%%%%%%%%%%%%%%%%%%%%%
\end{enumerate}
%%%%%%%%%%%%%%%%%%%%%%%%%%%%%%%%%%%%%%%%%%%%%%%%%%%%%%%%%%%%%%%%%%%%%%%%%%%%%%%%%%%%%%%%%%%%%%%%%%%%%%%
%
%%%%%%%%%%%%%%%%%%%%%%%%%%%%%%%%%%%%%%%%%%%%%%%%%%%%%%%%%%%%%%%%%%%%%%%%%%%%%%%%%%%%%%%%%%%%%%%%%%%%%%%
\end{enumerate}
%%%%%%%%%%%%%%%%%%%%%%%%%%%%%%%%%%%%%%%%%%%%%%%%%%%%%%%%%%%%%%%%%%%%%%%%%%%%%%%%%%%%%%%%%%%%%%%%%%%%%%%
Moreover, the following assertions also hold: 
%%%%%%%%%%%%%%%%%%%%%%%%%%%%%%%%%%%%%%%%%%%%%%%%%%%%%%%%%%%%%%%%%%%%%%%%%%%%%%%%%%%%%%%%%%%%%%%%%%%%%%%
\begin{enumerate}[label=$\triangleright$]
%%%%%%%%%%%%%%%%%%%%%%%%%%%%%%%%%%%%%%%%%%%%%%%%%%%%%%%%%%%%%%%%%%%%%%%%%%%%%%%%%%%%%%%%%%%%%%%%%%%%%%%
\item 
The instances of \fmc generated in 
\emph{(\emph{a})}
and 
\emph{(\emph{b})}
actually are instances of 
\sfmc.
%%%%%%%%%%%%%%%%%%%%%%%%%%%%%%%%%%%%%%%%%%%%%%%%%%%%%%%%%%%%%%%%%%%%%%%%%%%%%%%%%%%%%%%%%%%%%%%%%%%%%%%
\item 
For the instances of \fmc generated in 
\emph{(\emph{c})},
$\opt_\#=\chi=\nicefrac{n}{3}$
and, assuming $P\neq\NP$,
there is no polynomial time approximation algorithm that has either a finite approximation ratio 
or satisfies the coloring constraints 
in the $\eps$-approximate sense $($cf.\ eq.\ \eqref{eq1}$)$
for any finite $\eps$. 
%%%%%%%%%%%%%%%%%%%%%%%%%%%%%%%%%%%%%%%%%%%%%%%%%%%%%%%%%%%%%%%%%%%%%%%%%%%%%%%%%%%%%%%%%%%%%%%%%%%%%%%
\end{enumerate}
%%%%%%%%%%%%%%%%%%%%%%%%%%%%%%%%%%%%%%%%%%%%%%%%%%%%%%%%%%%%%%%%%%%%%%%%%%%%%%%%%%%%%%%%%%%%%%%%%%%%%%%
\end{lemma}
%%%%%%%%%%%%%%%%%%%%%%%%%%%%%%%%%%%%%%%%%%%%%%%%%%%%%%%%%%%%%%%%%%%%%%%%%%%%%%%%%%%%%%%%%%%%%%%%%%%%%%%

A proof of Lemma~\ref{lem1} appears in the appendix.

\begin{remark}
It may be tempting to conclude that an approach similar to what is stated below 
using the 
$k$-set coverage 
problem as a ``black box''
may make the claims in Lemma~\ref{lem1}
completely obvious.
We simply take any hard instance of $k$-set coverage and
equi-partition the universe arbitrarily into $\chi$ many color classes and let this be the corresponding instance 
of \fmc.
Using a suitable standard reductions of $\NP$-hardness 
the $k$-set coverage problem,
if there is a feasible solution 
of \fmc
then the $k$ sets trivially cover the entire universe and thus trivially satisfy 
the color constraints but otherwise one may (incorrectly) claim 
that no $k$ sets cover the universe and so the fairness constraints cannot be satisfied.
Additionally, one may be tempted to argue that if one takes a 
$k$-set coverage problem
instance with any additional structure (\EG, bounded occurrence of universe elements) then the coloring does not affect this 
additional property at all and hence the property is retained in the 
$k$-set coverage problem
instance with color constraints.

However, such a generic reduction will fail because 
it is incorrect and because it will not 
capture all the special parameter restrictions imposed in Lemma~\ref{lem1}.
For example:
%%%%%%%%%%%%%%%%%%%%%%%%%%%%%%%%%%%%%%%%%%%%%%%%%%%%%%%%%%%%%%%%%%%%%%%%%%%%%%%%%%%%%%%%%%%%%%%%%%%%%%%
\begin{enumerate}[label=$\triangleright$]
%%%%%%%%%%%%%%%%%%%%%%%%%%%%%%%%%%%%%%%%%%%%%%%%%%%%%%%%%%%%%%%%%%%%%%%%%%%%%%%%%%%%%%%%%%%%%%%%%%%%%%%
\item 
Even though the $k$ sets may not cover the entire universe, it is still possible that 
they may satisfy the color constraints. For example, 
consider the following instance of 
the $k$-set coverage problem: 
%%%%%%%%%%%%%%%%%%%%%%%%%%%%%%%%%%%%%%%%%%%%%%%%%%%%%%%%%%%%%%%%%%%%%%%%%%%%%%%%%%%%%%%%%%%%%%%%%%%%%%%
\begin{gather*}
\cU=\{u_1,u_2,u_3,u_4,u_5,u_6\},
\,
k=2,
\,
\cS_1 =\{u_1,u_2\},\,
\cS_2 =\{u_3,u_4\},\,
\cS_3 =\{u_5,u_6\}
\end{gather*}
%%%%%%%%%%%%%%%%%%%%%%%%%%%%%%%%%%%%%%%%%%%%%%%%%%%%%%%%%%%%%%%%%%%%%%%%%%%%%%%%%%%%%%%%%%%%%%%%%%%%%%%
Suppose we select the equi-partition 
$\{u_1,u_3,u_5\},\{u_2,u_4,u_6\}$, thus setting 
$
\cC(u_1)=
\cC(u_3)=
\cC(u_5)= 1
$
and
$
\cC(u_2)=
\cC(u_4)=
\cC(u_6)= 2
$. Then any two selected sets will satisfy the coloring constraints.
%%%%%%%%%%%%%%%%%%%%%%%%%%%%%%%%%%%%%%%%%%%%%%%%%%%%%%%%%%%%%%%%%%%%%%%%%%%%%%%%%%%%%%%%%%%%%%%%%%%%%%%
\item 
Consider the requirements of the \fmc\ instances in part 
{(\emph{c})}.
Since $f=1$, every element occurs in exactly one set and thus the set systems for the 
$k$-set coverage problem form a partition of the universe. Such an instance cannot be a hard instance 
of 
$k$-set coverage since it admits a trivial polynomial time solution: 
sort the sets in non-decreasing order of their cardinalities and simply take the first $k$ sets.
%%%%%%%%%%%%%%%%%%%%%%%%%%%%%%%%%%%%%%%%%%%%%%%%%%%%%%%%%%%%%%%%%%%%%%%%%%%%%%%%%%%%%%%%%%%%%%%%%%%%%%%
\end{enumerate}
%%%%%%%%%%%%%%%%%%%%%%%%%%%%%%%%%%%%%%%%%%%%%%%%%%%%%%%%%%%%%%%%%%%%%%%%%%%%%%%%%%%%%%%%%%%%%%%%%%%%%%%
\end{remark}

%%%%%%%%%%%%%%%%%%%%%%%%%%%%%%%%%%%%%%%%%%%%%%%%%%%%%%%%%%%%%%%%%%%%%%%%%%%%%%%%%%%%%%%%%%%%%%%%%%%%%%%
\section{Relaxing coloring constraints for algorithmic designs}
\label{sec-relax}
%%%%%%%%%%%%%%%%%%%%%%%%%%%%%%%%%%%%%%%%%%%%%%%%%%%%%%%%%%%%%%%%%%%%%%%%%%%%%%%%%%%%%%%%%%%%%%%%%%%%%%%

Based on Lemma~\ref{lem1} we need to make the following \emph{minimal} assumptions for the purpose of designing 
approximation algorithms with finite approximation ratios:
%%%%%%%%%%%%%%%%%%%%%%%%%%%%%%%%%%%%%%%%%%%%%%%%%%%%%%%%%%%%%%%%%%%%%%%%%%%%%%%%%%%%%%%%%%%%%%%%%%%%%%%
\begin{enumerate}[label=(\emph{\roman*})]
%%%%%%%%%%%%%%%%%%%%%%%%%%%%%%%%%%%%%%%%%%%%%%%%%%%%%%%%%%%%%%%%%%%%%%%%%%%%%%%%%%%%%%%%%%%%%%%%%%%%%%%
\item 
We assume the existence of at least one feasible solution for the given instance of \fmc.
%%%%%%%%%%%%%%%%%%%%%%%%%%%%%%%%%%%%%%%%%%%%%%%%%%%%%%%%%%%%%%%%%%%%%%%%%%%%%%%%%%%%%%%%%%%%%%%%%%%%%%%
\item 
We assume that $\opt_\#$ is sufficiently large compared to $\chi$, \EG, $\opt_\#\geq c\chi$ for some large constant $c>1$.
%%%%%%%%%%%%%%%%%%%%%%%%%%%%%%%%%%%%%%%%%%%%%%%%%%%%%%%%%%%%%%%%%%%%%%%%%%%%%%%%%%%%%%%%%%%%%%%%%%%%%%%
\end{enumerate}
%%%%%%%%%%%%%%%%%%%%%%%%%%%%%%%%%%%%%%%%%%%%%%%%%%%%%%%%%%%%%%%%%%%%%%%%%%%%%%%%%%%%%%%%%%%%%%%%%%%%%%%
%
Lemma~\ref{lem1} and the example in \FI{fig1} also show that 
satisfying the color constraint \emph{exactly} (\IE, requiring $\nicefrac{p_i}{p_j}$ to be exactly equal to 
$1$ for all $i$ and $j$)
need to be relaxed for the purpose of designing efficient algorithms since 
non-exact solutions of \fmc may not satisfy these constraints exactly. 
%%%%
We define an (\emph{deterministic}) ${\eps}$-\emph{approximate coloring} of \fmc
(for some $\eps\geq 1$) to be a coloring that satisfies the coloring constraints in the following manner: 
%%%%%%%%%%%%%%%%%%%%%%%%%%%%%%%%%%%%%%%%%%%%%%%%%%%%%%%%%%%%%%%%%%%%%%%%%%%%%%%%%%%%%%%%%%%%%%%%%%%%%%%
\begin{gather}
\text{\bf Deterministic ${\eps}$-approximate coloring}:
\,\,\,\,
\forall \, i,j\in\{1,\dots,\chi\} \,:\,
p_i \leq \eps p_j
\label{eq1}
\end{gather}
%%%%%%%%%%%%%%%%%%%%%%%%%%%%%%%%%%%%%%%%%%%%%%%%%%%%%%%%%%%%%%%%%%%%%%%%%%%%%%%%%%%%%%%%%%%%%%%%%%%%%%%
Note that~\eqref{eq1} automatically implies that 
$p_i \geq  p_j/\eps$ for all $i$ and $j$.
Thus, in our terminology, a $1$-approximate 
coloring corresponds to satisfying the coloring constraints exactly.
Finally, if our algorithm is randomized, then the $p_j$'s could be a random values, and 
then we will assume that the relevant constraints will be satisfied 
\emph{in expectation} or \emph{with high probability} 
in an appropriate sense. More precisely, \eqref{eq1} will be modified 
as follows:\footnote{We do not provide a bound on $\Ave{\nicefrac{p_i}{p_j}}$ since
$\nicefrac{p_i}{p_j}=\infty$ when $p_j=0$ and $p_j$ may be zero with a strictly positive probability, and 
for arbitrary $\chi$ selecting a set individually for each to avoid this situation in our 
randomized algorithms may select too many sets.}
%%%%%%%%%%%%%%%%%%%%%%%%%%%%%%%%%%%%%%%%%%%%%%%%%%%%%%%%%%%%%%%%%%%%%%%%%%%%%%%%%%%%%%%%%%%%%%%%%%%%%%%
\begin{description}[leftmargin=15pt,labelindent=10pt]
%%%%%%%%%%%%%%%%%%%%%%%%%%%%%%%%%%%%%%%%%%%%%%%%%%%%%%%%%%%%%%%%%%%%%%%%%%%%%%%%%%%%%%%%%%%%%%%%%%%%%%%
\item[\bf Randomized ${\eps}$-approximate coloring:]$\,$
%%%%%%%%%%%%%%%%%%%%%%%%%%%%%%%%%%%%%%%%%%%%%%%%%%%%%%%%%%%%%%%%%%%%%%%%%%%%%%%%%%%%%%%%%%%%%%%%%%%%%%%

\hspace*{1in}
$
\forall \, i,j\in\{1,\dots,\chi\} \,:\,
\Ave{p_i} \leq \eps \, \Ave{p_j}
$
\hfill {\eqref{eq1}$'$}
%%%%%%%%%%%%%%%%%%%%%%%%%%%%%%%%%%%%%%%%%%%%%%%%%%%%%%%%%%%%%%%%%%%%%%%%%%%%%%%%%%%%%%%%%%%%%%%%%%%%%%%
\smallskip
\item[\bf Randomized strong ${\eps}$-approximate coloring:]$\,$
%%%%%%%%%%%%%%%%%%%%%%%%%%%%%%%%%%%%%%%%%%%%%%%%%%%%%%%%%%%%%%%%%%%%%%%%%%%%%%%%%%%%%%%%%%%%%%%%%%%%%%%

\hspace*{1in}
$
\bigwedge_{ i,j\in\{1,\dots,\chi\} } 
\Big( \prob { p_i \leq \eps p_j } \Big) \geq 1 - o(1)
$
\hfill {\eqref{eq1}$''$}
%%%%%%%%%%%%%%%%%%%%%%%%%%%%%%%%%%%%%%%%%%%%%%%%%%%%%%%%%%%%%%%%%%%%%%%%%%%%%%%%%%%%%%%%%%%%%%%%%%%%%%%
\end{description}
%%%%%%%%%%%%%%%%%%%%%%%%%%%%%%%%%%%%%%%%%%%%%%%%%%%%%%%%%%%%%%%%%%%%%%%%%%%%%%%%%%%%%%%%%%%%%%%%%%%%%%%
Unless otherwise stated explicitly, our algorithms will select exactly $k$ sets.

%%%%%%%%%%%%%%%%%%%%%%%%%%%%%%%%%%%%%%%%%%%%%%%%%%%%%%%%%%%%%%%%%%%%%%%%%%%%%%%%%%%%%%%%%%%%%%%%%%%%%%%
\section{$\LP$-relaxation based randomized approximation algorithms for \fmc}
\label{sec-lp-rand}
%%%%%%%%%%%%%%%%%%%%%%%%%%%%%%%%%%%%%%%%%%%%%%%%%%%%%%%%%%%%%%%%%%%%%%%%%%%%%%%%%%%%%%%%%%%%%%%%%%%%%%%

If $k$ is a constant then we can solve 
\fmcg{\chi}{k}
\emph{exactly} in polynomial (\IE, $O(n^k)$) time by exhaustive enumeration, 
so we assume that $k$ is at least a sufficiently large constant.
In this section we will employ two slightly different $\LP$-relaxation of \fmc 
and combine three randomized rounding analysis on them 
to get an approximation algorithm whose approximation qualities depend on the range of various relevant 
parameters.
The combined approximation result is stated in 
Theorem~\ref{thm-main}.
\emph{In the proof of this theorem 
no serious attempt was made to optimize most constants since we are 
mainly interested in the asymptotic nature of the bounds, 
and to simplify exposition constants have been over-estimated to get nice integers}.
In the statement of 
Theorem~\ref{thm-main}
and in its proof 
we will refer to the three algorithms corresponding to the two $\LP$-relaxations as 
\algsm, \algmed and \algla.

%%%%%%%%%%%%%%%%%%%%%%%%%%%%%%%%%%%%%%%%%%%%%%%%%%%%%%%%%%%%%%%%%%%%%%%%%%%%%%%%%%%%%%%%%%%%%%%%%%%%%%%
\begin{theorem}\label{thm-main}
%%%%%%%%%%%%%%%%%%%%%%%%%%%%%%%%%%%%%%%%%%%%%%%%%%%%%%%%%%%%%%%%%%%%%%%%%%%%%%%%%%%%%%%%%%%%%%%%%%%%%%%
Suppose that the instance of \fmcg{\chi}{k} has $n$ elements and $m$ sets.
Then, we can design three randomized polynomial-time algorithms 
\algsm, \algmed and \algla
with the following properties: 
%%%%%%%%%%%%%%%%%%%%%%%%%%%%%%%%%%%%%%%%%%%%%%%%%%%%%%%%%%%%%%%%%%%%%%%%%%%%%%%%%%%%%%%%%%%%%%%%%%%%%%%
\begin{enumerate}[label=\textbf{\emph{(}\alph*\emph{)}},leftmargin=*]
%%%%%%%%%%%%%%%%%%%%%%%%%%%%%%%%%%%%%%%%%%%%%%%%%%%%%%%%%%%%%%%%%%%%%%%%%%%%%%%%%%%%%%%%%%%%%%%%%%%%%%%
\item
All the three algorithms
select $k$ sets $($with probability $1)$.
%%%%%%%%%%%%%%%%%%%%%%%%%%%%%%%%%%%%%%%%%%%%%%%%%%%%%%%%%%%%%%%%%%%%%%%%%%%%%%%%%%%%%%%%%%%%%%%%%%%%%%%
\item
All the three algorithms
are randomized  
$\varrho(f)$-approximation for \fmc, \emph{\IE}, 
the expected total weight of the selected elements for both 
algorithms is at least $\varrho(f)>1-\nicefrac{1}{\bee}$ times $\opt$.
%%%%%%%%%%%%%%%%%%%%%%%%%%%%%%%%%%%%%%%%%%%%%%%%%%%%%%%%%%%%%%%%%%%%%%%%%%%%%%%%%%%%%%%%%%%%%%%%%%%%%%%
\item
All the three algorithms
satisfy
the randomized $\eps$-approximate coloring constraints $($cf.\ Inequality~\eqref{eq1}$')$
for 
$\eps=O(f)$, \emph{\IE}, 
%%%%%%%%%%%%%%%%%%%%%%%%%%%%%%%%%%%%%%%%%%%%%%%%%%%%%%%%%%%%%%%%%%%%%%%%%%%%%%%%%%%%%%%%%%%%%%%%%%%%%%%
for all 
$i,j\in\{1,\dots,\chi\}$, 
$\frac{\Ave{p_i}}{\Ave{p_j}} \leq 
\frac{2 f }{  \varrho(f) }
< 3.16 f$.
%%%%%%%%%%%%%%%%%%%%%%%%%%%%%%%%%%%%%%%%%%%%%%%%%%%%%%%%%%%%%%%%%%%%%%%%%%%%%%%%%%%%%%%%%%%%%%%%%%%%%%%
\item
The algorithms satisfy
the strong randomized $\eps$-approximate coloring constraints $($cf.\ Equation~\eqref{eq1}$''$, \emph{\IE},
$\bigwedge_{ i,j\in\{1,\dots,\chi\} } \big( \prob { p_i \leq \eps p_j } \big) \geq 1 - o(1))$
for the values of $\eps$, $\opt_\#$ and $\chi$ as shown below:
%%%%%%%%%%%%%%%%%%%%%%%%%%%%%%%%%%%%%%%%%%%%%%%%%%%%%%%%%%%%%%%%%%%%%%%%%%%%%%%%%%%%%%%%%%%%%%%%%%%%%%%

%%%%%%%%%%%%%%%%%%%%%%%%%%%%%%%%%%%%%%%%%%%%%%%%%%%%%%%%%%%%%%%%%%%%%%%%%%%%%%%%%%%%%%%%%%%%%%%%%%%%%%%
\hspace*{-0.8in}
\begin{tabular}{l c c c c}
\midrule
%%%%%%%%%%%%%%%%%%%%%%%%%%%%%%%%%%%%%%%%%%%%%%%%%%%%%%%%%%%%%%%%%%%%%%%%%%%%%%%%%%%%%%%%%%%%%%%%%%%%%%%
& $\eps$ & range of $\opt_\#$ & range of $\chi$ & algorithm 
%%%%%%%%%%%%%%%%%%%%%%%%%%%%%%%%%%%%%%%%%%%%%%%%%%%%%%%%%%%%%%%%%%%%%%%%%%%%%%%%%%%%%%%%%%%%%%%%%%%%%%%
\\
\midrule
%%%%%%%%%%%%%%%%%%%%%%%%%%%%%%%%%%%%%%%%%%%%%%%%%%%%%%%%%%%%%%%%%%%%%%%%%%%%%%%%%%%%%%%%%%%%%%%%%%%%%%%
\text{(\emph{i})} & $O(f)$ & $\Omega(\chi \sqrt{n} \log\chi )$ & unrestricted & \text{\algla}
%%%%%%%%%%%%%%%%%%%%%%%%%%%%%%%%%%%%%%%%%%%%%%%%%%%%%%%%%%%%%%%%%%%%%%%%%%%%%%%%%%%%%%%%%%%%%%%%%%%%%%%
\\
\midrule
%%%%%%%%%%%%%%%%%%%%%%%%%%%%%%%%%%%%%%%%%%%%%%%%%%%%%%%%%%%%%%%%%%%%%%%%%%%%%%%%%%%%%%%%%%%%%%%%%%%%%%%
\text{(\emph{ii})} &  $O(f^2)$ & $\Omega(a\chi\log\chi)$ \text{and} $O(\chi \sqrt{n \log\chi} )$ & unrestricted & \text{\algmed}
%%%%%%%%%%%%%%%%%%%%%%%%%%%%%%%%%%%%%%%%%%%%%%%%%%%%%%%%%%%%%%%%%%%%%%%%%%%%%%%%%%%%%%%%%%%%%%%%%%%%%%%
\\
\midrule
%%%%%%%%%%%%%%%%%%%%%%%%%%%%%%%%%%%%%%%%%%%%%%%%%%%%%%%%%%%%%%%%%%%%%%%%%%%%%%%%%%%%%%%%%%%%%%%%%%%%%%%
\text{(\emph{iii})} &  $O( f^2 \sqrt{a \,\chi \opt_\#}\,)$ & unrestricted  & $\chi =O(\max\{1,\,\frac{\log n}{\log m}\})$ & \text{\algsm}
%%%%%%%%%%%%%%%%%%%%%%%%%%%%%%%%%%%%%%%%%%%%%%%%%%%%%%%%%%%%%%%%%%%%%%%%%%%%%%%%%%%%%%%%%%%%%%%%%%%%%%%
\\
\midrule
%%%%%%%%%%%%%%%%%%%%%%%%%%%%%%%%%%%%%%%%%%%%%%%%%%%%%%%%%%%%%%%%%%%%%%%%%%%%%%%%%%%%%%%%%%%%%%%%%%%%%%%
%%%%%%%%%%%%%%%%%%%%%%%%%%%%%%%%%%%%%%%%%%%%%%%%%%%%%%%%%%%%%%%%%%%%%%%%%%%%%%%%%%%%%%%%%%%%%%%%%%%%%%%
\end{tabular}
%%%%%%%%%%%%%%%%%%%%%%%%%%%%%%%%%%%%%%%%%%%%%%%%%%%%%%%%%%%%%%%%%%%%%%%%%%%%%%%%%%%%%%%%%%%%%%%%%%%%%%%
%
\end{enumerate}
%%%%%%%%%%%%%%%%%%%%%%%%%%%%%%%%%%%%%%%%%%%%%%%%%%%%%%%%%%%%%%%%%%%%%%%%%%%%%%%%%%%%%%%%%%%%%%%%%%%%%%%
\end{theorem}
%%%%%%%%%%%%%%%%%%%%%%%%%%%%%%%%%%%%%%%%%%%%%%%%%%%%%%%%%%%%%%%%%%%%%%%%%%%%%%%%%%%%%%%%%%%%%%%%%%%%%%%

%%%%%%%%%%%%%%%%%%%%%%%%%%%%%%%%%%%%%%%%%%%%%%%%%%%%%%%%%%%%%%%%%%%%%%%%%%%%%%%%%%%%%%%%%%%%%%%%%%%%%%%
\begin{remark}
Note that the 
high-probability $\eps=O(f)$ bound in 
Theorem~\ref{thm-main}\emph{(\emph{d})(\emph{i})}
is asymptotically \emph{the same} as the ``ratio of expectation'' bound in 
Theorem~\ref{thm-main}\emph{(\emph{c})}.
\end{remark}
%%%%%%%%%%%%%%%%%%%%%%%%%%%%%%%%%%%%%%%%%%%%%%%%%%%%%%%%%%%%%%%%%%%%%%%%%%%%%%%%%%%%%%%%%%%%%%%%%%%%%%%

%%%%%%%%%%%%%%%%%%%%%%%%%%%%%%%%%%%%%%%%%%%%%%%%%%%%%%%%%%%%%%%%%%%%%%%%%%%%%%%%%%%%%%%%%%%%%%%%%%%%%%%
\begin{remark}
The dependence on $f$ of the bounds for $\eps$ in 
Theorem~\ref{thm-main}\emph{(\emph{d})(\emph{i})--(\emph{ii})}
can be contrasted with 
the Beck-Fiala theorem in discrepancy minimization that shows 
that the discrepancy of any set system is at most $2f$.
\end{remark}
%%%%%%%%%%%%%%%%%%%%%%%%%%%%%%%%%%%%%%%%%%%%%%%%%%%%%%%%%%%%%%%%%%%%%%%%%%%%%%%%%%%%%%%%%%%%%%%%%%%%%%%

%%%%%%%%%%%%%%%%%%%%%%%%%%%%%%%%%%%%%%%%%%%%%%%%%%%%%%%%%%%%%%%%%%%%%%%%%%%%%%%%%%%%%%%%%%%%%%%%%%%%%%%
\begin{remark}
Consider the special case $\nfmc$
with $\chi=O(1)$: for this case $f=2$ and $a$ 
is equal to the maximum node-degree 
$\deg_{\max}$ 
in the graph. 
The bounds in Theorem~\ref{thm-main}\emph{(\emph{d})(\emph{i})--(\emph{ii})}
for this special case imply a $O(1)$-approximation of color constraints 
unless $\opt_\#$ is not sufficiently large compared to 
$\deg_{\max}$.  
To illustrate the bound for smaller $\opt_\#$
in Theorem~\ref{thm-main}\emph{(\emph{d})(\emph{iii})}, if 
$\opt_\#= \deg_{\max}^{({1}/{2})-\eps}$ for some $\eps>0$ then 
the approximation bound of the coloring constraints is 
$O(\deg_{\max}^{1-\eps})$.
\end{remark}
%%%%%%%%%%%%%%%%%%%%%%%%%%%%%%%%%%%%%%%%%%%%%%%%%%%%%%%%%%%%%%%%%%%%%%%%%%%%%%%%%%%%%%%%%%%%%%%%%%%%%%%

A proof of Theorem~\ref{thm-main} is discussed in the remaining subsections of this section. 
The following notations will be used uniformly throughout the proof.
%%%%%%%%%%%%%%%%%%%%%%%%%%%%%%%%%%%%%%%%%%%%%%%%%%%%%%%%%%%%%%%%%%%%%%%%%%%%%%%%%%%%%%%%%%%%%%%%%%%%%%%
\begin{enumerate}[label=$\triangleright$,leftmargin=*]
%%%%%%%%%%%%%%%%%%%%%%%%%%%%%%%%%%%%%%%%%%%%%%%%%%%%%%%%%%%%%%%%%%%%%%%%%%%%%%%%%%%%%%%%%%%%%%%%%%%%%%%
\item 
$x_1,\dots,x_n\in\{0,1\}$ and $y_1,\dots,y_m\in\{0,1\}$ are the usual \emph{indicator} variables for the elements 
$u_1,\dots,u_n$ and the sets $\cS_1,\dots,\cS_m$, respectively, 
Their values in an optimal solution of the $\LP$-relaxation under consideration will be denoted by 
$x_1^\ast,\dots,x_n^\ast$ and 
$y_1^\ast,\dots,y_m^\ast$, respectively.
%%%%%%%%%%%%%%%%%%%%%%%%%%%%%%%%%%%%%%%%%%%%%%%%%%%%%%%%%%%%%%%%%%%%%%%%%%%%%%%%%%%%%%%%%%%%%%%%%%%%%%%
\item 
$C_j$ is the set of all elements colored $j$ for $j\in\{1,\dotsc,\chi\}$ in the given instance of \fmc.
%%%%%%%%%%%%%%%%%%%%%%%%%%%%%%%%%%%%%%%%%%%%%%%%%%%%%%%%%%%%%%%%%%%%%%%%%%%%%%%%%%%%%%%%%%%%%%%%%%%%%%%
\item 
$\optf$ 
is the optimum value of the objective function of the $\LP$-relaxation under consideration.
%%%%%%%%%%%%%%%%%%%%%%%%%%%%%%%%%%%%%%%%%%%%%%%%%%%%%%%%%%%%%%%%%%%%%%%%%%%%%%%%%%%%%%%%%%%%%%%%%%%%%%%
\end{enumerate}
%%%%%%%%%%%%%%%%%%%%%%%%%%%%%%%%%%%%%%%%%%%%%%%%%%%%%%%%%%%%%%%%%%%%%%%%%%%%%%%%%%%%%%%%%%%%%%%%%%%%%%%

%%%%%%%%%%%%%%%%%%%%%%%%%%%%%%%%%%%%%%%%%%%%%%%%%%%%%%%%%%%%%%%%%%%%%%%%%%%%%%%%%%%%%%%%%%%%%%%%%%%%%%%
\subsection{An obvious generalization of $\LP$-relaxation of maximum $k$-set coverage fails}
%%%%%%%%%%%%%%%%%%%%%%%%%%%%%%%%%%%%%%%%%%%%%%%%%%%%%%%%%%%%%%%%%%%%%%%%%%%%%%%%%%%%%%%%%%%%%%%%%%%%%%%

%%%%%%%%%%%%%%%%%%%%%%%%%%%%%%%%%%%%%%%%%%%%%%%%%%%%%%%%%%%%%%%%%%%%%%%%%%%%%%%%%%%%%%%%%%%%%%%%%%%%%%%
%%%%%%%%%%%%%%%%%%%%%%%%%%%%%%%%%%%%%%%%%%%%%%%%%%%%%%%%%%%%%%%%%%%%%%%%%%%%%%%%%%%%%%%%%%%%%%%%%%%%%%%
\begin{figure}[ht]
\begin{center}
\begin{tabular}{r l l}
\toprule
maximize & $\sum_{i=1}^n w(u_i) x_i$ & 
\\
[4pt]
subject to & $x_j \leq \sum_{u_j\in\cS_\ell} y_\ell$ & for $j=1,\dots,n$   
\\
[4pt]
   & $\sum_{\ell=1}^m y_\ell=k$ &   
\\
[4pt]
   & $0\leq x_j \leq 1$ & for $j=1,\dots,n$
\\
[4pt]
   & $0\leq y_\ell \leq 1$ & for $\ell=1,\dots,m$
\\
\bottomrule
\end{tabular}
\end{center}
\caption{\label{f1}A well-known $\LP$-relaxation of the element-weighted maximum $k$-set coverage problem.}
\end{figure}
%%%%%%%%%%%%%%%%%%%%%%%%%%%%%%%%%%%%%%%%%%%%%%%%%%%%%%%%%%%%%%%%%%%%%%%%%%%%%%%%%%%%%%%%%%%%%%%%%%%%%%%
%%%%%%%%%%%%%%%%%%%%%%%%%%%%%%%%%%%%%%%%%%%%%%%%%%%%%%%%%%%%%%%%%%%%%%%%%%%%%%%%%%%%%%%%%%%%%%%%%%%%%%%

It is well-known that the $\LP$-relaxation for the element-weighted maximum $k$-set coverage problem 
as shown in \FI{f1} followed by a suitable deterministic or randomized rounding 
provides an optimal approximation algorithm for the problem (\EG, see~\cite{AS04,MR95}).  
A straightforward way to extend this $\LP$-relaxation is to add the following 
${\chi(\chi-1)}/{2}$ additional constraints, one corresponding to each pair of colors:
%%%%%%%%%%%%%%%%%%%%%%%%%%%%%%%%%%%%%%%%%%%%%%%%%%%%%%%%%%%%%%%%%%%%%%%%%%%%%%%%%%%%%%%%%%%%%%%%%%%%%%%
\[
\sum_{u_\ell \in C_i} x_\ell = \!\! \sum_{u_\ell \in C_j} x_\ell 
\text{ for } i,j\in\{1,\dots,\chi\}, i<j
\]
%%%%%%%%%%%%%%%%%%%%%%%%%%%%%%%%%%%%%%%%%%%%%%%%%%%%%%%%%%%%%%%%%%%%%%%%%%%%%%%%%%%%%%%%%%%%%%%%%%%%%%%
%
Unfortunately, this may not lead to a $\eps$-approximate coloring (\emph{cf}.\ Equation~\eqref{eq1}) for any non-trivial $\eps$ as 
the following example shows.
Suppose our instance of an unweighted \fmcg{2}{2} has four sets 
$\cS_1=\{u_1\}$,
$\cS_2=\{u_2,\dots,u_{n-2}\}$,
$\cS_3=\{u_{n-1}\}$ 
and 
$\cS_4=\{u_n\}$ 
with the elements $u_1$ and $u_{n-1}$ having color $1$ and all other elements having color $2$.
Clearly the solution to this instance consists of the sets 
$\cS_3$ and $\cS_4$ with $\opt=2$. 
On the other hand, 
the fractional solution 
$y_1^\ast=y_2^\ast=x_1^\ast=x_2^\ast=1$
and all remaining variables being zero is also an optimal solution of the $\LP$-relaxation, but 
any rounding approach that does not change the values of zero-valued variables in the fractional solution 
must necessarily result in an integral solutions in which $p_2/p_1=n-3$. 
The example is easily generalized for arbitrary $k$.

%%%%%%%%%%%%%%%%%%%%%%%%%%%%%%%%%%%%%%%%%%%%%%%%%%%%%%%%%%%%%%%%%%%%%%%%%%%%%%%%%%%%%%%%%%%%%%%%%%%%%%%
\subsection{\algla: strengthening $\LP$-relaxation via additional inequalities}
\label{sec-str}
%%%%%%%%%%%%%%%%%%%%%%%%%%%%%%%%%%%%%%%%%%%%%%%%%%%%%%%%%%%%%%%%%%%%%%%%%%%%%%%%%%%%%%%%%%%%%%%%%%%%%%%

One problem that the $\LP$-relaxation in \FI{f1} faces when applied to \fmc is the following. 
We would like \emph{each} element-indicator variable $x_j$ to satisfy 
$x_j^\ast = \min \big\{1, \, \sum_{u_j\in\cS_\ell} y_\ell^\ast \big\}$
in an optimum solution of the $\LP$, but this may not be true 
as shown in the simple example in the previous section. 
Past researchers have corrected this kind of
situation by introducing extra valid inequalities that hold for any solution to the problem
but restrict the feasible region of the $\LP$.
For example, Carnes and Shmoys in~\cite{CD08} and Carr~\EA in~\cite{CFLP00}
introduced a set of additional inequalities, which they called the KC (Knapsack Cover) 
inequalities, to strengthen the integrality gaps of 
certain types of capacitated covering problems.
Following their ideas, we add the extra $O(fn)$ ``covering inequalities'' which are satisfied by 
\emph{any} integral solution of the $\LP$: 
%%%%%%%%%%%%%%%%%%%%%%%%%%%%%%%%%%%%%%%%%%%%%%%%%%%%%%%%%%%%%%%%%%%%%%%%%%%%%%%%%%%%%%%%%%%%%%%%%%%%%%%
\begin{gather*}
x_j \geq y_\ell \text{ for } j=1,\dots,n,\, \ell=1,\dots,m, \text{ and } u_j\in\cS_\ell
\end{gather*}
%%%%%%%%%%%%%%%%%%%%%%%%%%%%%%%%%%%%%%%%%%%%%%%%%%%%%%%%%%%%%%%%%%%%%%%%%%%%%%%%%%%%%%%%%%%%%%%%%%%%%%%
In addition, we adjust our $\LP$-relaxation in the following manner. 
Since $\opt_\#$ is an integer from $\{\chi,2\chi,\dots,(\lfloor\nicefrac{n}{\chi}\rfloor)\chi\}$, 
we can ``guess'' the correct value of $\opt_\#$ by running the algorithm for \emph{each} of 
the $\lfloor\nicefrac{n}{\chi}\rfloor$ possible value of 
$\opt_\#$, consider those solutions that maximized its objective function and select that one among these
solutions that has the largest value of $\opt_\#$. Thus, we may assume that our $\LP$-relaxation  knows the value of 
$\opt_\#$ \emph{exactly}, and we add the following extra equality:

\medskip
%%%%%%%%%%%%%%%%%%%%%%%%%%%%%%%%%%%%%%%%%%%%%%%%%%%%%%%%%%%%%%%%%%%%%%%%%%%%%%%%%%%%%%%%%%%%%%%%%%%%%%%
\centerline{$\sum_{i=1}^n x_i =\opt_\#$}
%%%%%%%%%%%%%%%%%%%%%%%%%%%%%%%%%%%%%%%%%%%%%%%%%%%%%%%%%%%%%%%%%%%%%%%%%%%%%%%%%%%%%%%%%%%%%%%%%%%%%%%

\medskip
The resulting $\LP$-relaxation is shown in its completeness in \FI{tab1} for convenience.

%%%%%%%%%%%%%%%%%%%%%%%%%%%%%%%%%%%%%%%%%%%%%%%%%%%%%%%%%%%%%%%%%%%%%%%%%%%%%%%%%%%%%%%%%%%%%%%%%%%%%%%
%%%%%%%%%%%%%%%%%%%%%%%%%%%%%%%%%%%%%%%%%%%%%%%%%%%%%%%%%%%%%%%%%%%%%%%%%%%%%%%%%%%%%%%%%%%%%%%%%%%%%%%
\begin{figure}[h]
\begin{center}
\begin{tabular}{r l l}
\toprule
maximize & $\sum_{i=1}^n w(u_i) x_i$ & 
\\
[4pt]
subject to & $x_j \leq \sum_{u_j\in\cS_\ell} y_\ell$ & for $j=1,\dots,n$   
\\
[4pt]
   & $\sum_{\ell=1}^m y_\ell=k$ &   
\\
[4pt]
   & $x_j \geq y_\ell$ & for $j=1,\dots,n$, $\ell=1,\dots,m$, and $u_j\in\cS_\ell$ 
\\
[4pt]
   & $\sum_{i=1}^n x_i =\opt_\#$ &   
\\
[4pt]
   & $\sum_{u_\ell \in C_i} x_\ell = \sum_{u_\ell \in C_j} x_\ell$ &   
	  for $i,j\in\{1,\dots,\chi\}$, $i<j$
\\
[4pt]
   & $0\leq x_j \leq 1$ & for $j=1,\dots,n$
\\
[4pt]
   & $0\leq y_\ell \leq 1$ & for $\ell=1,\dots,m$
\\
\bottomrule
\end{tabular}
\end{center}
\vspace*{-0.2in}
\caption{\label{tab1}A $\LP$-relaxation for \algla with $n+m$ variables and $O\left(fn+\chi^2\right)$ constraints.}
\end{figure}
%%%%%%%%%%%%%%%%%%%%%%%%%%%%%%%%%%%%%%%%%%%%%%%%%%%%%%%%%%%%%%%%%%%%%%%%%%%%%%%%%%%%%%%%%%%%%%%%%%%%%%%
%%%%%%%%%%%%%%%%%%%%%%%%%%%%%%%%%%%%%%%%%%%%%%%%%%%%%%%%%%%%%%%%%%%%%%%%%%%%%%%%%%%%%%%%%%%%%%%%%%%%%%%

%%%%%%%%%%%%%%%%%%%%%%%%%%%%%%%%%%%%%%%%%%%%%%%%%%%%%%%%%%%%%%%%%%%%%%%%%%%%%%%%%%%%%%%%%%%%%%%%%%%%%%%
\subsection{A general technique for obtaining joint high-probability statement}
\label{sec-repeat}
%%%%%%%%%%%%%%%%%%%%%%%%%%%%%%%%%%%%%%%%%%%%%%%%%%%%%%%%%%%%%%%%%%%%%%%%%%%%%%%%%%%%%%%%%%%%%%%%%%%%%%%

Suppose that our randomized $\LP$-relaxation based algorithm guarantees 
that 
$\prob{p_i/p_j > \eps } < \nicefrac{1}{(c\chi^2)}$ for some constant $c\geq 3$
independently for all $i,j\in\{1,\dots,\chi\}$.  
Then 
%%%%%%%%%%%%%%%%%%%%%%%%%%%%%%%%%%%%%%%%%%%%%%%%%%%%%%%%%%%%%%%%%%%%%%%%%%%%%%%%%%%%%
\begin{multline*}
\bigwedge_{ i,j\in\{1,\dots,\chi\} } 
\hspace*{-0.2in}
\prob{ p_i\leq   \eps\,p_j    }
=
1 - 
\hspace*{-0.2in}
\bigvee_{ i,j\in\{1,\dots,\chi\} } 
\hspace*{-0.2in}
\prob{ p_i >  \eps\,p_j   }
%%%%%%%%%%%%%%%%%%%%%%%%%%%%%%%%%%%%%%%%%%%%%%%%%%%%%%%%%%%%%%%%%%%%%%%%%%%%%%%%%%%%%
\\
%%%%%%%%%%%%%%%%%%%%%%%%%%%%%%%%%%%%%%%%%%%%%%%%%%%%%%%%%%%%%%%%%%%%%%%%%%%%%%%%%%%%%
\geq
1 - 
\hspace*{-0.2in}
\sum_{ i,j\in\{1,\dots,\chi\} } 
\hspace*{-0.2in}
\prob{ p_i   \geq \eps\, p_j }
>
1 - 
(\chi^2)\frac{1}{c\chi^2}
=
1- \frac{1}{c}
=c'
\end{multline*}
%%%%%%%%%%%%%%%%%%%%%%%%%%%%%%%%%%%%%%%%%%%%%%%%%%%%%%%%%%%%%%%%%%%%%%%%%%%%%%%%%%%%%
To boost the success probability, 
we repeat the randomized rounding 
$c'\ln n$ times, compute the quantity 
$\sigma = \max_{i,j\in\{1,\dots,\chi\},\,p_j\neq 0}\{\nicefrac{p_i}{p_j}\}$ in each iteration, 
and output the solution in that iteration that resulted in the minimum value of $\sigma$.
It then follows that for the selected solution 
satisfies the strong randomized $\eps$-approximate coloring constraints 
since  
$
\bigwedge_{ i,j\in\{1,\dots,\chi\} } 
\prob{ p_i\leq \eps\, p_j }
\geq
1 - \left(\nicefrac{1}{c'}\right)^{c'\ln n} > 1 - \nicefrac{1}{n^2}
$.

%%%%%%%%%%%%%%%%%%%%%%%%%%%%%%%%%%%%%%%%%%%%%%%%%%%%%%%%%%%%%%%%%%%%%%%%%%%%%%%%%%%%%%%%%%%%%%%%%%%%%%%
\subsection{\algla: further details and relevant analysis}
\label{sec-algla1}
%%%%%%%%%%%%%%%%%%%%%%%%%%%%%%%%%%%%%%%%%%%%%%%%%%%%%%%%%%%%%%%%%%%%%%%%%%%%%%%%%%%%%%%%%%%%%%%%%%%%%%%

For our randomized rounding approach, we recall the following result from~\cite{S01}.

%%%%%%%%%%%%%%%%%%%%%%%%%%%%%%%%%%%%%%%%%%%%%%%%%%%%%%%%%%%%%%%%%%%%%%%%%%%%%%%%%%%%%%%%%%%%%%%%%%%%%%%
\begin{fact}{\rm\cite{S01}}\label{fact3}
Given numbers $p_1,\dots,p_r\in[0,1]$ such that $\ell=\sum_{i=1}^r p_i$ is an integer, 
there exists a polynomial-time algorithm that generates a sequence 
of integers $X_1,\dots,X_r$ such that 
\textbf{$(a)$}
$\sum_{i=1}^r X_i$ with probability $1$, 
\textbf{$(b)$} 
$\prob{X_i=1}=p_i$ for all $i\in\{1,\dots,r\}$, and 
\textbf{$(c)$} 
for any real numbers $\alpha_1,\dots,\alpha_r\in[0,1]$ the sum 
$\sum_{i=1}^r \alpha_i X_i$ 
satisfies standard Chernoff bounds.
\end{fact}
%%%%%%%%%%%%%%%%%%%%%%%%%%%%%%%%%%%%%%%%%%%%%%%%%%%%%%%%%%%%%%%%%%%%%%%%%%%%%%%%%%%%%%%%%%%%%%%%%%%%%%%

We round 
$y_1^\ast,\dots,y_m^\ast$ to 
$y_1^+,\dots,y_m^+$
using the algorithm mentioned in Fact~\ref{fact3}; this ensures
$\sum_{\ell=1}^m y_\ell^+=\sum_{\ell=1}^m y_\ell^\ast=k$ resulting in selection of \emph{exactly} $k$ sets.
This proves the claim in 
\textbf{(\emph{b})}.
%%%%
We round 
$x_1^\ast,\dots,x_n^\ast$ to 
$x_1^+,\dots,x_n^+$ in the following way: 
for $j=1,\dots,n$, 
if $u_j\in \cS_\ell$ for some $y_\ell^+=1$
then set $x_j^+=1$.

\medskip
\noindent
{\bf Proof of (\emph{b})}
\medskip

Our proof of 
\textbf{(\emph{c})}
is similar to that for the maximum $k$-set coverage and is included for the sake of completeness.
Note that 
$x_j^+=0$ if and only if $y_\ell^+=0$ for every set $\cS_\ell$ containing $u_j$ and thus: 
%%%%%%%%%%%%%%%%%%%%%%%%%%%%%%%%%%%%%%%%%%%%%%%%%%%%%%%%%%%%%%%%%%%%%%%%%%%%%%%%%%%%%%%%%%%%%%%%%%%%%%%
\begin{multline}
\textstyle
\Ave{x_j^+} = \Pr[x_j^+=1] = 1 - \Pr[x_j^+=0] 
=
1- \prod_{u_j\in\cS_\ell} \Pr[y_\ell^+=0]
= 
1- \prod_{u_j\in\cS_\ell} (1 - y_\ell^\ast)
\\
\textstyle
\geq 
1- \Big( \frac{\sum_{u_j\in\cS_\ell} \big( 1 - y_\ell^\ast \big) }{f_j} \Big)^{f_j}
=
1- \Big( 1 - \frac{\sum_{u_j\in\cS_\ell} y_\ell^\ast }{f_j} \Big)^{f_j}
\geq 
1- \Big( 1 - \frac{\sum_{u_j\in\cS_\ell} y_\ell^\ast }{f} \Big)^{f}
\label{eq2}
\end{multline}
%%%%%%%%%%%%%%%%%%%%%%%%%%%%%%%%%%%%%%%%%%%%%%%%%%%%%%%%%%%%%%%%%%%%%%%%%%%%%%%%%%%%%%%%%%%%%%%%%%%%%%%
where we have used inequality~\eqref{stdeq-new}. 
If 
$\sum_{u_j\in\cS_\ell} y_\ell^\ast\geq 1$ then obviously~\eqref{eq2} implies  
$\Ave{x_j^+} \geq 
\varrho(f)
\geq
\varrho(f) x_j^\ast
$. 
Otherwise, 
$x_j^\ast\leq \sum_{u_j\in\cS_\ell} y_\ell^\ast<1$ and then by~\eqref{stdeq-new2} we get
%%%%%%%%%%%%%%%%%%%%%%%%%%%%%%%%%%%%%%%%%%%%%%%%%%%%%%%%%%%%%%%%%%%%%%%%%%%%%%%%%%%%%%%%%%%%%%%%%%%%%%%
\begin{gather}
\textstyle
\Ave{x_j^+} \geq 
1- \left( 1 - \frac{\sum_{u_j\in\cS_\ell} y_\ell^\ast }{f} \right)^{f}
>
1- \left( 1 - \frac{x_j^\ast }{f} \right)^{f}
\geq
\varrho(f) x_j^\ast
\label{jj3}
\end{gather}
%%%%%%%%%%%%%%%%%%%%%%%%%%%%%%%%%%%%%%%%%%%%%%%%%%%%%%%%%%%%%%%%%%%%%%%%%%%%%%%%%%%%%%%%%%%%%%%%%%%%%%%
This implies our bound since 
%%%%%%%%%%%%%%%%%%%%%%%%%%%%%%%%%%%%%%%%%%%%%%%%%%%%%%%%%%%%%%%%%%%%%%%%%%%%%%%%%%%%%%%%%%%%%%%%%%%%%%%
\begin{gather*}
\textstyle
\Ave { \sum_{i=1}^n w(u_i) x_i^+ } 
=
\sum_{i=1}^n w(u_i) \Ave {x_i^+} 
\geq 
\sum_{i=1}^n w(u_i) 
\varrho(f) x_j^\ast
=
\varrho(f) \optf
\geq
\varrho(f) \opt
\end{gather*}
%%%%%%%%%%%%%%%%%%%%%%%%%%%%%%%%%%%%%%%%%%%%%%%%%%%%%%%%%%%%%%%%%%%%%%%%%%%%%%%%%%%%%%%%%%%%%%%%%%%%%%%

%%%%%%%%%%%%%%%%%%%%%%%%%%%%%%%%%%%%%%%%%%%%%%%%%%%%%%%%%%%%%%%%%%%%%%%%%%%%%%%%%%%%%%%%%%%%%%%%%%%%%%%
\medskip
\noindent
{\bf Proof of (\emph{c})}
\medskip
%%%%%%%%%%%%%%%%%%%%%%%%%%%%%%%%%%%%%%%%%%%%%%%%%%%%%%%%%%%%%%%%%%%%%%%%%%%%%%%%%%%%%%%%%%%%%%%%%%%%%%%

Note that inequalities~\eqref{stdeq1}~and~\eqref{stdeq3}
imply 
$1-x\leq \bee^{-x} \leq 1 - x + (x^2/2)$
for all $x\in[0,1]$.
In particular, the following implication holds: 
%%%%%%%%%%%%%%%%%%%%%%%%%%%%%%%%%%%%%%%%%%%%%%%%%%%%%%%%%%%%%%%%%%%%%%%%%%%%%%%%%%%%%%%%%%%%%%%%%%%%%%%
\begin{gather}
\forall\, c>1 \,\, \forall \, x\in [ 0, (2/c^2)(c-1) ] \,:\, 1-x \geq 1 - cx + (c^2x^2/2) \geq \bee^{-cx}
\label{stdeq4}
\end{gather}
%%%%%%%%%%%%%%%%%%%%%%%%%%%%%%%%%%%%%%%%%%%%%%%%%%%%%%%%%%%%%%%%%%%%%%%%%%%%%%%%%%%%%%%%%%%%%%%%%%%%%%%
We estimate an upper bound on 
$\Ave{x_j^+}$ 
in terms of $x_j^\ast$ in the following manner: 
%%%%%%%%%%%%%%%%%%%%%%%%%%%%%%%%%%%%%%%%%%%%%%%%%%%%%%%%%%%%%%%%%%%%%%%%%%%%%%%%%%%%%%%%%%%%%%%%%%%%%%%
\begin{description}
\item[Case 1: $\pmb{\exists \, \ell}$ such that $\pmb{u_j\in\cS_\ell}$ and $\pmb{y_\ell^\ast>\nicefrac{1}{2}}$.] 
Thus,  
$x_j^\ast \geq y_\ell^\ast>\nicefrac{1}{2}$, and 
$\Ave{x_j^+}\leq 1 \leq 2 x_j^\ast$. 
\item[Case 2: $\pmb{y_\ell^\ast\leq \nicefrac{1}{2}}$ for every index $\pmb{\ell}$ satisfying $\pmb{u_j\in\cS_\ell}$.] 
Note that 
$x_j^\ast\geq \left( \sum_{u_j\in\cS_\ell} y_\ell^\ast \right) / f$,
and 
setting $c=2$ in inequality~\eqref{stdeq4} we get 
$1-x \geq \bee^{-2 x}$ for all $x\in [0,\nicefrac{1}{2}]$.
Now, standard calculations show the following: 
%%%%%%%%%%%%%%%%%%%%%%%%%%%%%%%%%%%%%%%%%%%%%%%%%%%%%%%%%%%%%%%%%%%%%%%%%%%%%%%%%%%%%%%%%%%%%%%%%%%%%%%
\begin{multline*}
\Ave{x_j^+}
= 
1- \prod_{u_j\in\cS_\ell} (1 - y_\ell^\ast)
\leq 
1- \prod_{u_j\in\cS_\ell} \bee^{- 2 y_\ell^\ast}
=
1- \bee^{- 2 \sum_{u_j\in\cS_\ell} y_\ell^\ast}
\\
= 
1- \bee^{- 2 f x_j^\ast}
\leq 
1 - (1 -   2 f x_j^\ast) =  
2 f  x_j^\ast
\end{multline*}
%%%%%%%%%%%%%%%%%%%%%%%%%%%%%%%%%%%%%%%%%%%%%%%%%%%%%%%%%%%%%%%%%%%%%%%%%%%%%%%%%%%%%%%%%%%%%%%%%%%%%%%
\end{description}
%%%%%%%%%%%%%%%%%%%%%%%%%%%%%%%%%%%%%%%%%%%%%%%%%%%%%%%%%%%%%%%%%%%%%%%%%%%%%%%%%%%%%%%%%%%%%%%%%%%%%%%
Combining all the cases and using~\eqref{jj3}, it follows that 
%%%%%%%%%%%%%%%%%%%%%%%%%%%%%%%%%%%%%%%%%%%%%%%%%%%%%%%%%%%%%%%%%%%%%%%%%%%%%%%%%%%%%%%%%%%%%%%%%%%%%%%
$
\varrho(f) x_j^\ast
\leq \Ave{x_j^+}\leq 
\min \big\{ 1,\, 2 f  x_j^\ast \big\}
$.
%%%%%%%%%%%%%%%%%%%%%%%%%%%%%%%%%%%%%%%%%%%%%%%%%%%%%%%%%%%%%%%%%%%%%%%%%%%%%%%%%%%%%%%%%%%%%%%%%%%%%%%
Recall that 
$\sum_{u_\ell \in C_j} x_\ell^\ast = \nicefrac{\opt_\#}{\chi}>0$ for every $j\in\{1,\dots,\chi\}$.
Since 
$\Ave{p_j} = \Ave { \sum_{u_\ell \in C_j} x_\ell^+ }= \sum_{u_\ell \in C_j} \Ave{x_\ell^+}$,
we get the following bounds for all 
$j\in\{1,\dots,\chi\}$:
%%%%%%%%%%%%%%%%%%%%%%%%%%%%%%%%%%%%%%%%%%%%%%%%%%%%%%%%%%%%%%%%%%%%%%%%%%%%%%%%%%%%%%%%%%%%%%%%%%%%%%%
\begin{gather}
\textstyle
\varrho(f) \frac{\opt_\#}{\chi}
=
\sum_{u_\ell \in C_j}
\varrho(f) x_\ell^\ast
\leq
\Ave{p_j} = 
\sum_{u_\ell \in C_j} \Ave{x_\ell^+}
\leq 
\sum_{u_\ell \in C_j} (2 f )x_\ell^\ast
= 2 f 
\frac{\opt_\#}{\chi}
\label{eq3}
\end{gather}
%%%%%%%%%%%%%%%%%%%%%%%%%%%%%%%%%%%%%%%%%%%%%%%%%%%%%%%%%%%%%%%%%%%%%%%%%%%%%%%%%%%%%%%%%%%%%%%%%%%%%%%
which 
gives the bound
$
\frac{\Ave{c_i}}{\Ave{c_j}}
\leq
\frac{2 f }{  \varrho(f)}
$
for all 
$i,j\in\{1,\dots,\chi\}$.

%%%%%%%%%%%%%%%%%%%%%%%%%%%%%%%%%%%%%%%%%%%%%%%%%%%%%%%%%%%%%%%%%%%%%%%%%%%%%%%%%%%%%%%%%%%%%%%%%%%%%%%
\bigskip
\noindent
{\bf Proofs of (\emph{d})(\emph{i}) via Doob martingales}
\medskip
%%%%%%%%%%%%%%%%%%%%%%%%%%%%%%%%%%%%%%%%%%%%%%%%%%%%%%%%%%%%%%%%%%%%%%%%%%%%%%%%%%%%%%%%%%%%%%%%%%%%%%%

Note that the random variables 
$x_1^+,\dots,x_n^+$ may \emph{not} be pairwise independent since two distinct elements belonging to the same set are
correlated, and consequently the random variables 
$p_1,\dots,p_\chi$ also may \emph{not} be pairwise independent.
Indeed in the worst case
an element-selection variable may be correlated to $(a-1)f$ other element-selection variables, 
thereby ruling out straightforward use of Chernoff-type tail bounds.

For sufficient large $\opt_\#$, 
this situation can be somewhat remedied by using Doob martinagales and Azuma's inequality by finding 
a suitable ordering of the element-selection variables conditional on the rounding of the set-selection variables.
We assume that the reader is familiar with basic definitions and results for the theory of martingales 
(\EG, see~\cite[Section $4.4$]{MR95}).
Fix an arbitrary ordered sequence $y_1^+,\dots,y_m^+$
of the set-indicator variables. 
Call an element-indicator variable $x_i^+$ ``settled'' at the $t\tx$ step 
if and only if 
$\cup_{u_i\in\cS_j}\{y_j\}\not\subseteq \{y_1^+,\dots,y_{t-1}^+\}$
and 
$\cup_{u_i\in\cS_j}\{y_j\}\subseteq \{y_1^+,\dots,y_t^+\}$.
The elementary event in our underlying sample space 
$\pmb{\Omega}$
are all possible $2^n$ assignments of $0$-$1$ values to the variables 
$x_1^+,\dots,x_n^+$.
%%%%%
For each $t\in\{1,\dots,m\}$, let $V_t$ be the subset of element-selection 
variables whose values are 
settled at the $t\tx$ step, 
let $\pi_t$ be an arbitrary ordering of the variables in $V_t$, and let 
us relabel the element-indicator variable names so that 
$x_{1}^+,x_{2}^+,\dots,x_{n}^+$ 
be the ordering of all element-selection variables given by the ordering $\pi_1,\dots,\pi_m$.
%%%
For each $t\in \{0,1,\dots,m\}$ and each 
$w_1,\dots,w_t\in \{0,1\}$, let
$\mathbf{B}_{w_1,\dots,w_t}$ 
denote the event that 
$y_j^+=w_j$ for $j\in\{1,\dots,t\}$.
%%%%
Let 
$x_{1}^+,x_{2}^+,\dots,x_{q_t}^+$ 
be the union of set of all $q_t$ element-indicator variables that are settled at the $i\tx$ step 
over all $i\in\{1,\dots,t\}$, and 
suppose that the event
$\mathbf{B}_{w_1,\dots,w_t}$ 
induces the following assignment of values to the element-indicator variables: 
$x_{1}^+ = b_1, \dots, x_{q_t}^+ = b_{q_t}$ for some $b_1,\dots,b_{q_t}\in\{0,1\}$.
Define the block 
$\mathbf{B}'_{w_1,\dots,w_t}\subseteq \pmb{\Omega}$ induced by the event  
$\mathbf{B}_{w_1,\dots,w_t}$ 
as 
$\mathbf{B}'_{w_1,\dots,w_t}=\{b_1,\dots,b_{q_t} r_{q_t+1} \dots r_{n} \,|\, r_{q_t+1},\dots,r_{n}\in\{0,1\}\}$. 
%%%%%%%%%%%
Letting  
$\pmb{\mathbb{F}}_t$ 
be the $\sigma$-field generated by the partition of 
$\pmb{\Omega}$
into the blocks 
$\mathbf{B}'_{w_1,\dots,w_t}$ 
for each 
$w_1,\dots,w_t\in \{0,1\}$, 
it follows that 
$\pmb{\mathbb{F}}_0, \pmb{\mathbb{F}}_1, \dots, \pmb{\mathbb{F}}_m$
form a filter for the $\sigma$-field 
$(\pmb{\Omega},2^{\pmb{\Omega}})$.
Suppose that $\cU$ contains $n_i>0$ elements of color $i$,
let 
$x_{\alpha_1}^+,\dots,x_{\alpha_{n_i}}^+$
be the 
ordered sequence 
of the element-selection variables for elements of color $i$ determined by the subsequence of these variables in 
the ordering $x_1^+,\dots,x_n^+$, 
and suppose that 
$x_{\alpha_j}^+$
was settled at the 
$t_{\alpha_j}\tx$ step. 
Let $X^i=\sum_{j=1}^{n_i} x_{\alpha_j}^+$, 
and 
define the Doob martingale sequence $X_0,X_1,\dots,X_{n_i}$ 
where 
$X_0=\Ave{X^i}
=\sum_{j=1}^{n_i} \Ave{x_{\alpha_j}^+}
$, and 
$X_\ell=\Ave{ X^i \,|\, y_{1}^+,\dots,y_{t_{\alpha_\ell}}^+ }$ 
for all $\ell\in\{1,\dots,n_i\}$. 
Since 
$n_i< n $, 
$X_{n_i}=X^i$ and $|X_\ell - X_{\ell-1} | \leq 1$ for 
for all $\ell\in\{1,\dots,n_i\}$, 
by Azuma's inequality (for any $\Delta>\bee$),  
%%%%%%%%%%%%%%%%%%%%%%%%%%%%%%%%%%%%%%%%%%%%%%%%%%%%%%%%%%%%%%%%%%%%%%%%%%%%%%%%%%%%%%%%%%%%%%%%%%%%%%%
\begin{multline*}
\prob{ \, \left | {\textstyle \sum_{j=1}^{n_i} } x_{\alpha_j}^+ 
         - {\textstyle \sum_{j=1}^{n_i} } \Ave{x_{\alpha_j}^+} \right| \geq 3 \sqrt{\ln \Delta} \sqrt{n} }
\leq
\prob{ \, \left | {\textstyle \sum_{j=1}^{n_i} } x_{\alpha_j}^+ - X_0 \right| \geq 3 \sqrt{\ln \Delta} \sqrt{n_i} }
%%%%%%%%%%%%%%%%%%%%%%%%%%%%%%%%%%%%%%%%%%%%%%%%%%%%%%%%%%%%%%%%%%%%%%%%%%%%%%%%%%%%%%%%%%%%%%%%%%%%%%%
\\
%%%%%%%%%%%%%%%%%%%%%%%%%%%%%%%%%%%%%%%%%%%%%%%%%%%%%%%%%%%%%%%%%%%%%%%%%%%%%%%%%%%%%%%%%%%%%%%%%%%%%%%
=
\prob{|X_{n_i} - X_0| \geq 3 \sqrt{\ln \Delta} \sqrt{n_i} }
\leq
\bee^{-4\ln \Delta}
= \Delta^{-4}
%%%%%%%%%%%%%%%%%%%%%%%%%%%%%%%%%%%%%%%%%%%%%%%%%%%%%%%%%%%%%%%%%%%%%%%%%%%%%%%%%%%%%%%%%%%%%%%%%%%%%%%
\\
%%%%%%%%%%%%%%%%%%%%%%%%%%%%%%%%%%%%%%%%%%%%%%%%%%%%%%%%%%%%%%%%%%%%%%%%%%%%%%%%%%%%%%%%%%%%%%%%%%%%%%%
\Rightarrow \,
\prob{ \, 
         {\textstyle \sum_{j=1}^{n_i} } \Ave{x_{\alpha_j}^+} - 3 \sqrt{\ln \Delta} \sqrt{n} 
				 <
         {\textstyle \sum_{j=1}^{n_i} } x_{\alpha_j}^+ 
				 < 
         {\textstyle \sum_{j=1}^{n_i} } \Ave{x_{\alpha_j}^+} + 3 \sqrt{\ln \Delta} \sqrt{n} }
>
1- \Delta^{-4}
\end{multline*}
%%%%%%%%%%%%%%%%%%%%%%%%%%%%%%%%%%%%%%%%%%%%%%%%%%%%%%%%%%%%%%%%%%%%%%%%%%%%%%%%%%%%%%%%%%%%%%%%%%%%%%%
Recall from~\eqref{eq3} 
that 
$(1-\nicefrac{1}{\bee}) \frac{\opt_\#}{\chi} <  \sum_{j=1}^{n_i} \Ave{x_{\alpha_j}^+} \leq 2f \frac{\opt_\#}{\chi}$, 
and 
the inequality 
$\sum_{j=1}^{n_i} \Ave{x_{\alpha_j}^+} \geq 6 \sqrt{\ln \Delta} \sqrt{n}$
is therefore satisfied provided 
$ \opt \geq 6 \big(1 - \nicefrac{1}{\bee} \big) ^{-1} (\sqrt{ n \ln \Delta }) \chi$.
Setting 
$\Delta=2\chi$ and 
remembering that 
$p_i = {\textstyle \sum_{j=1}^{n_i} } x_{\alpha_j}^+$,
we get the following bound for any $i\in\{1,\dots,\chi\}$: 
%%%%%%%%%%%%%%%%%%%%%%%%%%%%%%%%%%%%%%%%%%%%%%%%%%%%%%%%%%%%%%%%%%%%%%%%%%%%%%%%%%%%%%%%%%%%%%%%%%%%%%%
\begin{gather*}
%%%%%%%%%%%%%%%%%%%%%%%%%%%%%%%%%%%%%%%%%%%%%%%%%%%%%%%%%%%%%%%%%%%%%%%%%%%%%%%%%%%%%%%%%%%%%%%%%%%%%%%
\prob{ \, \textstyle
         {\frac{1-\nicefrac{1}{\bee}}{2} } \left( \frac{\opt_\#}{\chi} \right)
				 <
         p_i 
				 < 
         3 f \left( \frac{\opt_\#}{\chi}  \right)}
>
1 - (1/16)\chi^{-4}
%%%%%%%%%%%%%%%%%%%%%%%%%%%%%%%%%%%%%%%%%%%%%%%%%%%%%%%%%%%%%%%%%%%%%%%%%%%%%%%%%%%%%%%%%%%%%%%%%%%%%%%
\end{gather*}
%%%%%%%%%%%%%%%%%%%%%%%%%%%%%%%%%%%%%%%%%%%%%%%%%%%%%%%%%%%%%%%%%%%%%%%%%%%%%%%%%%%%%%%%%%%%%%%%%%%%%%%
%
and therefore 
$
\prob { 
\frac {p_i}{p_j} \leq 
\frac {6} {1-\nicefrac{1}{\bee}} f
}
>
1 - (1/8)\chi^{-4}
$.
Thus, it follows that 
%%%%%%%%%%%%%%%%%%%%%%%%%%%%%%%%%%%%%%%%%%%%%%%%%%%%%%%%%%%%%%%%%%%%%%%%%%%%%%%%%%%%%
\begin{multline*}
\bigwedge_{ i,j\in\{1,\dots,\chi\} } 
\hspace*{-0.2in}
\prob{ p_i\leq   {\textstyle \frac {6} {1-\nicefrac{1}{\bee}} } \,p_j    }
=
1 - 
\hspace*{-0.2in}
\bigvee_{ i,j\in\{1,\dots,\chi\} } 
\hspace*{-0.2in}
\prob{ p_i >  {\textstyle \frac {6} {1-\nicefrac{1}{\bee}} } \,p_j   }
%%%%%%%%%%%%%%%%%%%%%%%%%%%%%%%%%%%%%%%%%%%%%%%%%%%%%%%%%%%%%%%%%%%%%%%%%%%%%%%%%%%%%
\\
%%%%%%%%%%%%%%%%%%%%%%%%%%%%%%%%%%%%%%%%%%%%%%%%%%%%%%%%%%%%%%%%%%%%%%%%%%%%%%%%%%%%%
\geq
1 - 
\hspace*{-0.2in}
\sum_{ i,j\in\{1,\dots,\chi\} } 
\hspace*{-0.2in}
\prob{ p_i   \geq {\textstyle \frac {6} {1-\nicefrac{1}{\bee}} } \, p_j }
>
{\textstyle
1 - 
(\chi^2)\frac{1}{8\chi^4}
\geq
\frac{1}{32}
}
\end{multline*}
%%%%%%%%%%%%%%%%%%%%%%%%%%%%%%%%%%%%%%%%%%%%%%%%%%%%%%%%%%%%%%%%%%%%%%%%%%%%%%%%%%%%%
This implies our claim in (\emph{d})(\emph{i})
using the technique in Section~\ref{sec-repeat}.

%%%%%%%%%%%%%%%%%%%%%%%%%%%%%%%%%%%%%%%%%%%%%%%%%%%%%%%%%%%%%%%%%%%%%%%%%%%%%%%%%%%%%%%%%%%%%%%%%%%%%%%
\subsection{\algmed: details and proofs of relevant claims in (\emph{a})--(\emph{c}) and (\emph{d})(\emph{ii})}
\label{sec-algmed1}
%%%%%%%%%%%%%%%%%%%%%%%%%%%%%%%%%%%%%%%%%%%%%%%%%%%%%%%%%%%%%%%%%%%%%%%%%%%%%%%%%%%%%%%%%%%%%%%%%%%%%%%

%%%%%%%%%%%%%%%%%%%%%%%%%%%%%%%%%%%%%%%%%%%%%%%%%%%%%%%%%%%%%%%%%%%%%%%%%%%%%%%%%%%%%%%%%%%%%%%%%%%%%%%
\medskip
\noindent
\textbf{\algmed: idea behind the modified $\LP$-relaxation and approach}
\medskip
%%%%%%%%%%%%%%%%%%%%%%%%%%%%%%%%%%%%%%%%%%%%%%%%%%%%%%%%%%%%%%%%%%%%%%%%%%%%%%%%%%%%%%%%%%%%%%%%%%%%%%%

A limitation of \algla is that we could not use Fact~\ref{fact3} of Srinivasan to the fullest extent. 
Although Fact~\ref{fact3} guaranteed that the set-indicator variables are negatively correlated and 
hence Chernoff-type tail bounds can be applied to them due to the result by 
Panconesi and Srinivasan~\cite{PS97}, our coloring constraints are primarily indicated by 
element-indicator variables which depend \emph{implicitly} on the set-indicator variables.
In fact, it is not difficult to see that the element-indicator variables are \emph{not} negatively 
correlated in the sense of~\cite{PS97,S01}\footnote{A set of binary random variables 
$z_1,\dots,z_r\in\{0,1\}$ are called negatively correlated in~\cite{PS97,S01} if the following holds:
$\forall \, I\subseteq\{1,\dots,r\}: \, 
\prob{\wedge_{i\in I} (z_i=0)} \leq \prod_{i\in I} \prob{z_i=0} 
\text{ and }
\prob{\wedge_{i\in I} (z_i=1)} \leq \prod_{i\in I} \prob{z_i=1}$}
even if the set-indicator variables are negatively correlated.

Our idea is to remedy the situation by expressing the coloring constraints also by set-indicator variables 
and use the element-indicator variables to \emph{implicitly} control the set-indicator variables in these
coloring constraints. This will also necessitate using additional variables.

%%%%%%%%%%%%%%%%%%%%%%%%%%%%%%%%%%%%%%%%%%%%%%%%%%%%%%%%%%%%%%%%%%%%%%%%%%%%%%%%%%%%%%%%%%%%%%%%%%%%%%%
\bigskip
\noindent
\textbf{A modification of the $\LP$-relaxation in \FI{tab1}}
\medskip
%%%%%%%%%%%%%%%%%%%%%%%%%%%%%%%%%%%%%%%%%%%%%%%%%%%%%%%%%%%%%%%%%%%%%%%%%%%%%%%%%%%%%%%%%%%%%%%%%%%%%%%

To begin, we quantify the number of elements of different colors in a set $\cS_i$ using the following notation: 
for $j\in\{1,\dots,\chi\}$, 
let 
$\nu_{i,j}$ be the number of elements in $\cS_i$ of color $j$.
Note that $0\leq \nu_{i,j}\leq a$.
Fix an optimal integral solution of \fmcg{\chi}{k} covering $\opt_\#$ elements and a color value $j$, and 
consider the following two quantities: 
$\cA=\sum_{u_\ell\in C_j} x_\ell$
and 
$\cB=\sum_{i=1}^m \nu_{i,j}\, y_i$.
Note that 
$\cA=\frac{\opt_\#}{\chi}$, 
$\cB\in\{k,k+1,\dots,ak\}$, and 
$\cA\leq \cB \leq f \cA$
by definition of $f$.
Thus, 
$\cB=\h_j \opt_\#$ 
is satisfied by 
a $\h_j$ that is a rational number from the set 
$
\big\{
\frac{1}{\chi},
\frac{1}{\chi}+\frac{1}{\opt_\#},
\frac{1}{\chi}+\frac{2}{\opt_\#},
\dots,
\frac{f}{\chi}-\frac{1}{\opt_\#},
\frac{f}{\chi}
\big\}
$. 
We will use the $\LP$-relaxation in \FI{tab1} 
with 
$\chi$ additional variables $\h_1,\dots,\h_\chi$ 
and the following additional constraints 

\smallskip
\begin{center}
\begin{tabular}{l l}
$\sum_{i=1}^m \nu_{i,j}\, y_i = \h_j\opt_\#$ &  for $j\in\{1,\dots,\chi\}$
\\
[4pt]
$\nicefrac{1}{\chi}\leq \h_j\leq \nicefrac{f}{\chi}$ & for $j=1,\dots,\chi$
\end{tabular}
\end{center}

For reader's convenience, the new $\LP$-relaxation in its entirety is shown in \FI{tab3}. 
%%%%%%%%%%%%%%%%%%%%%%%%%%%%%%%%%%%%%%%%%%%%%%%%%%%%%%%%%%%%%%%%%%%%%%%%%%%%%%%%%%%%%%%%%%%%%%%%%%%%%%%
%%%%%%%%%%%%%%%%%%%%%%%%%%%%%%%%%%%%%%%%%%%%%%%%%%%%%%%%%%%%%%%%%%%%%%%%%%%%%%%%%%%%%%%%%%%%%%%%%%%%%%%
\begin{figure}[ht]
\begin{center}
\begin{tabular}{r l l}
\toprule
maximize & $\sum_{i=1}^n w(u_i) x_i$ & 
\\
[4pt]
subject to & $x_j \leq \sum_{u_j\in\cS_\ell} y_\ell$ & for $j=1,\dots,n$   
\\
[4pt]
   & $\sum_{\ell=1}^m y_\ell=k$ &   
\\
[4pt]
   & $x_j \geq y_\ell$ & for $j=1,\dots,n$, $\ell=1,\dots,m$, and $u_j\in\cS_\ell$ 
\\
[4pt]
   & $\sum_{i=1}^n x_i =\opt_\#$ &   
\\
[4pt]
   & $\sum_{i=1}^m \nu_{i,j}\, y_i = \h_j\opt_\#$ &  for $j\in\{1,\dots,\chi\}$
\\
[4pt]
   & $\sum_{u_\ell \in C_i} x_\ell = \sum_{u_\ell \in C_j} x_\ell$ &   
	  for $i,j\in\{1,\dots,\chi\}$, $i<j$
\\
[4pt]
   & $0\leq x_j \leq 1$ & for $j=1,\dots,n$
\\
[4pt]
   & $0\leq y_\ell \leq 1$ & for $\ell=1,\dots,m$
\\
[4pt]
   & $\nicefrac{1}{\chi}\leq \h_j\leq \nicefrac{f}{\chi}$ & for $j=1,\dots,\chi$
\\
\bottomrule
\end{tabular}
\end{center}
\vspace*{-0.2in}
\caption{\label{tab3}A modified $\LP$-relaxation for \algmed with 
with $n+m+\chi$ variables and $O\left(f n+\chi^2\right)$ constraints.}
\end{figure}
%%%%%%%%%%%%%%%%%%%%%%%%%%%%%%%%%%%%%%%%%%%%%%%%%%%%%%%%%%%%%%%%%%%%%%%%%%%%%%%%%%%%%%%%%%%%%%%%%%%%%%%
%%%%%%%%%%%%%%%%%%%%%%%%%%%%%%%%%%%%%%%%%%%%%%%%%%%%%%%%%%%%%%%%%%%%%%%%%%%%%%%%%%%%%%%%%%%%%%%%%%%%%%%

%%%%%%%%%%%%%%%%%%%%%%%%%%%%%%%%%%%%%%%%%%%%%%%%%%%%%%%%%%%%%%%%%%%%%%%%%%%%%%%%%%%%%%%%%%%%%%%%%%%%%%%
\medskip
\noindent
\textbf{Analysis of the modified $\LP$-relaxation}
\medskip
%%%%%%%%%%%%%%%%%%%%%%%%%%%%%%%%%%%%%%%%%%%%%%%%%%%%%%%%%%%%%%%%%%%%%%%%%%%%%%%%%%%%%%%%%%%%%%%%%%%%%%%

By our assumption on $\h_j$'s, the $\LP$-relaxation has a feasible solution.  
We use the same randomized rounding procedure (using Fact~\ref{fact3}) as in Section~\ref{sec-algla1}
for \algla.
The proofs for parts 
(\emph{b})--(\emph{d})
are the same as before since all prior relevant inequalities are still included. Thus, we concentrate on the proof of 
(\emph{e})(\emph{ii}).
A crucial thing to note is the following simple observation: 
%%%

\begin{quote}
Consider the sum 
$\Delta=\sum_{i=1}^m \nu_{i,j}\, y_i$
for any assignment of values $y_1,\dots,y_m\in\{0,1\}$. 
Then, the number of elements covered by the sets corresponding to those variables that are set to $1$ is
between $\Delta/f$ and $\Delta$.
\end{quote}
%%%

Fix a color $j$.
Let $\cK_j=\h_j\opt_\#$, $\alpha_i=\frac{\nu_{i,j}}{a}$ and consider the summation 
$\cL_j=\sum_{i=1}^m \alpha_i \, y_i^+$. 
Since $\alpha_i\in[0,1]$ for all $i$, by Fact~\ref{fact3} we can apply standard Chernoff bounds~\cite{MR95}
for $\cL_j$.
Note that 
$\Ave{\cL_j}= \sum_{i=1}^m \alpha_i \, y_i^\ast = \frac{\cK_j}{a}$. 
Assuming 
$\cK_j\geq 16a\ln \chi$,
we get the following for the tail-bounds:
%%%%%%%%%%%%%%%%%%%%%%%%%%%%%%%%%%%%%%%%%%%%%%%%%%%%%%%%%%%%%%%%%%%%%%%%%%%%%%%%%%%%%%%%%%%%%%%%%%%%%%%
\begin{gather*}
\prob{ {\textstyle \sum_{i=1}^m  \nu_{i,j}\, y_i^+ > 5 \cK_j } }
=
\prob{\cL_j > 5 (\nicefrac{\cK_j}{a}) } < 
2^{-{6\cK_j}/{a}}
\leq 
2^{-96\ln\chi}
< 
{\chi}^{-96}
%%%%%%%%%%%%%%%%%%%%%%%%%%%%%%%%%%%%%%%%%%%%%%%%%%%%%%%%%%%%%%%%%%%%%%%%%%%%%%%%%%%%%%%%%%%%%%%%%%%%%%%
\\
%%%%%%%%%%%%%%%%%%%%%%%%%%%%%%%%%%%%%%%%%%%%%%%%%%%%%%%%%%%%%%%%%%%%%%%%%%%%%%%%%%%%%%%%%%%%%%%%%%%%%%%
\prob{ {\textstyle \sum_{i=1}^m  \nu_{i,j}\, y_i^+ < \nicefrac{\cK_j}{2} } }
=
\prob{ {\textstyle \cL_j <  \frac{{\cK_j}/{a}}{2} }  } < 
\bee^{-\frac{\cK_j}{8a}}
\leq \bee^{-2\ln\chi} 
=
{\chi}^{-2}
\end{gather*}
%%%%%%%%%%%%%%%%%%%%%%%%%%%%%%%%%%%%%%%%%%%%%%%%%%%%%%%%%%%%%%%%%%%%%%%%%%%%%%%%%%%%%%%%%%%%%%%%%%%%%%%
Remember that 
$p_j=\sum_{u_\ell \in C_j} x_\ell^+$ is the random variable denoting the 
number of elements of color $j$ selected by our randomized algorithm. 
Since 
$\frac{1}{f}\sum_{i=1}^m  \nu_{i,j}\, y_i^+ \leq c_j \leq \sum_{i=1}^m  \nu_{i,j}\, y_i^+$
we get 
%%%%%%%%%%%%%%%%%%%%%%%%%%%%%%%%%%%%%%%%%%%%%%%%%%%%%%%%%%%%%%%%%%%%%%%%%%%%%%%%%%%%%%%%%%%%%%%%%%%%%%%
\begin{gather*}
\prob{ p_j > 5 \cK_j }
\leq
\prob{ \textstyle \sum_{i=1}^m  \nu_{i,j}\, y_i^+ > 5 \cK_j }
<
{\chi}^{-96},
\\
\prob{ \textstyle p_j < \frac{\cK_j}{2f} }
\leq 
\prob{ \textstyle \sum_{i=1}^m  \nu_{i,j}\, y_i^+ < \frac{\cK_j}{2} }
< {\chi}^{-2}
\end{gather*}
%%%%%%%%%%%%%%%%%%%%%%%%%%%%%%%%%%%%%%%%%%%%%%%%%%%%%%%%%%%%%%%%%%%%%%%%%%%%%%%%%%%%%%%%%%%%%%%%%%%%%%%
Note that 
$
\frac{\opt_\#}{\chi}
\leq
\cK_j=\h_j\opt_\#
\leq 
\frac{f \opt_\#}{\chi}
$, and therefore 
$
\nicefrac{1}{f}
\leq
\frac{\cK_i}{\cK_j}
\leq
f
$
for any two $i,j\in\{1,\dots,\chi\}$.
Let $\cE_j$ be the event defined as 
$\cE_j\eqdef \frac{\cK_j}{2f} \leq c_j \leq 5 \cK_j$.
Then for any two $i,j\in\{1,\dots,\chi\}$ we get 
%%%%%%%%%%%%%%%%%%%%%%%%%%%%%%%%%%%%%%%%%%%%%%%%%%%%%%%%%%%%%%%%%%%%%%%%%%%%%%%%%%%%%%%%%%%%%%%%%%%%%%%
\begin{multline}
\prob{ {\textstyle \frac{p_i}{p_j} \leq 10f^2 } }
\geq
\prob{\cE_i \wedge \cE_j}
=
1- \prob{\,\overline{\cE_i} \vee \overline{\cE_j}\,}
\geq 
1- \prob{\,\overline{\cE_i} \,} - \prob{\,\overline{\cE_j}\,}
%%%%%%%%%%%%%%%%%%%%%%%%%%%%%%%%%%%%%%%%%%%%%%%%%%%%%%%%%%%%%%%%%%%%%%%%%%%%%%%%%%%%%%%%%%%%%%%%%%%%%%%
\\
%%%%%%%%%%%%%%%%%%%%%%%%%%%%%%%%%%%%%%%%%%%%%%%%%%%%%%%%%%%%%%%%%%%%%%%%%%%%%%%%%%%%%%%%%%%%%%%%%%%%%%%
\geq 
1 - 
\prob{ p_i > 5 \cK_i } - \prob{ p_i < \nicefrac{\cK_i}{(2f)} } - 
\prob{ p_j > 5 \cK_j } - \prob{ p_j < \nicefrac{\cK_j}{(2f)} } 
%%%%%%%%%%%%%%%%%%%%%%%%%%%%%%%%%%%%%%%%%%%%%%%%%%%%%%%%%%%%%%%%%%%%%%%%%%%%%%%%%%%%%%%%%%%%%%%%%%%%%%%
\\
%%%%%%%%%%%%%%%%%%%%%%%%%%%%%%%%%%%%%%%%%%%%%%%%%%%%%%%%%%%%%%%%%%%%%%%%%%%%%%%%%%%%%%%%%%%%%%%%%%%%%%%
>
1- 2 \big( {\chi}^{-96} + {\chi}^{-2} \big)
\label{eq:t1}
\end{multline}
%%%%%%%%%%%%%%%%%%%%%%%%%%%%%%%%%%%%%%%%%%%%%%%%%%%%%%%%%%%%%%%%%%%%%%%%%%%%%%%%%%%%%%%%%%%%%%%%%%%%%%%
The assumption of 
$\cK_j\geq 16a\ln \chi$,
is satisfied provided
$\opt_\#\geq 16a\chi \ln \chi$.
\eqref{eq:t1} 
implies our claim in (\emph{d})(\emph{ii})
using the technique in Section~\ref{sec-repeat}.

%%%%%%%%%%%%%%%%%%%%%%%%%%%%%%%%%%%%%%%%%%%%%%%%%%%%%%%%%%%%%%%%%%%%%%%%%%%%%%%%%%%%%%%%%%%%%%%%%%%%%%%
\subsection{\algsm: details and proofs of relevant claims in (\emph{a})--(\emph{c}) and (\emph{d})(\emph{iii})}
\label{sec-algsm1}
%%%%%%%%%%%%%%%%%%%%%%%%%%%%%%%%%%%%%%%%%%%%%%%%%%%%%%%%%%%%%%%%%%%%%%%%%%%%%%%%%%%%%%%%%%%%%%%%%%%%%%%

%%%%%%%%%%%%%%%%%%%%%%%%%%%%%%%%%%%%%%%%%%%%%%%%%%%%%%%%%%%%%%%%%%%%%%%%%%%%%%%%%%%%%%%%%%%%%%%%%%%%%%%
\medskip
\noindent
\textbf{Another modification of the $\LP$-relaxation in \FI{tab1}}
\medskip
%%%%%%%%%%%%%%%%%%%%%%%%%%%%%%%%%%%%%%%%%%%%%%%%%%%%%%%%%%%%%%%%%%%%%%%%%%%%%%%%%%%%%%%%%%%%%%%%%%%%%%%

Note that for this case $\chi=O(1)$. 
Fix an optimal solution for our instance of \fmc.
Let $\cA_i$ be the collection of those sets that contain \emph{at least} 
one element of color $i$
and let 
$Z_i=\sum_{\cS_\ell\in \cA_i} y_\ell$ 
indicate the number of sets from $\cA_i$ selected in an integral solution of the $\LP$; 
obviously $Z_i\geq 1$.
We consider two cases for $Z_i$ depending on whether it is at most $5\ln\chi$ or not.
We cannot know \emph{a priori} whether $Z_i\leq 5 \ln\chi$ or not.
However, for our analysis it suffices if we can guess \emph{just one} set belonging to $Z_i$ correctly.
We can do this by trying out all relevant possibilities exhaustively in the following manner. 
Let 
$\Psi=\{1,\dots,\chi\}$
be the set of indices of all colors. For each of the $2^\Psi-1$ subsets
$\Psi'$ of $\Psi$, we ``guess'' that 
$Z_i\leq 5 \ln\chi$ if and only if 
$i\in\Psi'$.
Of course, we still do not know one set among these $5\ln\chi$ subset
for each such $i$, so we will exhaustively try out one each of the at most $|\cA_i|\leq m$
sets for each $i$. 
For every such choice of $\Psi'$ and every such choice of 
a set $\cS_{i_{\Psi'}}\in\cA_i$ for each $i\in\Psi'$,  
we perform the following steps: 
%%%%%%%%%%%%%%%%%%%%%%%%%%%%%%%%%%%%%%%%%%%%%%%%%%%%%%%%%%%%%%%%%%%%%%%%%%%%%%%%%%%%%%%%%%%%%%%%%%%%%%%
\begin{enumerate}[label=$\triangleright$]
%%%%%%%%%%%%%%%%%%%%%%%%%%%%%%%%%%%%%%%%%%%%%%%%%%%%%%%%%%%%%%%%%%%%%%%%%%%%%%%%%%%%%%%%%%%%%%%%%%%%%%%
\item
Select the sets 
$\cS_{i_{\Psi'}}$ 
and their elements
for each $i\in\Psi'$
Set the variables corresponding to these sets and elements to $1$ in the $\LP$-relaxation
in \FI{tab1}, \IE, set $y_{i_{\Psi'}}=1$ and $x_j=1$ for every 
$i\in\Psi'\cS_{i_{\Psi'}}$ and 
$j\in$.
Remove any constraint that is already satisfied after the above step.
%%%%%%%%%%%%%%%%%%%%%%%%%%%%%%%%%%%%%%%%%%%%%%%%%%%%%%%%%%%%%%%%%%%%%%%%%%%%%%%%%%%%%%%%%%%%%%%%%%%%%%%
\item
Add the following additional (at most $\chi$) constraints to the $\LP$-relaxation:
%%%%%%%%%%%%%%%%%%%%%%%%%%%%%%%%%%%%%%%%%%%%%%%%%%%%%%%%%%%%%%%%%%%%%%%%%%%%%%%%%%%%%
\begin{center}
\begin{tabular}{l l}
$\sum_{ (u_\ell \in C_i)\wedge (u_\ell\in \cS_j) } y_j >  5\ln\chi$  & for $i\notin\Psi'$
%%%%%%%%%%%%%%%%%%%%%%%%%%%%%%%%%%%%%%%%%%%%%%%%%%%%%%%%%%%%%%%%%%%%%%%%%%%%%%%%%%%%%
\end{tabular}
\end{center}
%%%%%%%%%%%%%%%%%%%%%%%%%%%%%%%%%%%%%%%%%%%%%%%%%%%%%%%%%%%%%%%%%%%%%%%%%%%%%%%%%%%%%
\end{enumerate}
%%%%%%%%%%%%%%%%%%%%%%%%%%%%%%%%%%%%%%%%%%%%%%%%%%%%%%%%%%%%%%%%%%%%%%%%%%%%%%%%%%%%%%%%%%%%%%%%%%%%%%%
Note that the total number of iterations of the basic iterations that is needed is at most
$O( { (2m) }^{\chi})$, 
which is polynomial provided $\chi =O\big(\max\big\{1,\, \frac{\log n}{\log m}\big\}\big)$.

%%%%%%%%%%%%%%%%%%%%%%%%%%%%%%%%%%%%%%%%%%%%%%%%%%%%%%%%%%%%%%%%%%%%%%%%%%%%%%%%%%%%%%%%%%%%%%%%%%%%%%%
\bigskip
\noindent
\textbf{Analysis of the modified $\LP$-relaxation}
\medskip
%%%%%%%%%%%%%%%%%%%%%%%%%%%%%%%%%%%%%%%%%%%%%%%%%%%%%%%%%%%%%%%%%%%%%%%%%%%%%%%%%%%%%%%%%%%%%%%%%%%%%%%

We now analyze that iteration of the $\LP$-relaxation that correctly guesses the value of $\opt_\#$, 
the subset $\Psi'\subseteq\Psi$ and 
the sets $\cS_{i_{\Psi'}}\in\cA_i$ for each $i\in\Psi'$.
As already mentioned elsewhere, the random variables 
$x_1^+,\dots,x_n^+$ may \emph{not} be pairwise independent since two distinct elements belonging to the same set are
correlated, and consequently the random variables 
$p_1,\dots,p_\chi$ also may \emph{not} be pairwise independent.
For convenience, let 
$\mu_i=\Ave{x_i^+}$ and 
$\cE_i$ denote the event 
$\cE_i \equiv x_i^+=1$; note that $(1-\bee^{-1})x_i^\ast < \prob{\cE_i}=\mu_i\leq \min \{1,\, 2f x_i^\ast\}$.
We first calculate a bound on 
$\mbox{cov}(x_i^+,x_j^+)$ for all $i\neq j$ as follows. If  
$x_i^+$ and $x_j^+$
are independent then of course 
$\mbox{cov}(x_i^+,x_j^+)=0$, otherwise 
%%%%%%%%%%%%%%%%%%%%%%%%%%%%%%%%%%%%%%%%%%%%%%%%%%%%%%%%%%%%%%%%%%%%%%%%%%%%%%%%%%%%%%%%%%%%%%%%%%%%%%%
\begin{multline*}
- \min \left\{ \mu_i,\, \mu_j \right\}
\leq
- \prob {\cE_i} \prob{\cE_j} 
\leq
\prob{\cE_i \wedge \cE_j} 
- \prob {\cE_i} \prob{\cE_j} 
=
\Ave {x_i^+ x_j^+} - \mu_i \mu_j 
%%%%%%%%%%%%%%%%%%%%%%%%%%%%%%%%%%%%%%%%%%%%%%%%%%%%%%%%%%%%%%%%%%%%%%%%%%%%%%%%%%%%%%%%%%%%%%%%%%%%%%%
\\
%%%%%%%%%%%%%%%%%%%%%%%%%%%%%%%%%%%%%%%%%%%%%%%%%%%%%%%%%%%%%%%%%%%%%%%%%%%%%%%%%%%%%%%%%%%%%%%%%%%%%%%
=
\mbox{cov}(x_i^+,x_j^+)
\leq 
\prob {\cE_i \wedge \cE_j} - \mu_i \mu_j
\leq
\min \left\{ \mu_i,\, \mu_j \right\}
- \mu_i \mu_j
<
\min \left\{ \mu_i,\, \mu_j \right\}
\end{multline*}
%%%%%%%%%%%%%%%%%%%%%%%%%%%%%%%%%%%%%%%%%%%%%%%%%%%%%%%%%%%%%%%%%%%%%%%%%%%%%%%%%%%%%%%%%%%%%%%%%%%%%%%
giving the following bounds:
%%%%%%%%%%%%%%%%%%%%%%%%%%%%%%%%%%%%%%%%%%%%%%%%%%%%%%%%%%%%%%%%%%%%%%%%%%%%%%%%%%%%%%%%%%%%%%%%%%%%%%%
\begin{gather}
- \min \left\{ \varrho(f) x_i^\ast, \, \varrho(f) x_j^\ast \, \right\}
\leq
\mbox{cov}(x_i^+,x_j^+)
\leq
\min \left\{ 2 f x_i^\ast,\,  2 f x_j^\ast,\, 1 \right\}
\label{new1}
\end{gather}
%%%%%%%%%%%%%%%%%%%%%%%%%%%%%%%%%%%%%%%%%%%%%%%%%%%%%%%%%%%%%%%%%%%%%%%%%%%%%%%%%%%%%%%%%%%%%%%%%%%%%%%
%
For notational convenience, let 
$\cD_{i,j}=\{ \ell \,|\, u_i,u_j\in\cS_\ell,\, j\neq i \}$ 
be the indices of those sets in which both the elements $u_i$ and $u_j$ appear, and let 
$\cD_i = \cup_{j=1}^n \cD_{i,j}$. 
Note that 
$| \cD_{i,j} | \leq f$,  
$| \cD_i | \leq (a-1)f$,  
and 
the random variable $x_i^+$ is independent of 
all $x_j^+$ satisfying $j\notin\cD_i$.
Using this observation and~\eqref{new1}, 
for any $i \in \{1,\dots,n\}$ 
we get 
%%%%%%%%%%%%%%%%%%%%%%%%%%%%%%%%%%%%%%%%%%%%%%%%%%%%%%%%%%%%%%%%%%%%%%%%%%%%%%%%%%%%%%%%%%%%%%%%%%%%%%%
\begin{gather}
%%%%%%%%%%%%%%%%%%%%%%%%%%%%%%%%%%%%%%%%%%%%%%%%%%%%%%%%%%%%%%%%%%%%%%%%%%%%%%%%%%%%%%%%%%%%%%%%%%%%%%%
\sum_{j=1}^n\mbox{cov}(x_i^+,x_j^+)
\leq 
\sum_{j\in\cD_i} \!\! \big( \min \left\{ \mu_i,\,  \mu_j \right\} \big) 
\leq
|\cD_i| \, \mu_i
\leq
a f \mu_i
\leq
\min \big\{ 2af^2 x_i^\ast,\, af \big\}
\label{eq:resale1}
%%%%%%%%%%%%%%%%%%%%%%%%%%%%%%%%%%%%%%%%%%%%%%%%%%%%%%%%%%%%%%%%%%%%%%%%%%%%%%%%%%%%%%%%%%%%%%%%%%%%%%%
\end{gather}
%%%%%%%%%%%%%%%%%%%%%%%%%%%%%%%%%%%%%%%%%%%%%%%%%%%%%%%%%%%%%%%%%%%%%%%%%%%%%%%%%%%%%%%%%%%%%%%%%%%%%%%
\vspace*{-0.4in}
%%%%%%%%%%%%%%%%%%%%%%%%%%%%%%%%%%%%%%%%%%%%%%%%%%%%%%%%%%%%%%%%%%%%%%%%%%%%%%%%%%%%%%%%%%%%%%%%%%%%%%%
\begin{multline}
%%%%%%%%%%%%%%%%%%%%%%%%%%%%%%%%%%%%%%%%%%%%%%%%%%%%%%%%%%%%%%%%%%%%%%%%%%%%%%%%%%%%%%%%%%%%%%%%%%%%%%%
\sum_{j=1}^n\mbox{cov}(x_i^+,x_j^+)
\geq 
- \sum_{j\in\cD_i} \min \left\{  \mu_i,\,  \mu_j \right\}
\geq 
-  \min \Big\{  |\cD_i| \, \mu_i,\,  \varrho(f) \sum_{j\in \cD_i} x_j^\ast \Big\} 
%%%%%%%%%%%%%%%%%%%%%%%%%%%%%%%%%%%%%%%%%%%%%%%%%%%%%%%%%%%%%%%%%%%%%%%%%%%%%%%%%%%%%%%%%%%%%%%%%%%%%%%
\\
%%%%%%%%%%%%%%%%%%%%%%%%%%%%%%%%%%%%%%%%%%%%%%%%%%%%%%%%%%%%%%%%%%%%%%%%%%%%%%%%%%%%%%%%%%%%%%%%%%%%%%%
\geq 
- \min \Big\{  (a-1) f \,\mu_i, \, \varrho(f) \sum_{j\in \cD_i} x_j^\ast \Big\} 
>
- a f x_i^\ast
\label{eq:resale2}
%%%%%%%%%%%%%%%%%%%%%%%%%%%%%%%%%%%%%%%%%%%%%%%%%%%%%%%%%%%%%%%%%%%%%%%%%%%%%%%%%%%%%%%%%%%%%%%%%%%%%%%
\end{multline}
%%%%%%%%%%%%%%%%%%%%%%%%%%%%%%%%%%%%%%%%%%%%%%%%%%%%%%%%%%%%%%%%%%%%%%%%%%%%%%%%%%%%%%%%%%%%%%%%%%%%%%%
The above bounds can be used to bound 
the total pairwise co-variance between elements in two same or different color 
classes as follows. Consider two color classes $C_i$ and $C_j$ ($i=j$ is allowed). Then, 
%%%%%%%%%%%%%%%%%%%%%%%%%%%%%%%%%%%%%%%%%%%%%%%%%%%%%%%%%%%%%%%%%%%%%%%%%%%%%%%%%%%%%%%%%%%%%%%%%%%%%%%
\begin{multline}
%%%%%%%%%%%%%%%%%%%%%%%%%%%%%%%%%%%%%%%%%%%%%%%%%%%%%%%%%%%%%%%%%%%%%%%%%%%%%%%%%%%%%%%%%%%%%%%%%%%%%%%
\sum_{u_r\in C_i} \sum_{u_s\in C_j} \mbox{cov}(x_r^+,x_s^+)
\leq
\sum_{u_r\in C_i} \sum_{j=1}^n \mbox{cov}(x_r^+,x_j^+)
\leq
\sum_{u_r\in C_i} 
\min \big\{ 2af^2 x_r^\ast,\, af \big\}
%%%%%%%%%%%%%%%%%%%%%%%%%%%%%%%%%%%%%%%%%%%%%%%%%%%%%%%%%%%%%%%%%%%%%%%%%%%%%%%%%%%%%%%%%%%%%%%%%%%%%%%
\\
%%%%%%%%%%%%%%%%%%%%%%%%%%%%%%%%%%%%%%%%%%%%%%%%%%%%%%%%%%%%%%%%%%%%%%%%%%%%%%%%%%%%%%%%%%%%%%%%%%%%%%%
\textstyle
=
\min \Big\{ 2af^2 \sum_{u_r\in C_i} x_r^\ast,\, af |C_i| \Big\}
\leq
\min \big\{ 2af^2 \frac{\opt_\#}{\chi} ,\, af n \big\}
\label{eq:cor1}
%%%%%%%%%%%%%%%%%%%%%%%%%%%%%%%%%%%%%%%%%%%%%%%%%%%%%%%%%%%%%%%%%%%%%%%%%%%%%%%%%%%%%%%%%%%%%%%%%%%%%%%
\end{multline}
%%%%%%%%%%%%%%%%%%%%%%%%%%%%%%%%%%%%%%%%%%%%%%%%%%%%%%%%%%%%%%%%%%%%%%%%%%%%%%%%%%%%%%%%%%%%%%%%%%%%%%%
%%%%%%%%%%%%%%%%%%%%%%%%%%%%%%%%%%%%%%%%%%%%%%%%%%%%%%%%%%%%%%%%%%%%%%%%%%%%%%%%%%%%%%%%%%%%%%%%%%%%%%%
\vspace*{-0.2in}
%%%%%%%%%%%%%%%%%%%%%%%%%%%%%%%%%%%%%%%%%%%%%%%%%%%%%%%%%%%%%%%%%%%%%%%%%%%%%%%%%%%%%%%%%%%%%%%%%%%%%%%
\begin{gather}
%%%%%%%%%%%%%%%%%%%%%%%%%%%%%%%%%%%%%%%%%%%%%%%%%%%%%%%%%%%%%%%%%%%%%%%%%%%%%%%%%%%%%%%%%%%%%%%%%%%%%%%
\sum_{u_r\in C_i} \sum_{u_s\in C_j} \mbox{cov}(x_r^+,x_s^+)
\geq
- \sum_{u_r\in C_i} \sum_{j\in \cD_r} \min \{ \mu_r,\,\mu_j \}
>
- \sum_{u_r\in C_i} a f x_r^\ast
=
\textstyle
-a f \frac{\opt_\#}{\chi}
\label{eq:cor2}
%%%%%%%%%%%%%%%%%%%%%%%%%%%%%%%%%%%%%%%%%%%%%%%%%%%%%%%%%%%%%%%%%%%%%%%%%%%%%%%%%%%%%%%%%%%%%%%%%%%%%%%
\end{gather}
%%%%%%%%%%%%%%%%%%%%%%%%%%%%%%%%%%%%%%%%%%%%%%%%%%%%%%%%%%%%%%%%%%%%%%%%%%%%%%%%%%%%%%%%%%%%%%%%%%%%%%%
%
For calculations of probabilities of events of the form ``$p_i>\Delta p_j$'', 
we first need to bound the probability of events ``$p_j=0$'' for 
$j\in \{1,\dots,\chi \}$. 
If $j\in \Psi'$ then 
$\prob{p_j=0}=0$ since at least one set containing an element of color $j$ is always selected.
Otherwise, 
$\sum_{\cS_\ell\in \cA_j} y_\ell^\ast>5\ln\chi$, and 
$p_j=0$ if and only if $y_\ell^+=0$ for \emph{every} $\cS_\ell\in \cA_j$. 
This gives us the following bound for $j\notin\Psi'$:
%%%%%%%%%%%%%%%%%%%%%%%%%%%%%%%%%%%%%%%%%%%%%%%%%%%%%%%%%%%%%%%%%%%%%%%%%%%%%%%%%%%%%%%%%%%%%%%%%%%%%%%
\begin{multline*}
%%%%%%%%%%%%%%%%%%%%%%%%%%%%%%%%%%%%%%%%%%%%%%%%%%%%%%%%%%%%%%%%%%%%%%%%%%%%%%%%%%%%%%%%%%%%%%%%%%%%%%%
\prob{p_j=0}=
\prod_{\cS_\ell\in \cA_j} \prob{y_\ell^+=0} 
=
\prod_{\cS_\ell\in \cA_j} (1 - y_\ell^\ast) 
\leq
\prod_{\cS_\ell\in \cA_j} \bee^{- y_\ell^\ast}
\\
= \bee^{- \sum_{\cS_\ell\in \cA_j} y_\ell^\ast }
\leq
\bee^{- 5\ln\chi}
=
\chi^{-5}
%%%%%%%%%%%%%%%%%%%%%%%%%%%%%%%%%%%%%%%%%%%%%%%%%%%%%%%%%%%%%%%%%%%%%%%%%%%%%%%%%%%%%%%%%%%%%%%%%%%%%%%
\end{multline*}
%%%%%%%%%%%%%%%%%%%%%%%%%%%%%%%%%%%%%%%%%%%%%%%%%%%%%%%%%%%%%%%%%%%%%%%%%%%%%%%%%%%%%%%%%%%%%%%%%%%%%%%
Combining both cases, we have 
$\prob{p_j=0} \leq \nicefrac{1}{\chi^{5}}$ for all $j$.

We now can calculate the 
probabilities of events of the form ``$p_i>\Delta p_j$''
for $\Delta=\Delta_1 + \Delta_2 \geq 1$, $\Delta_1,\Delta_2 \geq 0$ and 
$i,j \in \{1,\dots,\chi \}$ as follows:
%%%%%%%%%%%%%%%%%%%%%%%%%%%%%%%%%%%%%%%%%%%%%%%%%%%%%%%%%%%%%%%%%%%%%%%%%%%%%%%%%%%%%%%%%%%%%%%%%%%%%%%
\begin{multline}
\prob{ p_i\geq \Delta p_j }
\leq
\prob{ p_j = 0} + \prob{ p_i \geq \Delta p_j \,|\, p_j\geq 1}
\leq
\chi^{-5}
+ \prob{ p_i \geq \Delta_1 p_j + \Delta_2 \,|\, p_j\geq 1}
%%%%%%%%%%%%%%%%%%%%%%%%%%%%%%%%%%%%%%%%%%%%%%%%%%%%%%%%%%%%%%%%%%%%%%%%%%%%%%%%%%%%%%%%%%%%%%%%%%%%%%%
\\
%%%%%%%%%%%%%%%%%%%%%%%%%%%%%%%%%%%%%%%%%%%%%%%%%%%%%%%%%%%%%%%%%%%%%%%%%%%%%%%%%%%%%%%%%%%%%%%%%%%%%%%
=
\chi^{-5}
+ \frac{ \prob{ ( p_i \geq \Delta_1 p_j + \Delta_2 )\wedge ( p_j\geq 1 ) } } { \prob{  p_j\geq 1 }  }
< 
\chi^{-5}
+ \frac{ \prob{ p_i \geq \Delta_1 p_j + \Delta_2 } } { 1 - \prob{  p_j=0 }  }
%%%%%%%%%%%%%%%%%%%%%%%%%%%%%%%%%%%%%%%%%%%%%%%%%%%%%%%%%%%%%%%%%%%%%%%%%%%%%%%%%%%%%%%%%%%%%%%%%%%%%%%
\\
%%%%%%%%%%%%%%%%%%%%%%%%%%%%%%%%%%%%%%%%%%%%%%%%%%%%%%%%%%%%%%%%%%%%%%%%%%%%%%%%%%%%%%%%%%%%%%%%%%%%%%%
<
\chi^{-5}
+
\frac{ \prob{ p_i \geq \Delta_1 p_j + \Delta_2 } } { 1 -   
\chi^{-5}
}
\label{gen1}
\end{multline}
%%%%%%%%%%%%%%%%%%%%%%%%%%%%%%%%%%%%%%%%%%%%%%%%%%%%%%%%%%%%%%%%%%%%%%%%%%%%%%%%%%%%%%%%%%%%%%%%%%%%%%%
%%
For a real number $\zeta>0$, let $\delta_{i,j}=p_i - \zeta p_j$. 
We have the following bound on $\Ave{\delta_{i,j}}$ for all $\zeta\geq 3f$: 
%%%%%%%%%%%%%%%%%%%%%%%%%%%%%%%%%%%%%%%%%%%%%%%%%%%%%%%%%%%%%%%%%%%%%%%%%%%%%%%%%%%%%%%%%%%%%%%%%%%%%%%
\begin{gather*}
\Ave{\delta_{i,j}} = \Ave{p_i} - \zeta\,\Ave{p_j}  \leq 2f \frac{\opt_\#}{\chi} - \zeta \varrho(f) \frac{\opt_\#}{\chi} <0
\end{gather*}
%%%%%%%%%%%%%%%%%%%%%%%%%%%%%%%%%%%%%%%%%%%%%%%%%%%%%%%%%%%%%%%%%%%%%%%%%%%%%%%%%%%%%%%%%%%%%%%%%%%%%%%
Therefore, using Chebyshev's inequality
we get (for all $\zeta\geq 3f$ and $\lambda>1$):
%%%%%%%%%%%%%%%%%%%%%%%%%%%%%%%%%%%%%%%%%%%%%%%%%%%%%%%%%%%%%%%%%%%%%%%%%%%%%%%%%%%%%%%%%%%%%%%%%%%%%%%
\begin{multline}
\prob{ p_i \geq \zeta p_j + \lambda \sqrt { \mbox{var} (\delta_{i,j})} }
=
\prob{ \delta_{i,j} \geq \lambda \sqrt { \mbox{var} (\delta_{i,j})} }
\\
<
\prob{ \left|\,\delta_{i,j} - \Ave{\delta_{i,j} } \,\right| > \lambda \sqrt { \mbox{var} (\delta_{i,j})} } 
\leq
\nicefrac{1}{\lambda^2}
\label{nv1}
\end{multline}
%%%%%%%%%%%%%%%%%%%%%%%%%%%%%%%%%%%%%%%%%%%%%%%%%%%%%%%%%%%%%%%%%%%%%%%%%%%%%%%%%%%%%%%%%%%%%%%%%%%%%%%
Using~\eqref{nv1} in~\eqref{gen1} with $\Delta_1=\zeta$, 
$\Delta_2= \lambda \sqrt { \mbox{var} (\delta_{i,j})}$
and $\lambda=10\chi$
we get
%%%%%%%%%%%%%%%%%%%%%%%%%%%%%%%%%%%%%%%%%%%%%%%%%%%%%%%%%%%%%%%%%%%%%%%%%%%%%%%%%%%%%%%%%%%%%%%%%%%%%%%
\begin{multline}
\textstyle
\prob{ p_i < \left(\zeta + 10 \chi \sqrt { \mbox{var} (\delta_{i,j})} \, \right) p_j }
=
1 - 
\prob{ p_i \geq \left(\zeta + 10 \chi \sqrt { \mbox{var} (\delta_{i,j})} \, \right) p_j }
%%%%%%%%%%%%%%%%%%%%%%%%%%%%%%%%%%%%%%%%%%%%%%%%%%%%%%%%%%%%%%%%%%%%%%%%%%%%%%%%%%%%%%%%%%%%%%%%%%%%%%%
\\
%%%%%%%%%%%%%%%%%%%%%%%%%%%%%%%%%%%%%%%%%%%%%%%%%%%%%%%%%%%%%%%%%%%%%%%%%%%%%%%%%%%%%%%%%%%%%%%%%%%%%%%
\textstyle
>
1- \chi^{-5}
-
\frac{\chi^{-100}}{1-\chi^{-5}}
> 
1-\chi^{-4}
\label{gen2}
\end{multline}
%%%%%%%%%%%%%%%%%%%%%%%%%%%%%%%%%%%%%%%%%%%%%%%%%%%%%%%%%%%%%%%%%%%%%%%%%%%%%%%%%%%%%%%%%%%%%%%%%%%%%%%
We now calculate a bound on 
$\mbox{var} (\delta_{i,j})$ 
using~\eqref{eq:cor1}~and~\eqref{eq:cor2}
as follows:
%%%%%%%%%%%%%%%%%%%%%%%%%%%%%%%%%%%%%%%%%%%%%%%%%%%%%%%%%%%%%%%%%%%%%%%%%%%%%%%%%%%%%%%%%%%%%%%%%%%%%%%
\begin{multline}
\mbox{var} (\delta_{i,j}) = 
\mbox{var} ( p_i - \zeta p_j ) = 
\mbox{var} \Big( \sum_{u_\ell \in C_i} \!\! x_\ell^+ +  \sum_{u_\ell \in C_j} \!\! (-\zeta x_\ell^+)  \Big)
%%%%%%%%%%%%%%%%%%%%%%%%%%%%%%%%%%%%%%%%%%%%%%%%%%%%%%%%%%%%%%%%%%%%%%%%%%%%%%%%%%%%%%%%%%%%%%%%%%%%%%%
\\
%%%%%%%%%%%%%%%%%%%%%%%%%%%%%%%%%%%%%%%%%%%%%%%%%%%%%%%%%%%%%%%%%%%%%%%%%%%%%%%%%%%%%%%%%%%%%%%%%%%%%%%
=
\sum_{u_\ell \in C_i} \!\! \mbox{var}\big( x_\ell^+ \big)
+ \zeta^2 \sum_{u_\ell \in C_j} \!\! \mbox{var}\big( x_\ell^+ \big)
+ \!\!\!\!\!\! \sum_{u_r,u_s\in C_i, r\neq s} \!\!\!\!\!\!\!\!\! \mbox{cov}(x_r^+,x_s^+)
+ \!\!\!\!\!\! \sum_{u_r,u_s\in C_j, r\neq s} \!\!\!\!\!\!\!\!\! \mbox{cov}(-\zeta x_r^+,-\zeta x_s^+)
\\
+ \!\!\!\!\!\! \sum_{u_r\in C_i,u_s\in C_j} \!\!\!\!\!\!\!\!\! \mbox{cov}(x_r^+,-\zeta x_s^+)
%%%%%%%%%%%%%%%%%%%%%%%%%%%%%%%%%%%%%%%%%%%%%%%%%%%%%%%%%%%%%%%%%%%%%%%%%%%%%%%%%%%%%%%%%%%%%%%%%%%%%%%
\\
%%%%%%%%%%%%%%%%%%%%%%%%%%%%%%%%%%%%%%%%%%%%%%%%%%%%%%%%%%%%%%%%%%%%%%%%%%%%%%%%%%%%%%%%%%%%%%%%%%%%%%%
\leq
\sum_{u_\ell \in C_i} \!\! \mu_\ell
+ \zeta^2 \sum_{u_\ell \in C_j} \!\! \mu_\ell
+ 2 a f^2 {\textstyle \frac{\opt_\#}{\chi} }
+ \zeta^2 \!\!\!\!\!\! \sum_{u_r,u_s\in C_j, r\neq s} \!\!\!\!\!\!\!\!\! \mbox{cov}(x_r^+,x_s^+)
- \zeta \!\!\!\!\!\! \sum_{u_r\in C_i,u_s\in C_j} \!\!\!\!\!\!\!\!\! \mbox{cov}(x_r^+,x_s^+)
%%%%%%%%%%%%%%%%%%%%%%%%%%%%%%%%%%%%%%%%%%%%%%%%%%%%%%%%%%%%%%%%%%%%%%%%%%%%%%%%%%%%%%%%%%%%%%%%%%%%%%%
\\
%%%%%%%%%%%%%%%%%%%%%%%%%%%%%%%%%%%%%%%%%%%%%%%%%%%%%%%%%%%%%%%%%%%%%%%%%%%%%%%%%%%%%%%%%%%%%%%%%%%%%%%
\leq
\Ave{c_i} 
+ \zeta^2 \Ave{c_j}
+ 2 a f^2 {\textstyle \frac{\opt_\#}{\chi} }
+ 2 \zeta^2 a f^2 {\textstyle \frac{\opt_\#}{\chi} }
+ \zeta a f {\textstyle \frac{\opt_\#}{\chi} }
%%%%%%%%%%%%%%%%%%%%%%%%%%%%%%%%%%%%%%%%%%%%%%%%%%%%%%%%%%%%%%%%%%%%%%%%%%%%%%%%%%%%%%%%%%%%%%%%%%%%%%%
\\
%%%%%%%%%%%%%%%%%%%%%%%%%%%%%%%%%%%%%%%%%%%%%%%%%%%%%%%%%%%%%%%%%%%%%%%%%%%%%%%%%%%%%%%%%%%%%%%%%%%%%%%
\leq
2 f {\textstyle \frac{\opt_\#}{\chi} }
   + \zeta^2 2 f {\textstyle \frac{\opt_\#}{\chi} }
   + 2 a f^2 {\textstyle \frac{\opt_\#}{\chi} }
   + \zeta^2 2 a f^2 {\textstyle \frac{\opt_\#}{\chi} }
   + \zeta a f {\textstyle \frac{\opt_\#}{\chi} }
\leq 
4 \zeta^2  a f^2 {\textstyle \frac{\opt_\#}{\chi} }
%%%%%%%%%%%%%%%%%%%%%%%%%%%%%%%%%%%%%%%%%%%%%%%%%%%%%%%%%%%%%%%%%%%%%%%%%%%%%%%%%%%%%%%%%%%%%%%%%%%%%%%
\\
%%%%%%%%%%%%%%%%%%%%%%%%%%%%%%%%%%%%%%%%%%%%%%%%%%%%%%%%%%%%%%%%%%%%%%%%%%%%%%%%%%%%%%%%%%%%%%%%%%%%%%%
\Rightarrow \,
\sqrt{\mbox{var} (\delta_{i,j})}
\leq
2 \zeta  \sqrt{a} f \sqrt{\opt_\#/\chi}
\label{eq:gen3}
%%%%%%%%%%%%%%%%%%%%%%%%%%%%%%%%%%%%%%%%%%%%%%%%%%%%%%%%%%%%%%%%%%%%%%%%%%%%%%%%%%%%%%%%%%%%%%%%%%%%%%%
\end{multline}
%%%%%%%%%%%%%%%%%%%%%%%%%%%%%%%%%%%%%%%%%%%%%%%%%%%%%%%%%%%%%%%%%%%%%%%%%%%%%%%%%%%%%%%%%%%%%%%%%%%%%%%
Setting $\zeta=3f$ and 
using~\eqref{eq:gen3} in~\eqref{gen2}
we get
%%%%%%%%%%%%%%%%%%%%%%%%%%%%%%%%%%%%%%%%%%%%%%%%%%%%%%%%%%%%%%%%%%%%%%%%%%%%%%%%%%%%%%%%%%%%%%%%%%%%%%%
$\prob{ c_i < \left(3f +  60 \sqrt{a} f^2 \sqrt{\opt_\# \chi} \, \right) c_j } > 1-\chi^{-4}$.
%%%%%%%%%%%%%%%%%%%%%%%%%%%%%%%%%%%%%%%%%%%%%%%%%%%%%%%%%%%%%%%%%%%%%%%%%%%%%%%%%%%%%%%%%%%%%%%%%%%%%%%
This implies our claim in (\emph{d})(\emph{iii})
using the technique in Section~\ref{sec-repeat}.

%%%%%%%%%%%%%%%%%%%%%%%%%%%%%%%%%%%%%%%%%%%%%%%%%%%%%%%%%%%%%%%%%%%%%%%%%%%%%%%%%%%%%%%%%%%%%%%%%%%%%%%
\subsection{Limitations of our $\LP$-relaxation: ``a gap of factor $f$'' for coloring constraints}
\label{sec-limit-lp}
%%%%%%%%%%%%%%%%%%%%%%%%%%%%%%%%%%%%%%%%%%%%%%%%%%%%%%%%%%%%%%%%%%%%%%%%%%%%%%%%%%%%%%%%%%%%%%%%%%%%%%%

The coloring constraint bounds in Theorem~\ref{thm-main}(\emph{e})(\emph{i})--(\emph{ii})
depend on $f$ or $f^2$ only.
It is natural to ask as a possible first direction of improvement whether this dependence
can be eliminated or improved by better analysis of our $\LP$-relaxations. 
Proposition~\ref{prop1}
shows that this may \emph{not} be possible even for $\chi=2$ 
unless one uses a significantly different $\LP$-relaxation for \fmcg{\chi}{k}.

%%%%%%%%%%%%%%%%%%%%%%%%%%%%%%%%%%%%%%%%%%%%%%%%%%%%%%%%%%%%%%%%%%%%%%%%%%%%%%%%%%%%%%%%%%%%%%%%%%%%%%%
\begin{proposition}\label{prop1}
There exists optimal non-integral solutions of \fmc with the following property: 
any rounding approach that does not change the values of zero-valued variables in the fractional solution 
must necessarily result in an integral solutions in which the color constraints differ by at least a factor of $f$.
\end{proposition}
%%%%%%%%%%%%%%%%%%%%%%%%%%%%%%%%%%%%%%%%%%%%%%%%%%%%%%%%%%%%%%%%%%%%%%%%%%%%%%%%%%%%%%%%%%%%%%%%%%%%%%%

\begin{proof}
We will show our result for the $\LP$-relaxation in 
\FI{tab1}; proofs for other modified versions of this $\LP$-relaxation are similar. 
Consider $\alpha\gg 1$ disjoint collections of sets and elements of the following type: 
for $j\in\{1,\dots,\alpha\}$, the $j\tx$ collection consists of a set of $\alpha+1$ elements 
$\cU^j=\{u_1^j,\dots,u_{\alpha+1}^j\}$ with 
$\cC(u_1^j)=1$ and 
$\cC(u_2^j)=\dots=\cC(u_{\alpha+1}^j)=2$,
and the $\alpha+1$ sets 
$\cS_1^j,\dots,\cS_{\alpha+1}^j$ where 
$\cS_i^j=\cU^j \setminus \{u_i^j\}$ for $i\in\{1,\dots,\alpha+1\}$
(note that each element $u_i^j$ is in exactly $\alpha$ sets). 
Add to these collections the additional 
$2\alpha+2$ elements 
$u_1^\ell,u_2^\ell$ 
with 
$\cC(u_1^\ell)=1$ and $\cC(u_2^\ell)=2$, 
and the $\alpha+1$ sets $\cS^\ell=\{u_1^\ell,u_2^\ell\}$
for $\ell\in\{\alpha+1,\dots,2\alpha+1\}$. 
Note that for our created instance $f=\alpha$.
Consider the following two different 
solutions of the $\LP$-relaxation:
%%%%%%%%%%%%%%%%%%%%%%%%%%%%%%%%%%%%%%%%%%%%%%%%%%%%%%%%%%%%%%%%%%%%%%%%%%%%%%%%%%%%%%%%%%%%%%%%%%%%%%%
\begin{description}
\item[(1)]
For a non-integral solution, let 
$y_1^j=\dots=y_{\alpha+1}^j=\nicefrac{1}{\alpha}$,
let $x_1^j=1$ and let $x_2^j=\dots=x_{\alpha+1}^j=\nicefrac{1}{\alpha}$ 
for $j\in\{1,\dots,\alpha\}$, and set all other variables to zero.
This results in a solution with summation of set variables being $\alpha+1$ 
(\IE, $\alpha+1$ sets are selected non-integrally),
and summation of element variables being 
$2\alpha+2$ (\IE, $2\alpha+2$ elements are selected non-integrally).
Moreover, the summation of element variables with the color value of $1$ is precisely the same as summation 
of element variables with the color value of $2$ since both are equal to $\alpha+1$. 
%%%%%%%%%%%%%%%%%%%%%%%%%%%%%%%%%%%%%%%%%%%%%%%%%%%%%%%%%%%%%%%%%%%%%%%%%%%%%%%%%%%%%%%%%%%%%%%%%%%%%%%
\item[(2)]
For an integral solution, let 
$y^\ell=x_1^\ell=x_2^\ell=1$ for 
$\ell\in\{\alpha+1,\dots,2\alpha+1\}$. 
This also results in a solution in which $\alpha+1$ sets are selected, the number of elements 
covered is $2\alpha+2$ and the number of elements of each color is $\alpha+1$.
\end{description}
%%%%%%%%%%%%%%%%%%%%%%%%%%%%%%%%%%%%%%%%%%%%%%%%%%%%%%%%%%%%%%%%%%%%%%%%%%%%%%%%%%%%%%%%%%%%%%%%%%%%%%%
The crucial things to note here is that the two above solutions are disjoint (\IE, non-zero 
variables in one solution are zero in the other and \emph{vice versa}), and thus 
any rounding approach for the solution in \textbf{(1)} that does not change values of the zero-valued variables 
results in an integral solution in which the number of elements of color $2$ is $f$ times 
the number of elements of color $1$. 
\end{proof}

%%%%%%%%%%%%%%%%%%%%%%%%%%%%%%%%%%%%%%%%%%%%%%%%%%%%%%%%%%%%%%%%%%%%%%%%%%%%%%%
%%%%%%%%%%%%%%%%%%%%%%%%%%%%%%%%%%%%%%%%%%%%%%%%%%%%%%%%%%%%%%%%%%%%%%%%%%%%%%%
%%%%%%%%%%%%%%%%%%%%%%%%%%%%%%%%%%%%%%%%%%%%%%%%%%%%%%%%%%%%%%%%%%%%%%%%%%%%%%%

%%%%%%%%%%%%%%%%%%%%%%%%%%%%%%%%%%%%%%%%%%%%%%%%%%%%%%%%%%%%%%%%%%%%%%%%%%%%%%%%%%%%%%%%%%%%%%%%%%%%%%%
\section{A tale of fewer colors: deterministic approximation for \fmc when $\chi$ is ``not too large''}
\label{sec-nfmc-det}
%%%%%%%%%%%%%%%%%%%%%%%%%%%%%%%%%%%%%%%%%%%%%%%%%%%%%%%%%%%%%%%%%%%%%%%%%%%%%%%%%%%%%%%%%%%%%%%%%%%%%%%

In this section we provide polynomial-time deterministic approximations of \fmc via the \emph{iterated rounding} technique 
for $\LP$-relaxations. 
We assume that the reader is familiar with the basic concepts related to this approach as described, 
for example, in~\cite{LRS11}. 
Our approximation qualities will depend on the parameters $f$ and $\chi$ and the coloring constraint bounds 
are interesting only if $\chi$ is not too large, \EG, no more than, say, poly-logarithmic in $n$.
\emph{For better understanding of the idea, we will first consider the special case 
\nfmc of \fmc for which $f=2$, and later on describe how to adopt the same approach for arbitrary $f$}.
As per the proof of Theorem~\ref{thm-main} (see Section~\ref{sec-str})
we may assume we know the value of $\nopt$ exactly.
A main ingredient of the iterated rounding approach is the following 
``rank lemma''.

%%%%%%%%%%%%%%%%%%%%%%%%%%%%%%%%%%%%%%%%%%%%%%%%%%%%%%%%%%%%%%%%%%%%%%%%%%%%%%%%%%%%%%%%%%%%%%%%%%%%%%%
\begin{fact}[Rank lemma]{\em\cite[Lemma 2.1.4]{LRS11}}\label{fact2}
Consider any convex polytope 
$P\eqdef \{\x\in\R^n \,|\, \mathbf{A}_j \x \geq b_j \text{ for } j\in\{1,\dots,m\}, \, \x\geq \mathbf{0}\}$
for some $\mathbf{A}_1,\dots,\mathbf{A}_m\in\R^{n}$ and $(b_1,\dots,b_m)^{\mathrm{T}}\in\R^m$. 
Then the following property holds for every extreme-point for $P$: the number of 
any maximal set of linearly independent tight constraints $($\emph{\IE}, constraints satisfying  
$\mathbf{A}_j \x = b_j$ for some $j)$ in this solution 
equals the number of non-zero variables.
\end{fact}
%%%%%%%%%%%%%%%%%%%%%%%%%%%%%%%%%%%%%%%%%%%%%%%%%%%%%%%%%%%%%%%%%%%%%%%%%%%%%%%%%%%%%%%%%%%%%%%%%%%%%%%

%%%%%%%%%%%%%%%%%%%%%%%%%%%%%%%%%%%%%%%%%%%%%%%%%%%%%%%%%%%%%%%%%%%%%%%%%%%%%%%%%%%%%%%%%%%%%%%%%%%%%%%
\subsection{Approximating \nfmc}
\label{sec-nfmc-con}
%%%%%%%%%%%%%%%%%%%%%%%%%%%%%%%%%%%%%%%%%%%%%%%%%%%%%%%%%%%%%%%%%%%%%%%%%%%%%%%%%%%%%%%%%%%%%%%%%%%%%%%

\begin{theorem}\label{thm-nfmc-con}
We can design a deterministic polynomial-time approximation algorithm \algiter for \nfmc with the following 
properties: 
%
%%%%%%%%%%%%%%%%%%%%%%%%%%%%%%%%%%%%%%%%%%%%%%%%%%%%%%%%%%%%%%%%%%%%%%%%%%%%%%%%%%%%%%%%%%%%%%%%%%%%%%%
\begin{enumerate}[label=\textbf{\emph{(}\alph*\emph{)}},leftmargin=*]
%%%%%%%%%%%%%%%%%%%%%%%%%%%%%%%%%%%%%%%%%%%%%%%%%%%%%%%%%%%%%%%%%%%%%%%%%%%%%%%%%%%%%%%%%%%%%%%%%%%%%%%
\item
The algorithm selects 
$\tau$ nodes where 
$
\tau \leq 
\begin{cases}
k+\frac{\chi-1}{2}, & \mbox{if $\chi=O(1)$}
\\
k +\chi-1, & \mbox{otherwise}
\end{cases}
$
%%%%%%%%%%%%%%%%%%%%%%%%%%%%%%%%%%%%%%%%%%%%%%%%%%%%%%%%%%%%%%%%%%%%%%%%%%%%%%%%%%%%%%%%%%%%%%%%%%%%%%%
\item
The algorithm
is a   
$\frac{1}{2}$-approximation for \nfmc, \emph{\IE}, 
the total weight of the selected elements 
is at least $\nicefrac{\opt}{2}$.
%%%%%%%%%%%%%%%%%%%%%%%%%%%%%%%%%%%%%%%%%%%%%%%%%%%%%%%%%%%%%%%%%%%%%%%%%%%%%%%%%%%%%%%%%%%%%%%%%%%%%%%
\item
The algorithm
satisfies
the $\eps$-approximate coloring constraints $($cf.\ Inequality~\eqref{eq1}$)$
as follows:
%%%%%%%%%%%%%%%%%%%%%%%%%%%%%%%%%%%%%%%%%%%%%%%%%%%%%%%%%%%%%%%%%%%%%%%%%%%%%%%%%%%%%%%%%%%%%%%%%%%%%%%
\begin{quote}
for all 
$i,j\in\{1,\dots,\chi\}$, 
$\frac{p_i}{p_j}<
\begin{cases}
4+4\chi,  & \mbox{if $\chi=O(1)$}
\\
4+2\chi+4\chi^2, & \mbox{otherwise}
\end{cases}
$
\end{quote}
%%%%%%%%%%%%%%%%%%%%%%%%%%%%%%%%%%%%%%%%%%%%%%%%%%%%%%%%%%%%%%%%%%%%%%%%%%%%%%%%%%%%%%%%%%%%%%%%%%%%%%%
\end{enumerate}
%%%%%%%%%%%%%%%%%%%%%%%%%%%%%%%%%%%%%%%%%%%%%%%%%%%%%%%%%%%%%%%%%%%%%%%%%%%%%%%%%%%%%%%%%%%%%%%%%%%%%%%
\end{theorem}

We discuss the proof in the rest of this section.
Let $G=(V,E)$ be the given graph, and let $\deg(v)$ denote the degree of node $v$.
Assume that $G$ has no isolated nodes.

%%%%%%%%%%%%%%%%%%%%%%%%%%%%%%%%%%%%%%%%%%%%%%%%%%%%%%%%%%%%%%%%%%%%%%%%%%%%%%%%%%%%%%%%%%%%%%%%%%%%%%%
\subsubsection{The case of $\chi=O(1)$} 
%%%%%%%%%%%%%%%%%%%%%%%%%%%%%%%%%%%%%%%%%%%%%%%%%%%%%%%%%%%%%%%%%%%%%%%%%%%%%%%%%%%%%%%%%%%%%%%%%%%%%%%

Since the problem can be exactly solved in polynomial time by exhaustive enumeration if $k$ is a constant, 
we can assume $k$ is at least a sufficiently large integer, \EG, assume that $k>10\chi$.

%%%%%%%%%%%%%%%%%%%%%%%%%%%%%%%%%%%%%%%%%%%%%%%%%%%%%%%%%%%%%%%%%%%%%%%%%%%%%%%%%%%%%%%%%%%%%%%%%%%%%%%
\medskip
\noindent
\textbf{Initial preprocessing}
\smallskip
%%%%%%%%%%%%%%%%%%%%%%%%%%%%%%%%%%%%%%%%%%%%%%%%%%%%%%%%%%%%%%%%%%%%%%%%%%%%%%%%%%%%%%%%%%%%%%%%%%%%%%%

To begin, we ``guess'' $\chi+1$ nodes, say $v_1,\dots,v_{\chi+1}\in V$ 
with $\deg(v_1)\leq \deg(v_2)\leq\dots\leq \deg(v_{\chi+1})$ such that 
there exists an optimal solution contains these $\chi+1$ nodes 
with the following property: ``the remaining $k-(\chi+1)$ nodes in the solution have degree \emph{at most} $\deg(v_1)$''.
Since there are at most $\binom{n}{\chi+1}=n^{O(1)}$ choices for such $\chi+1$ nodes, we can try them out in an 
\emph{exhaustive} fashion. 
Thus, we only need to analyze that run of our algorithm where the our guess is \emph{correct}.
Once these $\chi+1$ nodes have been selected, we will 
use the following sets of nodes in $\Vhat$ and ''incidence-indexed'' edges $\Ehat$ as input 
to our algorithm (note that an edge $e\eqdef \{u,v\}$ may appear as two members $(e,u)$ and $(e,v)$ in $\Ehat$ if 
both $u$ and $v$ are in $\Vhat$):  
%%%%%%%%%%%%%%%%%%%%%%%%%%%%%%%%%%%%%%%%%%%%%%%%%%%%%%%%%%%%%%%%%%%%%%%%%%%%%%%%%%%%%%%%%%%%%%%%%%%%%%%
\begin{gather*}
\Vhat = V \setminus 
\big ( \, \{ v_1,\dots,v_{\chi+1} \} \,\cup\, \{ v \,|\, \text{degree of $v$ in $G$ is strictly larger than $\deg(v_{1})$ } \} \, \big)
\\
\Ehat =  \{ \, (e,u) \,|\, \big( e\eqdef\{u,v\} \in E \big) \bigwedge \big( v\notin \{ v_1,\dots,v_{\chi+1} \} \big) \bigwedge 
            \big(  u \in \Vhat \big) \, \}
\end{gather*}
%%%%%%%%%%%%%%%%%%%%%%%%%%%%%%%%%%%%%%%%%%%%%%%%%%%%%%%%%%%%%%%%%%%%%%%%%%%%%%%%%%%%%%%%%%%%%%%%%%%%%%%
Fix an optimal solution $\vopt\subseteq V$ that includes the nodes 
$v_1,\dots,v_{\chi+1}$.
We next make the following parameter adjustments: 
%%%%%%%%%%%%%%%%%%%%%%%%%%%%%%%%%%%%%%%%%%%%%%%%%%%%%%%%%%%%%%%%%%%%%%%%%%%%%%%%%%%%%%%%%%%%%%%%%%%%%%%
\begin{enumerate}[label=$\triangleright$,leftmargin=*]
%%%%%%%%%%%%%%%%%%%%%%%%%%%%%%%%%%%%%%%%%%%%%%%%%%%%%%%%%%%%%%%%%%%%%%%%%%%%%%%%%%%%%%%%%%%%%%%%%%%%%%%
\item
We update an \emph{estimate} for 
$p_j$ (the number of edges of color $j$ covered by the optimal solution) from its initial value of 
$\nicefrac{\opt_\#}{\chi}$ 
in the following manner.
%%%
Let $\mu_j$ be the number of edges of color $j$ incident on at least one of the nodes in 
$\{v_1,\dots,v_{\chi+1}\}$.
Consider the quantity 
$
\qhat_j=
\sum_{u\in\vopt\setminus \{v_1,\dots,v_{\chi+1}\} }
\big|
\,
\big\{
e\eqdef\{u,v\}\in E \,|\,
( v\notin \{v_1,\dots,v_{\chi+1} \} ) 
\,\wedge\, 
(\cC(e)=j )
\big\}
\,
\big|
$.
Note that 
$
\sum_{i=1}^{\chi+1}\deg(v_i)
\leq 2\nopt
$
%%%%%
and 
$\qhat_j$ 
is an integer in the set 
$\Big\{\frac{\nopt}{\chi} - \mu_j,\frac{\nopt}{\chi} - \mu_j+1,\dots,2\big(\frac{\nopt}{\chi} - \mu_j\big) \Big\}\subseteq \{0,1,2,\dots,n\}$
since any edge can be covered by either one or two nodes.
Note that there are at most 
$\frac{\nopt}{\chi}-\mu_j+1 \leq \frac{\nopt}{\chi}\leq \frac{n}{\chi}$
possible number of integers values that each 
$\qhat_j$ may take. 
Since $\chi=O(1)$, we can try out all possible combinations of $\qhat_j$ values over \emph{all} colors in 
polynomial time since 
$(n/\chi)^\chi=n^{O(1)}$. Thus, we henceforth assume that we know the \emph{correct} value of $\qhat_j$ for 
each $j\in\{1,\dots,\chi\}$.
Note that $\qhat_j\geq 0$ since our guess is correct.
%%%%%%%%%%%%%%%%%%%%%%%%%%%%%%%%%%%%%%%%%%%%%%%%%%%%%%%%%%%%%%%%%%%%%%%%%%%%%%%%%%%%%%%%%%%%%%%%%%%%%%%
\item
Update $k$ (the number of nodes to be selected) by subtracting $\chi+1$ from it, and call the new value $\khat$.
%%%%%%%%%%%%%%%%%%%%%%%%%%%%%%%%%%%%%%%%%%%%%%%%%%%%%%%%%%%%%%%%%%%%%%%%%%%%%%%%%%%%%%%%%%%%%%%%%%%%%%%
\item
Update $C_i$ (the set of edges of color $i$) to be the set of edges \emph{in} 
$\Ehat$ that are of color $i$.
%%%%%%%%%%%%%%%%%%%%%%%%%%%%%%%%%%%%%%%%%%%%%%%%%%%%%%%%%%%%%%%%%%%%%%%%%%%%%%%%%%%%%%%%%%%%%%%%%%%%%%%
\end{enumerate}
%%%%%%%%%%%%%%%%%%%%%%%%%%%%%%%%%%%%%%%%%%%%%%%%%%%%%%%%%%%%%%%%%%%%%%%%%%%%%%%%%%%%%%%%%%%%%%%%%%%%%%%

%%%%%%%%%%%%%%%%%%%%%%%%%%%%%%%%%%%%%%%%%%%%%%%%%%%%%%%%%%%%%%%%%%%%%%%%%%%%%%%%%%%%%%%%%%%%%%%%%%%%%%%
\medskip
\noindent
\textbf{Yet another $\LP$-relaxation}
\smallskip
%%%%%%%%%%%%%%%%%%%%%%%%%%%%%%%%%%%%%%%%%%%%%%%%%%%%%%%%%%%%%%%%%%%%%%%%%%%%%%%%%%%%%%%%%%%%%%%%%%%%%%%

Let 
$|\Vhat|=\nhat$ and $|\Ehat|=\mhat$.
We will start with an initial $\LP$-relaxation of \nfmc on $\Ghat$ which will be iteratively modified 
by our rounding approach. Our $\LP$-relaxation is the following modified version of the $\LP$-relaxation in \FI{tab1}. 
%%%%%%%%%%%%%%%%%%%%%%%%%%%%%%%%%%%%%%%%%%%%%%%%%%%%%%%%%%%%%%%%%%%%%%%%%%%%%%%%%%%%%%%%%%%%%%%%%%%%%%%
\begin{enumerate}[label=$\triangleright$,leftmargin=*]
%%%%%%%%%%%%%%%%%%%%%%%%%%%%%%%%%%%%%%%%%%%%%%%%%%%%%%%%%%%%%%%%%%%%%%%%%%%%%%%%%%%%%%%%%%%%%%%%%%%%%%%
\item
There is a node indicator variable $y_v$ for every node $v\in \Vhat$ 
and 
an edge indicator variable
$x_{e,u}$ for every edge $(e,u)\in\Ehat$; thus we have $\nhat+\mhat$ variables in total.
%%%%%%%%%%%%%%%%%%%%%%%%%%%%%%%%%%%%%%%%%%%%%%%%%%%%%%%%%%%%%%%%%%%%%%%%%%%%%%%%%%%%%%%%%%%%%%%%%%%%%%%
\item
Constraints of the form ``$x_j \geq y_\ell$'' and 
``$x_j \leq \sum_{u_j\in\cS_\ell} y_\ell$''
in \FI{tab1} are removed now and instead replaced by at most two 
constraints $x_{e,u}=y_u$ if $y_u\in\Vhat$ and $x_{e,v}=y_v$ if $y_v\in\Vhat$ .
This is done so that we can apply the rank lemma in a meaningful way.
%%%%%%%%%%%%%%%%%%%%%%%%%%%%%%%%%%%%%%%%%%%%%%%%%%%%%%%%%%%%%%%%%%%%%%%%%%%%%%%%%%%%%%%%%%%%%%%%%%%%%%%
\item
Note that 
the quantity 
$\sum_{u\in\Vhat} \sum_{ e\eqdef\{u,v\}\in C_i} x_{e,u}$
for each color $i$
is the integer $\qhat_i$ mentioned before.
%%%%%%%%%%%%%%%%%%%%%%%%%%%%%%%%%%%%%%%%%%%%%%%%%%%%%%%%%%%%%%%%%%%%%%%%%%%%%%%%%%%%%%%%%%%%%%%%%%%%%%%
\item
To maximize the parameter ranges over which our algorithm can be applied, we 
replace the $\binom{\chi}{2}$ constraints in \FI{tab1} of the form 
``$\sum_{u_\ell \in C_i} x_\ell = \sum_{u_\ell \in C_j} x_\ell$ for $i,j\in\{1,\dots,\chi\}$, $i<j$''
by 
the $\chi$ constraints
$\sum_{u\in\Vhat}\sum_{ e\eqdef\{u,v\}\in C_i} x_{e,u}  = q_i'$ for 
$i\in\{1,\dots,\chi\}$.
%%%%%%%%%%%%%%%%%%%%%%%%%%%%%%%%%%%%%%%%%%%%%%%%%%%%%%%%%%%%%%%%%%%%%%%%%%%%%%%%%%%%%%%%%%%%%%%%%%%%%%%
\end{enumerate}
%%%%%%%%%%%%%%%%%%%%%%%%%%%%%%%%%%%%%%%%%%%%%%%%%%%%%%%%%%%%%%%%%%%%%%%%%%%%%%%%%%%%%%%%%%%%%%%%%%%%%%%
The entire initial $\LP$-relaxation $\cL$ for $\Ghat$ is shown in \FI{tab4} for convenience.
Note that the number of constraints in lines 
\textbf{(1)}--\textbf{(3)} of \FI{tab4} is \emph{exactly} 
$\mhat+\chi+1$.

%%%%%%%%%%%%%%%%%%%%%%%%%%%%%%%%%%%%%%%%%%%%%%%%%%%%%%%%%%%%%%%%%%%%%%%%%%%%%%%%%%%%%%%%%%%%%%%%%%%%%%%
%%%%%%%%%%%%%%%%%%%%%%%%%%%%%%%%%%%%%%%%%%%%%%%%%%%%%%%%%%%%%%%%%%%%%%%%%%%%%%%%%%%%%%%%%%%%%%%%%%%%%%%
\begin{figure}[h]
\begin{center}
\begin{tabular}{l r l l}
\toprule
& maximize & \multicolumn{2}{l}{$\psi=\sum_{u\in\Vhat}\sum_{(e,u)\in \Ehat} w(e) x_{e,u}$}
\\
[4pt]
& subject to &  & 
\\
\textbf{(1)}
& & $x_{e,u}=y_u$ 
  & for all $u\in\Vhat$ and $(e,u)\in \Ehat$   
\\
[8pt]
\textbf{(2)}
   & & $\sum_{v\in V_1} y_v=\khat$ &   
\\
[4pt]
\textbf{(3)}
   & & $\sum_{u\in\Vhat} \sum_{ e\eqdef\{u,v\}\in C_i} x_{e,u} = \qhat_i$ &   
	  for all $i\in\{1,\dots,\chi\}$
\\
[4pt]
\textbf{(4)}
   & & $0\leq  x_{e,u} \leq 1$ & for all $(e,u)\in \Ehat$
\\
[4pt]
\textbf{(5)}
   & & $0\leq y_v \leq 1$ & for all $v\in \Vhat$
\\
\bottomrule
\end{tabular}
\end{center}
\vspace*{-0.2in}
\caption{\label{tab4}The initial $\LP$-relaxation $\cL=\cL^{(0)}$ for the graph $\Ghat$ used in Theorem~\ref{thm-nfmc-con}.
The iterated rounding approach will successively modify the $\LP$ to create a sequence $\cL^{(1)},\cL^{(2)},\dots$ of $\LP$'s.}
\end{figure}
%%%%%%%%%%%%%%%%%%%%%%%%%%%%%%%%%%%%%%%%%%%%%%%%%%%%%%%%%%%%%%%%%%%%%%%%%%%%%%%%%%%%%%%%%%%%%%%%%%%%%%%
%%%%%%%%%%%%%%%%%%%%%%%%%%%%%%%%%%%%%%%%%%%%%%%%%%%%%%%%%%%%%%%%%%%%%%%%%%%%%%%%%%%%%%%%%%%%%%%%%%%%%%%

%%%%%%%%%%%%%%%%%%%%%%%%%%%%%%%%%%%%%%%%%%%%%%%%%%%%%%%%%%%%%%%%%%%%%%%%%%%%%%%%%%%%%%%%%%%%%%%%%%%%%%%
\medskip
\noindent
\textbf{Details of iterated rounding}
\smallskip
%%%%%%%%%%%%%%%%%%%%%%%%%%%%%%%%%%%%%%%%%%%%%%%%%%%%%%%%%%%%%%%%%%%%%%%%%%%%%%%%%%%%%%%%%%%%%%%%%%%%%%%

We will use the variable $t\in\{0,1,2,\dots,n\}$ to denote the iteration number of our rounding, with 
$t=0$ being the situation before any rounding has been performed, and we will use a ``\emph{superscript} $(t)$'' for the relevant 
quantities to indicate their values or status \emph{after} the $t\tx$ iteration of the rounding, \EG, 
$\nhat^{(0)}=\nhat$ and 
$\nhat^{(1)}$ is the value of $\nhat$ after the first iteration of rounding. 
Our iterated rounding algorithm \algiter in high level details is shown in \FI{tab5}, where 
the following 
notation is used for brevity
for a node $u\in\Vhat$: 
%%%%%%%%%%%%%%%%%%%%%%%%%%%%%%%%%%%%%%%%%%%%%%%%%%%%%%%%%%%%%%%%%%%%%%%%%%%%%%%%%%%%%%%%%%%%%%%%%%%%%%%
\begin{gather*}
%%%%%%%%%%%%%%%%%%%%%%%%%%%%%%%%%%%%%%%%%%%%%%%%%%%%%%%%%%%%%%%%%%%%%%%%%%%%%%%%%%%%%%%%%%%%%%%%%%%%%%%
Z_u^{\mathrm{variables}}
=\{y_u\} \, \bigcup\,  \big\{ \, x_{e,u} \,|\, (e,u)\in \Ehat \big\}
%%%%%%%%%%%%%%%%%%%%%%%%%%%%%%%%%%%%%%%%%%%%%%%%%%%%%%%%%%%%%%%%%%%%%%%%%%%%%%%%%%%%%%%%%%%%%%%%%%%%%%%
\end{gather*}
%%%%%%%%%%%%%%%%%%%%%%%%%%%%%%%%%%%%%%%%%%%%%%%%%%%%%%%%%%%%%%%%%%%%%%%%%%%%%%%%%%%%%%%%%%%%%%%%%%%%%%%
For concise analysis of our algorithm, we will use the following notations: 
%%%%%%%%%%%%%%%%%%%%%%%%%%%%%%%%%%%%%%%%%%%%%%%%%%%%%%%%%%%%%%%%%%%%%%%%%%%%%%%%%%%%%%%%%%%%%%%%%%%%%%%
\begin{enumerate}[label=$\triangleright$]
%%%%%%%%%%%%%%%%%%%%%%%%%%%%%%%%%%%%%%%%%%%%%%%%%%%%%%%%%%%%%%%%%%%%%%%%%%%%%%%%%%%%%%%%%%%%%%%%%%%%%%%
\item
$W^{\chi+1}$
is the sum of weights of all the edges 
incident to one or more nodes from the set of nodes
$\{ v_1,\dots,v_{\chi+1} \}$.
%%%%%%%%%%%%%%%%%%%%%%%%%%%%%%%%%%%%%%%%%%%%%%%%%%%%%%%%%%%%%%%%%%%%%%%%%%%%%%%%%%%%%%%%%%%%%%%%%%%%%%%
\item
$w(X)=\sum_{x_{u,e}\in X}w(e)$ for a subset of variable $X\subseteq \{ x_{e,u} \,|\, (e,u)\in \Ehat \}$.
%%%%%%%%%%%%%%%%%%%%%%%%%%%%%%%%%%%%%%%%%%%%%%%%%%%%%%%%%%%%%%%%%%%%%%%%%%%%%%%%%%%%%%%%%%%%%%%%%%%%%%%
\item
$
W_{\!\!\text{\tiny\sc Alg}}^{(t)}
=
w(\Xsol^{(t)})
$
is 
the \emph{sum} of weights of the edges whose variables are in $\Xsol^{(t)}$ 
(thus, for example, 
$W_{\!\!\text{\tiny\sc Alg}}^{(0)}=0$).
%%%%%%%%%%%%%%%%%%%%%%%%%%%%%%%%%%%%%%%%%%%%%%%%%%%%%%%%%%%%%%%%%%%%%%%%%%%%%%%%%%%%%%%%%%%%%%%%%%%%%%
\item
$\optf^{(t)}$ 
is the optimum value of the objective function of the $\LP$-relaxation $\cL^{(t)}$ during the $t\tx$ iteration 
of rounding.
%%%%%%%%%%%%%%%%%%%%%%%%%%%%%%%%%%%%%%%%%%%%%%%%%%%%%%%%%%%%%%%%%%%%%%%%%%%%%%%%%%%%%%%%%%%%%%%%%%%%%%%
\item
$\phat_i^{(t)}$
is the number of edges of color $i$ selected by \algiter
\emph{up to and including} the $t\tx$ iteration of rounding.
%%%%%%%%%%%%%%%%%%%%%%%%%%%%%%%%%%%%%%%%%%%%%%%%%%%%%%%%%%%%%%%%%%%%%%%%%%%%%%%%%%%%%%%%%%%%%%%%%%%%%%%
\item
$\tf$ is the value of $t$ in the \emph{last} iteration of rounding. 
%%%%%%%%%%%%%%%%%%%%%%%%%%%%%%%%%%%%%%%%%%%%%%%%%%%%%%%%%%%%%%%%%%%%%%%%%%%%%%%%%%%%%%%%%%%%%%%%%%%%%%%
\end{enumerate}
%%%%%%%%%%%%%%%%%%%%%%%%%%%%%%%%%%%%%%%%%%%%%%%%%%%%%%%%%%%%%%%%%%%%%%%%%%%%%%%%%%%%%%%%%%%%%%%%%%%%%%%

%%%%%%%%%%%%%%%%%%%%%%%%%%%%%%%%%%%%%%%%%%%%%%%%%%%%%%%%%%%%%%%%%%%%%%%%%%%%%%%%%%%%%%%%%%%%%%%%%%%%%%%
%%%%%%%%%%%%%%%%%%%%%%%%%%%%%%%%%%%%%%%%%%%%%%%%%%%%%%%%%%%%%%%%%%%%%%%%%%%%%%%%%%%%%%%%%%%%%%%%%%%%%%%
\begin{figure}[h]
\begin{center}
\scalebox{0.9}[0.9]{
\begin{tabular}{l}
\toprule
%%%%%%%%%%%%%%%%%%%%%%%%%%%%%%%%%%%%%%%%%%%%%%%%%%%%%%%%%%%%%%%%%%%%%%%%%%%%%%%%%%%%%%%%%%%%%%%%%%%%%%%
$(* \text{ initialization } *)$ 
%%%%%%%%%%%%%%%%%%%%%%%%%%%%%%%%%%%%%%%%%%%%%%%%%%%%%%%%%%%%%%%%%%%%%%%%%%%%%%%%%%%%%%%%%%%%%%%%%%%%%%%
\\
[5pt]
%%%%%%%%%%%%%%%%%%%%%%%%%%%%%%%%%%%%%%%%%%%%%%%%%%%%%%%%%%%%%%%%%%%%%%%%%%%%%%%%%%%%%%%%%%%%%%%%%%%%%%%
$t \leftarrow 0$ 
\hspace*{0.1in}
$\Vsol\leftarrow\emptyset$; 
\hspace*{0.1in}
$\Xsol\leftarrow\emptyset$; 
\hspace*{0.1in}
$\mbox{var-count}_{\mathrm{remaining}} \leftarrow \nhat+\mhat$; 
\hspace*{0.1in}
$\rhat \leftarrow \mhat$; 
\hspace*{0.1in}
$\nhat \leftarrow \nhat$; 
%%%%%%%%%%%%%%%%%%%%%%%%%%%%%%%%%%%%%%%%%%%%%%%%%%%%%%%%%%%%%%%%%%%%%%%%%%%%%%%%%%%%%%%%%%%%%%%%%%%%%%%
\\
[5pt]
%%%%%%%%%%%%%%%%%%%%%%%%%%%%%%%%%%%%%%%%%%%%%%%%%%%%%%%%%%%%%%%%%%%%%%%%%%%%%%%%%%%%%%%%%%%%%%%%%%%%%%%
$(* \text{ iterations of rounding } *)$ 
%%%%%%%%%%%%%%%%%%%%%%%%%%%%%%%%%%%%%%%%%%%%%%%%%%%%%%%%%%%%%%%%%%%%%%%%%%%%%%%%%%%%%%%%%%%%%%%%%%%%%%%
\\
[5pt]
%%%%%%%%%%%%%%%%%%%%%%%%%%%%%%%%%%%%%%%%%%%%%%%%%%%%%%%%%%%%%%%%%%%%%%%%%%%%%%%%%%%%%%%%%%%%%%%%%%%%%%%
\wwhile\ $(\mbox{var-count}_{\mathrm{remaining}}\neq 0)$ \ddo
%%%%%%%%%%%%%%%%%%%%%%%%%%%%%%%%%%%%%%%%%%%%%%%%%%%%%%%%%%%%%%%%%%%%%%%%%%%%%%%%%%%%%%%%%%%%%%%%%%%%%%%
\\
[5pt]
%%%%%%%%%%%%%%%%%%%%%%%%%%%%%%%%%%%%%%%%%%%%%%%%%%%%%%%%%%%%%%%%%%%%%%%%%%%%%%%%%%%%%%%%%%%%%%%%%%%%%%%
\hspace*{0.2in}
$t\leftarrow t+1$
%%%%%%%%%%%%%%%%%%%%%%%%%%%%%%%%%%%%%%%%%%%%%%%%%%%%%%%%%%%%%%%%%%%%%%%%%%%%%%%%%%%%%%%%%%%%%%%%%%%%%%%
\\
[5pt]
\hspace*{0.2in}
%%%%%%%%%%%%%%%%%%%%%%%%%%%%%%%%%%%%%%%%%%%%%%%%%%%%%%%%%%%%%%%%%%%%%%%%%%%%%%%%%%%%%%%%%%%%%%%%%%%%%%%
find an \emph{extreme-point} optimal solution of 
objective value $\optf^{(t-1)}$
for the $\LP$ $\cL^{(t-1)}$ (\emph{cf}.\ \FI{tab4})
%%%%%%%%%%%%%%%%%%%%%%%%%%%%%%%%%%%%%%%%%%%%%%%%%%%%%%%%%%%%%%%%%%%%%%%%%%%%%%%%%%%%%%%%%%%%%%%%%%%%%%%
\\
[5pt]
%%%%%%%%%%%%%%%%%%%%%%%%%%%%%%%%%%%%%%%%%%%%%%%%%%%%%%%%%%%%%%%%%%%%%%%%%%%%%%%%%%%%%%%%%%%%%%%%%%%%%%%
\hspace*{0.2in}
\textbf{begin cases} 
%%%%%%%%%%%%%%%%%%%%%%%%%%%%%%%%%%%%%%%%%%%%%%%%%%%%%%%%%%%%%%%%%%%%%%%%%%%%%%%%%%%%%%%%%%%%%%%%%%%%%%%
\\
[5pt]
%%%%%%%%%%%%%%%%%%%%%%%%%%%%%%%%%%%%%%%%%%%%%%%%%%%%%%%%%%%%%%%%%%%%%%%%%%%%%%%%%%%%%%%%%%%%%%%%%%%%%%%
\hspace*{0.3in}
\textbf{Case 1:} there exists a variable $y_u$ in the solution such that $y_u=0$
%%%%%%%%%%%%%%%%%%%%%%%%%%%%%%%%%%%%%%%%%%%%%%%%%%%%%%%%%%%%%%%%%%%%%%%%%%%%%%%%%%%%%%%%%%%%%%%%%%%%%%%
\\
%%%%%%%%%%%%%%%%%%%%%%%%%%%%%%%%%%%%%%%%%%%%%%%%%%%%%%%%%%%%%%%%%%%%%%%%%%%%%%%%%%%%%%%%%%%%%%%%%%%%%%%
\hspace*{0.5in}
$\mbox{var-count}_{\mathrm{remaining}} \leftarrow \mbox{var-count}_{\mathrm{remaining}} - 
\big| \, Z_u^{\mathrm{variables}} \, \big|
$; 
$\,\,\,\nhat \leftarrow \nhat - 1$
%%%%%%%%%%%%%%%%%%%%%%%%%%%%%%%%%%%%%%%%%%%%%%%%%%%%%%%%%%%%%%%%%%%%%%%%%%%%%%%%%%%%%%%%%%%%%%%%%%%%%%%
\\
[5pt]
%%%%%%%%%%%%%%%%%%%%%%%%%%%%%%%%%%%%%%%%%%%%%%%%%%%%%%%%%%%%%%%%%%%%%%%%%%%%%%%%%%%%%%%%%%%%%%%%%%%%%%%
\hspace*{0.5in}
$\rhat \leftarrow \rhat - \big| \,\{ \, x_{e,u} \,|\, x_{e,u} \in Z_u^{\mathrm{variables}} \, \}\, \big|$; 
$\,\,\,\Xsol \leftarrow \Xsol \cup \{ \, x_{e,u} \,|\, x_{e,u} \in Z_u^{\mathrm{variables}} \, \}$
%%%%%%%%%%%%%%%%%%%%%%%%%%%%%%%%%%%%%%%%%%%%%%%%%%%%%%%%%%%%%%%%%%%%%%%%%%%%%%%%%%%%%%%%%%%%%%%%%%%%%%%
\\
[5pt]
%%%%%%%%%%%%%%%%%%%%%%%%%%%%%%%%%%%%%%%%%%%%%%%%%%%%%%%%%%%%%%%%%%%%%%%%%%%%%%%%%%%%%%%%%%%%%%%%%%%%%%%
\hspace*{0.5in}
remove the variables in $ Z_u^{\mathrm{variables}}$ 
from $\cL^{(t-1)}$, and delete or update the constraints 
%%%%%%%%%%%%%%%%%%%%%%%%%%%%%%%%%%%%%%%%%%%%%%%%%%%%%%%%%%%%%%%%%%%%%%%%%%%%%%%%%%%%%%%%%%%%%%%%%%%%%%%
\\
%%%%%%%%%%%%%%%%%%%%%%%%%%%%%%%%%%%%%%%%%%%%%%%%%%%%%%%%%%%%%%%%%%%%%%%%%%%%%%%%%%%%%%%%%%%%%%%%%%%%%%%
\hspace*{0.8in}
and the objective function to reflect the removal of variables  
%%%%%%%%%%%%%%%%%%%%%%%%%%%%%%%%%%%%%%%%%%%%%%%%%%%%%%%%%%%%%%%%%%%%%%%%%%%%%%%%%%%%%%%%%%%%%%%%%%%%%%%
\\
[5pt]
%%%%%%%%%%%%%%%%%%%%%%%%%%%%%%%%%%%%%%%%%%%%%%%%%%%%%%%%%%%%%%%%%%%%%%%%%%%%%%%%%%%%%%%%%%%%%%%%%%%%%%%
\hspace*{0.3in}
\textbf{Case 2:} there exists a variable $y_u$ in the solution such that $y_u=1$
%%%%%%%%%%%%%%%%%%%%%%%%%%%%%%%%%%%%%%%%%%%%%%%%%%%%%%%%%%%%%%%%%%%%%%%%%%%%%%%%%%%%%%%%%%%%%%%%%%%%%%%
\\
%%%%%%%%%%%%%%%%%%%%%%%%%%%%%%%%%%%%%%%%%%%%%%%%%%%%%%%%%%%%%%%%%%%%%%%%%%%%%%%%%%%%%%%%%%%%%%%%%%%%%%%
\hspace*{0.5in}
$\Vsol \leftarrow \Vsol \cup \{u\}$;
$\,\,\,\mbox{var-count}_{\mathrm{remaining}} \leftarrow \mbox{var-count}_{\mathrm{remaining}} - 
\big| \, Z_u^{\mathrm{variables}} \, \big|
$ 
%%%%%%%%%%%%%%%%%%%%%%%%%%%%%%%%%%%%%%%%%%%%%%%%%%%%%%%%%%%%%%%%%%%%%%%%%%%%%%%%%%%%%%%%%%%%%%%%%%%%%%%
\\
[5pt]
%%%%%%%%%%%%%%%%%%%%%%%%%%%%%%%%%%%%%%%%%%%%%%%%%%%%%%%%%%%%%%%%%%%%%%%%%%%%%%%%%%%%%%%%%%%%%%%%%%%%%%%
\hspace*{0.5in}
$\khat \leftarrow \khat - 1$;
$\,\,\,\nhat \leftarrow \nhat - 1$;
$\,\,\,\rhat \leftarrow \rhat - \big| \,\{ \, x_{e,u} \,|\, x_{e,u} \in Z_u^{\mathrm{variables}} \, \}\, \big|$
%%%%%%%%%%%%%%%%%%%%%%%%%%%%%%%%%%%%%%%%%%%%%%%%%%%%%%%%%%%%%%%%%%%%%%%%%%%%%%%%%%%%%%%%%%%%%%%%%%%%%%%
\\
[5pt]
%%%%%%%%%%%%%%%%%%%%%%%%%%%%%%%%%%%%%%%%%%%%%%%%%%%%%%%%%%%%%%%%%%%%%%%%%%%%%%%%%%%%%%%%%%%%%%%%%%%%%%%
\hspace*{0.5in}
$\Xsol \leftarrow \Xsol \cup \{ \, x_{e,u} \,|\, x_{e,u} \in Z_u^{\mathrm{variables}} \, \}$
%%%%%%%%%%%%%%%%%%%%%%%%%%%%%%%%%%%%%%%%%%%%%%%%%%%%%%%%%%%%%%%%%%%%%%%%%%%%%%%%%%%%%%%%%%%%%%%%%%%%%%%
\\
[5pt]
%%%%%%%%%%%%%%%%%%%%%%%%%%%%%%%%%%%%%%%%%%%%%%%%%%%%%%%%%%%%%%%%%%%%%%%%%%%%%%%%%%%%%%%%%%%%%%%%%%%%%%%
\hspace*{0.5in}
remove the variables in $ Z_u^{\mathrm{variables}}$ 
from $\cL^{(t-1)}$, and delete or update the constraints 
%%%%%%%%%%%%%%%%%%%%%%%%%%%%%%%%%%%%%%%%%%%%%%%%%%%%%%%%%%%%%%%%%%%%%%%%%%%%%%%%%%%%%%%%%%%%%%%%%%%%%%%
\\
%%%%%%%%%%%%%%%%%%%%%%%%%%%%%%%%%%%%%%%%%%%%%%%%%%%%%%%%%%%%%%%%%%%%%%%%%%%%%%%%%%%%%%%%%%%%%%%%%%%%%%%
\hspace*{0.8in}
and the objective function to reflect the removal of variables  
%%%%%%%%%%%%%%%%%%%%%%%%%%%%%%%%%%%%%%%%%%%%%%%%%%%%%%%%%%%%%%%%%%%%%%%%%%%%%%%%%%%%%%%%%%%%%%%%%%%%%%%
\\
[5pt]
%%%%%%%%%%%%%%%%%%%%%%%%%%%%%%%%%%%%%%%%%%%%%%%%%%%%%%%%%%%%%%%%%%%%%%%%%%%%%%%%%%%%%%%%%%%%%%%%%%%%%%%
\hspace*{0.5in}
$\forall \, i\in \{1,\dots,\chi\}: \, 
\qhat_i \leftarrow \qhat_i - \big| \,\{ \, x_{e,u} \,|\, x_{e,u} \in Z_u^{\mathrm{variables}} \text{ and } \cC(e)=i \, \}\, \big|
$
%%%%%%%%%%%%%%%%%%%%%%%%%%%%%%%%%%%%%%%%%%%%%%%%%%%%%%%%%%%%%%%%%%%%%%%%%%%%%%%%%%%%%%%%%%%%%%%%%%%%%%%
\\
[5pt]
%%%%%%%%%%%%%%%%%%%%%%%%%%%%%%%%%%%%%%%%%%%%%%%%%%%%%%%%%%%%%%%%%%%%%%%%%%%%%%%%%%%%%%%%%%%%%%%%%%%%%%%
\hspace*{0.5in}
subtract the value
$
\sum_{x_{e,u} \in Z_u^{\mathrm{variables}}} w(e)x_{e,u}
$
from the objective function $\psi^{(t-1)}$
%%%%%%%%%%%%%%%%%%%%%%%%%%%%%%%%%%%%%%%%%%%%%%%%%%%%%%%%%%%%%%%%%%%%%%%%%%%%%%%%%%%%%%%%%%%%%%%%%%%%%%%
\\
[5pt]
%%%%%%%%%%%%%%%%%%%%%%%%%%%%%%%%%%%%%%%%%%%%%%%%%%%%%%%%%%%%%%%%%%%%%%%%%%%%%%%%%%%%%%%%%%%%%%%%%%%%%%%
\hspace*{0.3in}
\textbf{Case 3:} 
$1\leq \nhat \leq \chi+1$
%%%%%%%%%%%%%%%%%%%%%%%%%%%%%%%%%%%%%%%%%%%%%%%%%%%%%%%%%%%%%%%%%%%%%%%%%%%%%%%%%%%%%%%%%%%%%%%%%%%%%%%
\\
%%%%%%%%%%%%%%%%%%%%%%%%%%%%%%%%%%%%%%%%%%%%%%%%%%%%%%%%%%%%%%%%%%%%%%%%%%%%%%%%%%%%%%%%%%%%%%%%%%%%%%%
\hspace*{0.5in}
let $y_{u_1},\dots,y_{u_{\nhat'}}$ be the remaining non-zero $1\leq \nhat'\leq\nicefrac{\nhat}{2}$ node indicator variables
%%%%%%%%%%%%%%%%%%%%%%%%%%%%%%%%%%%%%%%%%%%%%%%%%%%%%%%%%%%%%%%%%%%%%%%%%%%%%%%%%%%%%%%%%%%%%%%%%%%%%%%
\\
[5pt]
%%%%%%%%%%%%%%%%%%%%%%%%%%%%%%%%%%%%%%%%%%%%%%%%%%%%%%%%%%%%%%%%%%%%%%%%%%%%%%%%%%%%%%%%%%%%%%%%%%%%%%%
\hspace*{0.5in}
$\mbox{var-count}_{\mathrm{remaining}} \leftarrow 0$;
$\,\,\,\khat \leftarrow 0$;
$\,\,\,\nhat \leftarrow 0$; 
$\,\,\,\rhat \leftarrow 0$
%%%%%%%%%%%%%%%%%%%%%%%%%%%%%%%%%%%%%%%%%%%%%%%%%%%%%%%%%%%%%%%%%%%%%%%%%%%%%%%%%%%%%%%%%%%%%%%%%%%%%%%
\\
[5pt]
%%%%%%%%%%%%%%%%%%%%%%%%%%%%%%%%%%%%%%%%%%%%%%%%%%%%%%%%%%%%%%%%%%%%%%%%%%%%%%%%%%%%%%%%%%%%%%%%%%%%%%%
\hspace*{0.5in}
$\forall \, i\in \{1,\dots,\chi\}: \, 
\qhat_i \leftarrow \qhat_i - \big| \,\{ \, x_{e,u} \,|\, x_{e,u} \in Z_u^{\mathrm{variables}} \text{ and } \cC(e)=i \, \}\, \big|
$
%%%%%%%%%%%%%%%%%%%%%%%%%%%%%%%%%%%%%%%%%%%%%%%%%%%%%%%%%%%%%%%%%%%%%%%%%%%%%%%%%%%%%%%%%%%%%%%%%%%%%%%
\\
[5pt]
%%%%%%%%%%%%%%%%%%%%%%%%%%%%%%%%%%%%%%%%%%%%%%%%%%%%%%%%%%%%%%%%%%%%%%%%%%%%%%%%%%%%%%%%%%%%%%%%%%%%%%%
\hspace*{0.5in}
$\Vsol \leftarrow \Vsol \cup \{u_1,\dots,u_{\nhat'}\}$;
$\,\,\,\Xsol \leftarrow \Xsol \cup \{ \, x_{e,u_j} \,|\, x_{e,u_j} \in Z_{u_j}^{\mathrm{variables}}, \, j\in\{1,\dots,\nhat'\} \, \}$
%%%%%%%%%%%%%%%%%%%%%%%%%%%%%%%%%%%%%%%%%%%%%%%%%%%%%%%%%%%%%%%%%%%%%%%%%%%%%%%%%%%%%%%%%%%%%%%%%%%%%%%
%%%%%%%%%%%%%%%%%%%%%%%%%%%%%%%%%%%%%%%%%%%%%%%%%%%%%%%%%%%%%%%%%%%%%%%%%%%%%%%%%%%%%%%%%%%%%%%%%%%%%%%
\\
[5pt]
%%%%%%%%%%%%%%%%%%%%%%%%%%%%%%%%%%%%%%%%%%%%%%%%%%%%%%%%%%%%%%%%%%%%%%%%%%%%%%%%%%%%%%%%%%%%%%%%%%%%%%%
\hspace*{0.2in}
\textbf{end cases} 
%%%%%%%%%%%%%%%%%%%%%%%%%%%%%%%%%%%%%%%%%%%%%%%%%%%%%%%%%%%%%%%%%%%%%%%%%%%%%%%%%%%%%%%%%%%%%%%%%%%%%%%
\\
[5pt]
%%%%%%%%%%%%%%%%%%%%%%%%%%%%%%%%%%%%%%%%%%%%%%%%%%%%%%%%%%%%%%%%%%%%%%%%%%%%%%%%%%%%%%%%%%%%%%%%%%%%%%%
\textbf{end while}
%%%%%%%%%%%%%%%%%%%%%%%%%%%%%%%%%%%%%%%%%%%%%%%%%%%%%%%%%%%%%%%%%%%%%%%%%%%%%%%%%%%%%%%%%%%%%%%%%%%%%%%
\\
[5pt]
%%%%%%%%%%%%%%%%%%%%%%%%%%%%%%%%%%%%%%%%%%%%%%%%%%%%%%%%%%%%%%%%%%%%%%%%%%%%%%%%%%%%%%%%%%%%%%%%%%%%%%%
$\optf \leftarrow 0$
%%%%%%%%%%%%%%%%%%%%%%%%%%%%%%%%%%%%%%%%%%%%%%%%%%%%%%%%%%%%%%%%%%%%%%%%%%%%%%%%%%%%%%%%%%%%%%%%%%%%%%%
\\
[5pt]
%%%%%%%%%%%%%%%%%%%%%%%%%%%%%%%%%%%%%%%%%%%%%%%%%%%%%%%%%%%%%%%%%%%%%%%%%%%%%%%%%%%%%%%%%%%%%%%%%%%%%%%
\rreturn
$\Vsol$
%%%%%%%%%%%%%%%%%%%%%%%%%%%%%%%%%%%%%%%%%%%%%%%%%%%%%%%%%%%%%%%%%%%%%%%%%%%%%%%%%%%%%%%%%%%%%%%%%%%%%%%
\\
\bottomrule
\end{tabular}
}
\end{center}
\vspace*{-0.2in}
\caption{\label{tab5}Pseudo-code of the iterated rounding algorithm \algiter used in Theorem~\ref{thm-nfmc-con}. $\Vsol^{(\tf)}$ is the
set of nodes selected in our solution.}
\end{figure}
%%%%%%%%%%%%%%%%%%%%%%%%%%%%%%%%%%%%%%%%%%%%%%%%%%%%%%%%%%%%%%%%%%%%%%%%%%%%%%%%%%%%%%%%%%%%%%%%%%%%%%%
%%%%%%%%%%%%%%%%%%%%%%%%%%%%%%%%%%%%%%%%%%%%%%%%%%%%%%%%%%%%%%%%%%%%%%%%%%%%%%%%%%%%%%%%%%%%%%%%%%%%%%%

\begin{lemma}\label{lem-term}
\algiter terminates after at most $n$ iterations and selects 
at most $k+\frac{\chi-1}{2}$ nodes.
\end{lemma}

\begin{proof}
For finite termination, 
it suffices to show that at least one of the three cases in \algiter always applies.
Consider the first iteration, say when $t=\alpha$, when neither Case~1 nor Case~2 applies. Note that this  
also implies that $x_{e,u}\notin\{0,1\}$ for any variable $x_{e,u}$ in 
$\LP^{(\alpha)}$
since otherwise the variable $y_u$ in 
$\LP^{(\alpha)}$ will be either $0$ or $1$ via the equality constraint $y_u=x_{e,u}$ and one of Case~1 or Case~2 
will apply. 
Thus the total number of non-zero variables is 
$\nhat^{(\alpha)}+\rhat^{(\alpha)}$.
Since the constraints in lines 
\textbf{(4)}--\textbf{(5)} of \FI{tab4} are not strict constraints now (\IE, not satisfied with equalities), 
the total number of \emph{any} maximal set of strict constraints 
is at most the total number of constraints in lines 
\textbf{(1)}--\textbf{(3)} of \FI{tab4}, \IE, at most 
$\rhat^{(\alpha)}+\chi+1$.
By the rank lemma (Fact~\ref{fact2})
$\rhat^{(\alpha)}+\chi+1 \geq \nhat^{(\alpha)}+\rhat^{(\alpha)}
\equiv 
\nhat^{(\alpha)}\leq \chi+1$, 
which implies 
Case~3 applies and the algorithm terminates.

We now prove the bound on the number of selected sets.
%%%
The value of $\khat$ decreases by $1$ every time a new node is selected
in Case~2 and remains unchanged in Case~1 where no node is selected. 
In the very last iteration involving Case~3, since $G$ has no isolated nodes
the number of node indicator variables is at least the number of edge indicator 
variables, implying $\nhat'\leq\nicefrac{\nhat}{2}\leq \frac{\chi+1}{2}$. Since 
$\khat^{(\tf-1)}\geq 1$, the total number of nodes selected is at most
$k+(\nhat'-1)\leq k+\frac{\chi-1}{2}$. 
\end{proof}

\begin{lemma}
The sum of weights $\Gamma$ of the edges selected by \algiter is at least
$\nicefrac{\opt}{2}$.
\end{lemma}

\begin{proof}
Let 
$
W_{\!\!\text{\footnotesize\sc Alg}\!-W}^{(t)}
=
W_{\!\!\text{\tiny\sc Alg}}^{(t)}- W^{\chi+1}
$, and 
$
\opt_{\!-W}
=
\opt - W^{\chi+1}
$.
The proof of Lemma~\ref{lem-term} shows that Case~3 of \algiter is executed only when $t=\tf$.
Thus, the details of \algiter in \FI{tab5} imply the following sequence of assertions: 
%%%%%%%%%%%%%%%%%%%%%%%%%%%%%%%%%%%%%%%%%%%%%%%%%%%%%%%%%%%%%%%%%%%%%%%%%%%%%%%%%%%%%%%%%%%%%%%%%%%%%%%
\begin{enumerate}[label=\textbf{(\emph{\roman*})},leftmargin=*]
%%%%%%%%%%%%%%%%%%%%%%%%%%%%%%%%%%%%%%%%%%%%%%%%%%%%%%%%%%%%%%%%%%%%%%%%%%%%%%%%%%%%%%%%%%%%%%%%%%%%%%%
\item
$
\optf^{(0)} 
\geq 
\opt_{\!-W}
$ and  
$
\optf^{(t)} = 
\optf^{(t-1)} - \big( w(\Xsol^{(t)}) - w(\Xsol^{(t-1)}) \big)
$
for $t\in\{1,\dots,\tf-1\}$.
Since the variables 
$x_{e_{u,j}}\in \Xsol^{(\tf)}\setminus \Xsol^{(\tf-1)}$
are at most $1$, we have 
%%%%%%%%%%%%%%%%%%%%%%%%%%%%%%%%%%%%%%%%%%%%%%%%%%%%%%%%%%%%%%%%%%%%%%%%%%%%%%%
\begin{multline*}
\optf^{(\tf)}
=
\optf^{(\tf-1)} - 
\hspace*{-0.4in}
\sum _{x_{e_{u,j}} \in \Xsol^{(\tf)}\setminus \Xsol^{(\tf-1)}} \hspace*{-0.2in} w(e) x_{e_{u,j}}
\geq 
\optf^{(\tf-1)} - 
\hspace*{-0.4in}
\sum _{x_{e_{u,j}} \in \Xsol^{(\tf)}\setminus \Xsol^{(\tf-1)}} \hspace*{-0.2in} w(e) 
%%%%%%%%%%%%%%%%%%%%%%%%%%%%%%%%%%%%%%%%%%%%%%%%%%%%%%%%%%%%%%%%%%%%%%%%%%%%%%%
\\
%%%%%%%%%%%%%%%%%%%%%%%%%%%%%%%%%%%%%%%%%%%%%%%%%%%%%%%%%%%%%%%%%%%%%%%%%%%%%%%
=
\optf^{(\tf-1)} - 
\big( w(\Xsol^{(\tf)}) - w(\Xsol^{(\tf-1)}) \big)
\end{multline*}
%%%%%%%%%%%%%%%%%%%%%%%%%%%%%%%%%%%%%%%%%%%%%%%%%%%%%%%%%%%%%%%%%%%%%%%%%%%%%%%
Using the fact that 
$\optf^{(\tf)}=0$, we can therefore
unravel the recurrence to get 
%%%%%%%%%%%%%%%%%%%%%%%%%%%%%%%%%%%%%%%%%%%%%%%%%%%%%%%%%%%%%%%%%%%%%%%%%%%%%%%
\begin{gather}
\optf^{(\tf)}
\geq
\optf^{(0)} - 
w(\Xsol^{(\tf)})
\,\Rightarrow\,
w(\Xsol^{(\tf)})
\geq
\optf^{(0)}
\geq 
\opt_{\!-W}
\label{eq:iter2}
\end{gather}
%%%%%%%%%%%%%%%%%%%%%%%%%%%%%%%%%%%%%%%%%%%%%%%%%%%%%%%%%%%%%%%%%%%%%%%%%%%%%%%
\item
$
W_{\!\!\text{\footnotesize\sc Alg}\!-W}^{(0)} = 0
$ and 
$
W_{\!\!\text{\footnotesize\sc Alg}\!-W}^{(t)} = 
W_{\!\!\text{\footnotesize\sc Alg}\!-W}^{(t-1)}
+ \big( w(\Xsol^{(t)}) -  w(\Xsol^{(t-1)}) \big) 
$
for $t\in\{1,\dots,\tf\}$.
Using~\eqref{eq:iter2} we can 
unravel the recurrence we get 
%%%%%%%%%%%%%%%%%%%%%%%%%%%%%%%%%%%%%%%%%%%%%%%%%%%%%%%%%%%%%%%%%%%%%%%%%%%%%%%
\begin{gather*}
W_{\!\!\text{\footnotesize\sc Alg}\!-W}^{(\tf)} = 
w(\Xsol^{(\tf)})
\geq 
\optf^{(0)}
\geq 
\opt_{\!-W}
\end{gather*}
%%%%%%%%%%%%%%%%%%%%%%%%%%%%%%%%%%%%%%%%%%%%%%%%%%%%%%%%%%%%%%%%%%%%%%%%%%%%%%%
\end{enumerate}
%%%%%%%%%%%%%%%%%%%%%%%%%%%%%%%%%%%%%%%%%%%%%%%%%%%%%%%%%%%%%%%%%%%%%%%%%%%%%%%
Noting that an edge $e\eqdef\{u,v\}$
can contribute the value of $w(e)$ \emph{twice} in 
$W_{\!\!\text{\tiny\sc Alg}}^{(\tf)}$
corresponding to the two variables 
$x_{e,u}$ and 
$x_{e,v}$, the total weight $\Gamma$ of selected edges in our solution is at least 
%%%%%%%%%%%%%%%%%%%%%%%%%%%%%%%%%%%%%%%%%%%%%%%%%%%%%%%%%%%%%%%%%%%%%%%%%%%%%%%
\begin{gather*}
\Gamma \geq 
W^{\chi+1}
+
\frac{1}{2} W_{\!\!\text{\tiny\sc Alg}}^{(\tf)}
\geq
\frac{  W^{\chi+1} + W_{\!\!\text{\tiny\sc Alg}}^{(\tf)} }{2}
\geq
\frac{  W^{\chi+1} +  \opt_{\!-W} }{2}
=
\frac{\opt}{2}
\qed
\end{gather*}
%%%%%%%%%%%%%%%%%%%%%%%%%%%%%%%%%%%%%%%%%%%%%%%%%%%%%%%%%%%%%%%%%%%%%%%%%%%%%%%
\end{proof}

Our proof of Theorem~\ref{thm-nfmc-con}
is therefore completed once we prove the following lemma.

\begin{lemma}\label{lem-ratio}
For all $i,j\in\{1,\dots,\chi\}$ 
$
\frac
{\phat_i^{(\tf)}}
{\phat_j^{(\tf)}}
\leq 
4+ 4 \chi
$.
\end{lemma}

\begin{proof}
When $t=\tf$ Case~3 applies and,  
since the variables 
$x_{e_{u,j}}\in \Xsol^{(\tf)}\setminus \Xsol^{(\tf-1)}$
are at most $1$, 
$\qhat_i^{(\tf)}\leq 0$ and consequently 
%%%%%%%%%%%%%%%%%%%%%%%%%%%%%%%%%%%%%%%%%%%%%%%%%%%%%%%%%%%%%%%%%%%%%%%%%%%%%%%
$
\qhat_i^{(\tf-1)} - \qhat_i^{(\tf)}
\geq 
\qhat_i^{(\tf-1)}
$.
%%%%%%%%%%%%%%%%%%%%%%%%%%%%%%%%%%%%%%%%%%%%%%%%%%%%%%%%%%%%%%%%%%%%%%%%%%%%%%%
Noting that an edge $e\eqdef\{u,v\}$
can contribute \emph{twice} in 
the various $\qhat_i^{(t)}$'s 
corresponding to the two variables 
$x_{e,u}$ and 
$x_{e,v}$ 
and remembering that 
$\qhat_i^{(0)}=\qhat_i$,
we get 
%%%%%%%%%%%%%%%%%%%%%%%%%%%%%%%%%%%%%%%%%%%%%%%%%%%%%%%%%%%%%%%%%%%%%%%%%%%%%%%
\begin{multline*}
\phat_i^{(\tf)}
\geq 
\frac{1}{2} 
\Big( \sum_{t=1}^{\tf} \!\! \big( \qhat_i^{(t-1)} - \qhat_i^{(t)}  \big)
\Big) 
+\mu_i
\geq 
\frac{1}{2} 
\Big( \sum_{t=1}^{\tf-1} \!\!\! \big( \qhat_i^{(t-1)} - \qhat_i^{(t)}  \big)
+
\qhat_i^{(\tf-1)}
\Big) 
+\mu_i
=
\frac{\qhat_i^{(0)}}{2}
+\mu_i
\\
=
\frac{\qhat_i}{2}
+\mu_i
\geq 
\frac{\opt_\#}{2\chi}
-
\frac{\mu_i}{2}
+\mu_i
=
%%\frac{q_i}{2}
\frac{\opt_\#}{2\chi}
+
\frac{\mu_i}{2}
\end{multline*}
%%%%%%%%%%%%%%%%%%%%%%%%%%%%%%%%%%%%%%%%%%%%%%%%%%%%%%%%%%%%%%%%%%%%%%%%%%%%%%%
We can get an upper bound on 
$\phat_i^{(\tf)}$
by getting an upper bound on 
$
\qhat_i^{(\tf-1)} - \qhat_i^{(\tf)}
$
in the following manner.
Consider the 
$\nhat^{(\tf)'}<\nhat^{(\tf)}\leq\chi+1$ 
nodes 
$u_1,\dots,u_{\nhat^{(\tf)'}}$
in Case~3. 
By choice of the nodes 
$v_1,\dots,v_{\chi+1}$ of degrees 
$\deg(v_1),\dots,\deg(v_{\chi+1})$, respectively, 
the number of edges incident on $u_i$ is at most $\deg(v_i)$ for all
$i\in\{1,\dots, \nhat^{(\tf)'} \}$. 
Thus, we get 
$
\qhat_i^{(\tf-1)} - \qhat_i^{(\tf)}
\leq 
\sum_{j=1}^{\chi+1} \deg(v_j)
\leq 
2 \opt_\#
$,
and 
consequently
%%%%%%%%%%%%%%%%%%%%%%%%%%%%%%%%%%%%%%%%%%%%%%%%%%%%%%%%%%%%%%%%%%%%%%%%%%%%%%%
\begin{multline*}
\phat_i^{(\tf)}
\leq 
\sum_{t=1}^{\tf} \!\! \big( \qhat_i^{(t-1)} - \qhat_i^{(t)}  \big)
\leq
\sum_{t=1}^{\tf-1} \!\! \big( \qhat_i^{(t-1)} - \qhat_i^{(t)}  \big)
+
2 \opt_\#
=
\qhat_i
+
2 \opt_\#
\\
\leq 
2 \left( {\textstyle \frac{\nopt}{\chi}-\mu_i } \right) 
+
2 \opt_\#
=
{\textstyle (2+2\chi) \frac{ \opt_\#}{\chi} 
}
\end{multline*}
%%%%%%%%%%%%%%%%%%%%%%%%%%%%%%%%%%%%%%%%%%%%%%%%%%%%%%%%%%%%%%%%%%%%%%%%%%%%%%%
Thus, for all $i,j\in\{1,\dots,\chi\}$ we have 
%%%%%%%%%%%%%%%%%%%%%%%%%%%%%%%%%%%%%%%%%%%%%%%%%%%%%%%%%%%%%%%%%%%%%%%%%%%%%%%
\[
\frac
{\phat_i^{(\tf)}}
{\phat_j^{(\tf)}}
\leq 
\frac { 
        { (2 + 2\chi) \frac{ \opt_\#}{\chi} }
      }
			{
        \frac{\opt_\#}{2\chi} + \frac{\mu_j}{2}
			}
< 4 + 4\chi
\qed
\]
%%%%%%%%%%%%%%%%%%%%%%%%%%%%%%%%%%%%%%%%%%%%%%%%%%%%%%%%%%%%%%%%%%%%%%%%%%%%%%%
\end{proof}

%%%%%%%%%%%%%%%%%%%%%%%%%%%%%%%%%%%%%%%%%%%%%%%%%%%%%%%%%%%%%%%%%%%%%%%%%%%%%%%%%%%%%%%%%%%%%%%%%%%%%%%
\subsubsection{The case of arbitrary $\chi$} 
%%%%%%%%%%%%%%%%%%%%%%%%%%%%%%%%%%%%%%%%%%%%%%%%%%%%%%%%%%%%%%%%%%%%%%%%%%%%%%%%%%%%%%%%%%%%%%%%%%%%%%%

As stated below,  
there are two steps in the previous algorithm that cannot be executed in polynomial time 
when $\chi$ is \emph{not} a constant: 
%%%%%%%%%%%%%%%%%%%%%%%%%%%%%%%%%%%%%%%%%%%%%%%%%%%%%%%%%%%%%%%%%%%%%%%%%%%%%%%%%%%%%%%%%%%%%%%%%%%%%%%
\begin{enumerate}[label=\textbf{(\arabic*)},leftmargin=*]
%%%%%%%%%%%%%%%%%%%%%%%%%%%%%%%%%%%%%%%%%%%%%%%%%%%%%%%%%%%%%%%%%%%%%%%%%%%%%%%%%%%%%%%%%%%%%%%%%%%%%%%
\item\label{s1}
We cannot guess the $\chi+1$ nodes $v_1,\dots,v_{\chi+1}$ 
in polynomial time. Instead, we guess only one node $v_1$ 
such that there exists an optimal solution contains $v_1$
with the following property: ``the remaining $k-1$ nodes in the solution have degree \emph{at most} $\deg(v_1)$''.
%%%%%%%%%%%%%%%%%%%%%%%%%%%%%%%%%%%%%%%%%%%%%%%%%%%%%%%%%%%%%%%%%%%%%%%%%%%%%%%%%%%%%%%%%%%%%%%%%%%%%%%
\item\label{s2}
We cannot guess the exact value of $\qhat_i$ by exhaustive enumeration and therefore we cannot 
use the $\chi$ constraints 
``$\sum_{u\in\Vhat} \sum_{ e\eqdef\{u,v\}\in C_i} x_{e,u} = \qhat_i$''
in line 
\textbf{(3)}
of the $\LP$-relaxation in \FI{tab4} anymore.
However, note that it still holds that 
$\qhat_i$ 
is an integer in the set 
$\big\{\frac{\nopt}{\chi} - \mu_i,\frac{\nopt}{\chi} - \mu_i+1,\dots,2\big(\frac{\nopt}{\chi} - \mu_i\big) \big\}$.
Thus, instead we use the $2\chi$ constraints 
\[
\textbf{(3) }
\,\,\,\,\,\,
{\textstyle \frac{\nopt}{\chi} - \mu_i }
\leq
   \sum_{u\in\Vhat} \sum_{ e\eqdef\{u,v\}\in C_i} x_{e,u} = \qhat_i
\leq
{\textstyle 2 \big( \frac{\nopt}{\chi} - \mu_i \big)  }
\,\,\,\,\,\,
     \text{ for all } i\in\{1,\dots,\chi\}
\]
%%%%%%%%%%%%%%%%%%%%%%%%%%%%%%%%%%%%%%%%%%%%%%%%%%%%%%%%%%%%%%%%%%%%%%%%%%%%%%%%%%%%%%%%%%%%%%%%%%%%%%%
\end{enumerate}
%%%%%%%%%%%%%%%%%%%%%%%%%%%%%%%%%%%%%%%%%%%%%%%%%%%%%%%%%%%%%%%%%%%%%%%%%%%%%%%%%%%%%%%%%%%%%%%%%%%%%%%
We need modifications of the bounds in the previous proof to reflect these changes as follows:
%%%%%%%%%%%%%%%%%%%%%%%%%%%%%%%%%%%%%%%%%%%%%%%%%%%%%%%%%%%%%%%%%%%%%%%%%%%%%%%%%%%%%%%%%%%%%%%%%%%%%%%
\begin{enumerate}[label=$\triangleright$]
%%%%%%%%%%%%%%%%%%%%%%%%%%%%%%%%%%%%%%%%%%%%%%%%%%%%%%%%%%%%%%%%%%%%%%%%%%%%%%%%%%%%%%%%%%%%%%%%%%%%%%%
\item
We make some obvious parameter value adjustments such as: 
$\vhat=V\setminus \{v_1\}$, $\Ehat=E\setminus \{ \{u,v_1\} \,|\, \{u,v_1\}\in E\}$, 
$\khat^{(0)}=k-1$, $\mu_i\leq\deg(v_1)$ for all $i$.
%%%%%%%%%%%%%%%%%%%%%%%%%%%%%%%%%%%%%%%%%%%%%%%%%%%%%%%%%%%%%%%%%%%%%%%%%%%%%%%%%%%%%%%%%%%%%%%%%%%%%%%
\item
The number of constraints in lines 
\textbf{(1)}--\textbf{(3)} of \FI{tab4} is now 
$\mhat+2\chi+1$. 
%%%%%%%%%%%%%%%%%%%%%%%%%%%%%%%%%%%%%%%%%%%%%%%%%%%%%%%%%%%%%%%%%%%%%%%%%%%%%%%%%%%%%%%%%%%%%%%%%%%%%%%
\item
The condition in Case~3 of \FI{tab5} is now 
$1\leq \nhat \leq 2\chi+1$.
%%%%%%%%%%%%%%%%%%%%%%%%%%%%%%%%%%%%%%%%%%%%%%%%%%%%%%%%%%%%%%%%%%%%%%%%%%%%%%%%%%%%%%%%%%%%%%%%%%%%%%%
\item
In Lemma~\ref{lem-term}, we select at most $\big\lfloor k+\frac{2\chi-1}{2}\big\rfloor = k+\chi-1$ nodes.
%%%%%%%%%%%%%%%%%%%%%%%%%%%%%%%%%%%%%%%%%%%%%%%%%%%%%%%%%%%%%%%%%%%%%%%%%%%%%%%%%%%%%%%%%%%%%%%%%%%%%%%
\item
The calculations for the upper bound 
for $\phat_i^{(\tf)}$
in Lemma~\ref{lem-ratio} change as follows.
By choice of the node $v_1$ of degree $\deg(v_1)$, 
the number of edges incident on $u_i$ is at most $\deg(v_1)$ for all
$i\in\{1,\dots, \nhat^{(\tf)'} \}$. 
This now gives 
$
\qhat_i^{(\tf-1)} - \qhat_i^{(\tf)}
\leq 
(2\chi+1) \deg(v_1)
\leq 
(2\chi+1) \nopt
$,
and therefore 
$
\phat_i^{(\tf)}
\leq 
\qhat_i
+
(2\chi+1) \nopt
\leq
2 \left( {\textstyle \frac{\nopt}{\chi}-\mu_i } \right) 
+
(2\chi+1) \nopt
<
(2+\chi+2\chi^2) \frac{\nopt}{\chi}
$.
This gives us the following updated bound: 
\[
\frac
{\phat_i^{(\tf)}}
{\phat_j^{(\tf)}}
\leq 
\frac { 
        { (2+\chi+2\chi^2) \frac{\opt_\#}{\chi} }
      }
			{
        \frac{\opt_\#}{2\chi} + \frac{\mu_j}{2}
			}
<
4+2\chi+4\chi^2
\]
%%%%%%%%%%%%%%%%%%%%%%%%%%%%%%%%%%%%%%%%%%%%%%%%%%%%%%%%%%%%%%%%%%%%%%%%%%%%%%%%%%%%%%%%%%%%%%%%%%%%%%%
\end{enumerate}
%%%%%%%%%%%%%%%%%%%%%%%%%%%%%%%%%%%%%%%%%%%%%%%%%%%%%%%%%%%%%%%%%%%%%%%%%%%%%%%%%%%%%%%%%%%%%%%%%%%%%%%

%%%%%%%%%%%%%%%%%%%%%%%%%%%%%%%%%%%%%%%%%%%%%%%%%%%%%%%%%%%%%%%%%%%%%%%%%%%%%%%%%%%%%%%%%%%%%%%%%%%%%%%
\subsection{The general case: approximating \fmc}
\label{sec-fmc-con}
%%%%%%%%%%%%%%%%%%%%%%%%%%%%%%%%%%%%%%%%%%%%%%%%%%%%%%%%%%%%%%%%%%%%%%%%%%%%%%%%%%%%%%%%%%%%%%%%%%%%%%%

\begin{theorem}[generalizing Theorem~\ref{thm-nfmc-con} for \fmc]\label{thm-fmc-con}
We can design a deterministic polynomial-time approximation algorithm for \fmc with the following 
properties: 
%%%%%%%%%%%%%%%%%%%%%%%%%%%%%%%%%%%%%%%%%%%%%%%%%%%%%%%%%%%%%%%%%%%%%%%%%%%%%%%%%%%%%%%%%%%%%%%%%%%%%%%
\begin{enumerate}[label=\textbf{\emph{(}\alph*\emph{)}},leftmargin=*]
%%%%%%%%%%%%%%%%%%%%%%%%%%%%%%%%%%%%%%%%%%%%%%%%%%%%%%%%%%%%%%%%%%%%%%%%%%%%%%%%%%%%%%%%%%%%%%%%%%%%%%%
\item
The algorithm selects 
$\tau$ sets where 
$
\tau \leq 
\begin{cases}
k+\frac{\chi-1}{2}, & \mbox{if $\chi=O(1)$}
\\
k +\chi-1, & \mbox{otherwise}
\end{cases}
$
%%%%%%%%%%%%%%%%%%%%%%%%%%%%%%%%%%%%%%%%%%%%%%%%%%%%%%%%%%%%%%%%%%%%%%%%%%%%%%%%%%%%%%%%%%%%%%%%%%%%%%%
\item
The algorithm
is a   
$\nicefrac{1}{f}$-approximation for \nfmc, \emph{\IE}, 
the total weight of the selected elements 
is at least $\nicefrac{\opt}{f}$.
%%%%%%%%%%%%%%%%%%%%%%%%%%%%%%%%%%%%%%%%%%%%%%%%%%%%%%%%%%%%%%%%%%%%%%%%%%%%%%%%%%%%%%%%%%%%%%%%%%%%%%%
\item
The algorithm
satisfies
the $\eps$-approximate coloring constraints $($cf.\ Inequality~\eqref{eq1}$)$
as follows:
%%%%%%%%%%%%%%%%%%%%%%%%%%%%%%%%%%%%%%%%%%%%%%%%%%%%%%%%%%%%%%%%%%%%%%%%%%%%%%%%%%%%%%%%%%%%%%%%%%%%%%%
\begin{quote}
for all 
$i,j\in\{1,\dots,\chi\}$, 
$\frac{p_i}{p_j}<
\begin{cases}
O( \min \{ \chi^2f,\, \chi f^2\} ),  & \mbox{if $\chi=O(1)$}
\\
O(f^2+\chi^2f), & \mbox{otherwise}
\end{cases}
$
\end{quote}
%%%%%%%%%%%%%%%%%%%%%%%%%%%%%%%%%%%%%%%%%%%%%%%%%%%%%%%%%%%%%%%%%%%%%%%%%%%%%%%%%%%%%%%%%%%%%%%%%%%%%%%
\end{enumerate}
%%%%%%%%%%%%%%%%%%%%%%%%%%%%%%%%%%%%%%%%%%%%%%%%%%%%%%%%%%%%%%%%%%%%%%%%%%%%%%%%%%%%%%%%%%%%%%%%%%%%%%%
\end{theorem}

The proof of 
Theorem~\ref{thm-fmc-con}
is a suitable modified version of the proof of 
Theorem~\ref{thm-nfmc-con}.
We point out the important alterations that are needed.

%%%%%%%%%%%%%%%%%%%%%%%%%%%%%%%%%%%%%%%%%%%%%%%%%%%%%%%%%%%%%%%%%%%%%%%%%%%%%%%%%%%%%%%%%%%%%%%%%%%%%%%
\subsubsection*{General modifications}
%%%%%%%%%%%%%%%%%%%%%%%%%%%%%%%%%%%%%%%%%%%%%%%%%%%%%%%%%%%%%%%%%%%%%%%%%%%%%%%%%%%%%%%%%%%%%%%%%%%%%%%

%%%%%%%%%%%%%%%%%%%%%%%%%%%%%%%%%%%%%%%%%%%%%%%%%%%%%%%%%%%%%%%%%%%%%%%%%%%%%%%%%%%%%%%%%%%%%%%%%%%%%%%
\begin{enumerate}[label=$\triangleright$,leftmargin=*]
%%%%%%%%%%%%%%%%%%%%%%%%%%%%%%%%%%%%%%%%%%%%%%%%%%%%%%%%%%%%%%%%%%%%%%%%%%%%%%%%%%%%%%%%%%%%%%%%%%%%%%%
\item
Nodes and edges now correspond to sets and elements, respectively, incidence of an edge on a node 
corresponds to membership of an element in a set, and degree of a node correspond to number of elements in 
a set.
%%%%%%%%%%%%%%%%%%%%%%%%%%%%%%%%%%%%%%%%%%%%%%%%%%%%%%%%%%%%%%%%%%%%%%%%%%%%%%%%%%%%%%%%%%%%%%%%%%%%%%%
\item
There is a set indicator variable $y_j$ for every element $\cS_j\in \Vhat$. 
For every element $(u_i,\cS_j)\in\Ehat$, there is 
an element indicator variable
$x_{i,j}$ 
and a constraint 
$x_{i,j} = y_j$.
%%%%%%%%%%%%%%%%%%%%%%%%%%%%%%%%%%%%%%%%%%%%%%%%%%%%%%%%%%%%%%%%%%%%%%%%%%%%%%%%%%%%%%%%%%%%%%%%%%%%%%%
\item
Now 
$
\sum_{i=1}^{\chi+1} |\cS_i|
\leq 
\min\{\chi,f\} \nopt
$
since any element in any one of the sets from $\cS_1,\dots,\cS_{\chi+1}$ can appear in at most 
$\min\{\chi,f\}$ other sets in the collection of sets $\cS_1,\dots,\cS_{\chi+1}$.
Also, $\qhat_j$ 
is an integer in the set 
$\Big\{\frac{\nopt}{\chi} - \mu_j,\frac{\nopt}{\chi} - \mu_j+1,\dots,f\big(\frac{\nopt}{\chi} - \mu_j\big) \Big\}\subseteq \{0,1,2,\dots,n\}$
since any element can appear in at most $f$ sets.
%%%%%%%%%%%%%%%%%%%%%%%%%%%%%%%%%%%%%%%%%%%%%%%%%%%%%%%%%%%%%%%%%%%%%%%%%%%%%%%%%%%%%%%%%%%%%%%%%%%%%%%
\item
An element $u_i$ appearing in $f_i\leq f$ sets, say sets $\cS_1,\dots,\cS_{f_i}$, 
can now contribute the value of $w(u_i)$ at most $f_i\leq f$ times in 
$W_{\!\!\text{\tiny\sc Alg}}^{(\tf)}$
corresponding to the $f_i$ variables 
$x_{i,1},\dots,x_{i,f_i}$.
Thus, we get a $\nicefrac{1}{f}$-approximation to the objective function.
%%%%%%%%%%%%%%%%%%%%%%%%%%%%%%%%%%%%%%%%%%%%%%%%%%%%%%%%%%%%%%%%%%%%%%%%%%%%%%%%%%%%%%%%%%%%%%%%%%%%%%%
\end{enumerate}
%%%%%%%%%%%%%%%%%%%%%%%%%%%%%%%%%%%%%%%%%%%%%%%%%%%%%%%%%%%%%%%%%%%%%%%%%%%%%%%%%%%%%%%%%%%%%%%%%%%%%%%

%%%%%%%%%%%%%%%%%%%%%%%%%%%%%%%%%%%%%%%%%%%%%%%%%%%%%%%%%%%%%%%%%%%%%%%%%%%%%%%%%%%%%%%%%%%%%%%%%%%%%%%
\subsubsection*{Modifications related to $\chi=O(1)$ case}
%%%%%%%%%%%%%%%%%%%%%%%%%%%%%%%%%%%%%%%%%%%%%%%%%%%%%%%%%%%%%%%%%%%%%%%%%%%%%%%%%%%%%%%%%%%%%%%%%%%%%%%

%%%%%%%%%%%%%%%%%%%%%%%%%%%%%%%%%%%%%%%%%%%%%%%%%%%%%%%%%%%%%%%%%%%%%%%%%%%%%%%%%%%%%%%%%%%%%%%%%%%%%%%
\begin{enumerate}[label=$\triangleright$,leftmargin=*]
%%%%%%%%%%%%%%%%%%%%%%%%%%%%%%%%%%%%%%%%%%%%%%%%%%%%%%%%%%%%%%%%%%%%%%%%%%%%%%%%%%%%%%%%%%%%%%%%%%%%%%%
\item
An element $u_i$ appearing in $f_i\leq f$ sets, say sets $\cS_1,\dots,\cS_{f_i}$, 
can now contribute at most $f_i\leq f$ times in 
the various $\qhat_i^{(t)}$'s 
corresponding to the $f_i$ variables 
$x_{i,1},\dots,x_{i,f_i}$.
This modifies the relevant inequality for 
$\phat_i^{(\tf)}$ as follows:
%%%%%%%%%%%%%%%%%%%%%%%%%%%%%%%%%%%%%%%%%%%%%%%%%%%%%%%%%%%%%%%%%%%%%%%%%%%%%%%%%%%%%%%%%%%%%%%%%%%%%%%
\begin{gather*}
\phat_i^{(\tf)}
\geq 
\frac{\qhat_i^{(0)}}{f} +\mu_i
\geq 
\frac{\opt_\#}{f\chi}
-
\frac{\mu_i}{f}
+\mu_i
>
\frac{\opt_\#}{f\chi}
%%%%%%%%%%%%%%%%%%%%%%%%%%%%%%%%%%%%%%%%%%%%%%%%%%%%%%%%%%%%%%%%%%%%%%%%%%%%%%%%%%%%%%%%%%%%%%%%%%%%%%%
\\
%%%%%%%%%%%%%%%%%%%%%%%%%%%%%%%%%%%%%%%%%%%%%%%%%%%%%%%%%%%%%%%%%%%%%%%%%%%%%%%%%%%%%%%%%%%%%%%%%%%%%%%
\phat_i^{(\tf)}
\leq 
\qhat_i
+
\sum_{i=1}^{\chi+1} |\cS_i|
\leq
{\textstyle f\big(\frac{\nopt}{\chi} - \mu_i\big) }
+
\min\{\chi,f\} \nopt
<
{\textstyle \min \{ \chi^2+f,\, 2 \chi f\} \frac{\nopt}{\chi} }
%%%%%%%%%%%%%%%%%%%%%%%%%%%%%%%%%%%%%%%%%%%%%%%%%%%%%%%%%%%%%%%%%%%%%%%%%%%%%%%%%%%%%%%%%%%%%%%%%%%%%%%
\\
%%%%%%%%%%%%%%%%%%%%%%%%%%%%%%%%%%%%%%%%%%%%%%%%%%%%%%%%%%%%%%%%%%%%%%%%%%%%%%%%%%%%%%%%%%%%%%%%%%%%%%%
\frac
{\phat_i^{(\tf)}}
{\phat_j^{(\tf)}}
\leq 
\frac {
     {\textstyle \min \{ \chi^2+f,\, 2 \chi f\} \frac{\nopt}{\chi} }
  }
	{
     \frac{\opt_\#}{f\chi}
	}
\leq
\min \{ \chi^2f+f^2,\, 2 \chi f^2\}
=O( \min \{ \chi^2f,\, \chi f^2\} )
\end{gather*}
%%%%%%%%%%%%%%%%%%%%%%%%%%%%%%%%%%%%%%%%%%%%%%%%%%%%%%%%%%%%%%%%%%%%%%%%%%%%%%%%%%%%%%%%%%%%%%%%%%%%%%%
\end{enumerate}
%%%%%%%%%%%%%%%%%%%%%%%%%%%%%%%%%%%%%%%%%%%%%%%%%%%%%%%%%%%%%%%%%%%%%%%%%%%%%%%%%%%%%%%%%%%%%%%%%%%%%%%

%%%%%%%%%%%%%%%%%%%%%%%%%%%%%%%%%%%%%%%%%%%%%%%%%%%%%%%%%%%%%%%%%%%%%%%%%%%%%%%%%%%%%%%%%%%%%%%%%%%%%%%
\subsubsection*{Modifications related to the arbitrary $\chi$ case}
%%%%%%%%%%%%%%%%%%%%%%%%%%%%%%%%%%%%%%%%%%%%%%%%%%%%%%%%%%%%%%%%%%%%%%%%%%%%%%%%%%%%%%%%%%%%%%%%%%%%%%%

%%%%%%%%%%%%%%%%%%%%%%%%%%%%%%%%%%%%%%%%%%%%%%%%%%%%%%%%%%%%%%%%%%%%%%%%%%%%%%%%%%%%%%%%%%%%%%%%%%%%%%%
\begin{enumerate}[label=$\triangleright$,leftmargin=*]
%%%%%%%%%%%%%%%%%%%%%%%%%%%%%%%%%%%%%%%%%%%%%%%%%%%%%%%%%%%%%%%%%%%%%%%%%%%%%%%%%%%%%%%%%%%%%%%%%%%%%%%
\item
The calculations for the upper bound 
for $\phat_i^{(\tf)}$
in Lemma~\ref{lem-ratio} change as follows.
%%%%%%%%%%%%%%%%%%%%%%%%%%%%%%%%%%%%%%%%%%%%%%%%%%%%%%%%%%%%%%%%%%%%%%%%%%%%%%%%%%%%%%%%%%%%%%%%%%%%%%%
\begin{gather*}
{\textstyle f\big(\frac{\nopt}{\chi} - \mu_i\big) }
+
(2\chi+1)\nopt
<
{\textstyle  \big( f + \chi + 2\chi^2 \big) \frac{\nopt}{\chi} }
<
{\textstyle  \big( f + 3\chi^2 \big) \frac{\nopt}{\chi} }
\end{gather*}
%%%%%%%%%%%%%%%%%%%%%%%%%%%%%%%%%%%%%%%%%%%%%%%%%%%%%%%%%%%%%%%%%%%%%%%%%%%%%%%%%%%%%%%%%%%%%%%%%%%%%%%
This gives the final bound of 
$
\frac
{\phat_i^{(\tf)}}
{\phat_j^{(\tf)}}
\leq 
f^2 + 3\chi^2 f =O(f^2+\chi^2f)
$.
%%%%%%%%%%%%%%%%%%%%%%%%%%%%%%%%%%%%%%%%%%%%%%%%%%%%%%%%%%%%%%%%%%%%%%%%%%%%%%%%%%%%%%%%%%%%%%%%%%%%%%%
\end{enumerate}
%%%%%%%%%%%%%%%%%%%%%%%%%%%%%%%%%%%%%%%%%%%%%%%%%%%%%%%%%%%%%%%%%%%%%%%%%%%%%%%%%%%%%%%%%%%%%%%%%%%%%%%

%%%%%%%%%%%%%%%%%%%%%%%%%%%%%%%%%%%%%%%%%%%%%%%%%%%%%%%%%%%%%%%%%%%%%%%%%%%%%%%%%%%%%%%%%%%%%%%%%%%%%%%
\section{Approximation algorithms for two special cases of \fmc}
\label{sec-special-fmc}
%%%%%%%%%%%%%%%%%%%%%%%%%%%%%%%%%%%%%%%%%%%%%%%%%%%%%%%%%%%%%%%%%%%%%%%%%%%%%%%%%%%%%%%%%%%%%%%%%%%%%%%

For approximating these special cases of \fmc, which are still $\NP$-complete,
we will be specific about the various constants and will 
try to provide approximation algorithms with as tight a constant as we can.
For this section, let $\varrho = \max\{ \varrho(f), \, \varrho(k) \}$. Note that $\varrho > 1 - \nicefrac{1}{\bee}$.

%%%%%%%%%%%%%%%%%%%%%%%%%%%%%%%%%%%%%%%%%%%%%%%%%%%%%%%%%%%%%%%%%%%%%%%%%%%%%%%%%%%%%%%%%%%%%%%%%%%%%%%
\subsection{\sfmc: almost optimal deterministic approximation with ``at most'' $k$ sets}
\label{sec-segr}
%%%%%%%%%%%%%%%%%%%%%%%%%%%%%%%%%%%%%%%%%%%%%%%%%%%%%%%%%%%%%%%%%%%%%%%%%%%%%%%%%%%%%%%%%%%%%%%%%%%%%%%

%
Note that Lemma~\ref{lem1} shows that finding a feasible solution is $\NP$-complete even for unweighted \sfmc 
with $\chi=2$.
Further inapproximability results for \sfmc are stated in 
Remark~\ref{rr1}.

%%%%%%%%%%%%%%%%%%%%%%%%%%%%%%%%%%%%%%%%%%%%%%%%%%%%%%%%%%%%%%%%%%%%%%%%%%%%%%%%%%%%%%%%%%%%%%%%%%%%%%%
\begin{theorem}\label{thm-segr-app}
There exists a polynomial-time deterministic algorithm \alggreed that, given an instance of unweighted 
\sfmc($\chi,k$) 
outputs a solution with the following properties: 
%%%%%%%%%%%%%%%%%%%%%%%%%%%%%%%%%%%%%%%%%%%%%%%%%%%%%%%%%%%%%%%%%%%%%%%%%%%%%%%%%%%%%%%%%%%%%%%%%%%%%%%
\begin{description}
\item[(\emph{a})]
The number of selected sets is at most $k$. 
%%%%%%%%%%%%%%%%%%%%%%%%%%%%%%%%%%%%%%%%%%%%%%%%%%%%%%%%%%%%%%%%%%%%%%%%%%%%%%%%%%%%%%%%%%%%%%%%%%%%%%%
\item[(\emph{b})]
The approximation ratio is at least 
$\varrho > 1 - \nicefrac{1}{\bee}$.
%%%%%%%%%%%%%%%%%%%%%%%%%%%%%%%%%%%%%%%%%%%%%%%%%%%%%%%%%%%%%%%%%%%%%%%%%%%%%%%%%%%%%%%%%%%%%%%%%%%%%%%
\item[(\emph{c})]
The coloring constraints are  
$2$-approximately 
satisfied $($cf.~\emph{\eqref{eq1}}$)$, \emph{\IE}, 
%%%
\[
\forall \,  i,j\in\{1,\dots,\chi\}: \, 
p_i \leq 2 \, p_j 
\]
%%%%%%%%%%%%%%%%%%%%%%%%%%%%%%%%%%%%%%%%%%%%%%%%%%%%%%%%%%%%%%%%%%%%%%%%%%%%%%%%%%%%%%%%%%%%%%%%%%%%%%%
\end{description}
%%%%%%%%%%%%%%%%%%%%%%%%%%%%%%%%%%%%%%%%%%%%%%%%%%%%%%%%%%%%%%%%%%%%%%%%%%%%%%%%%%%%%%%%%%%%%%%%%%%%%%%
\end{theorem}
%%%%%%%%%%%%%%%%%%%%%%%%%%%%%%%%%%%%%%%%%%%%%%%%%%%%%%%%%%%%%%%%%%%%%%%%%%%%%%%%%%%%%%%%%%%%%%%%%%%%%%%

%%%%%%%%%%%%%%%%%%%%%%%%%%%%%%%%%%%%%%%%%%%%%%%%%%%%%%%%%%%%%%%%%%%%%%%%%%%%%%%%%%%%%%%%%%%%%%%%%%%%%%%
\begin{remark}\label{rr1}
Based on the $(1-\nicefrac{1}{\bee})$-inapproximability result of Feige in~\emph{\cite{F98}}
for the maximum $k$-set coverage problem, it is not difficult to see the two constants in 
Theorem~\ref{thm-segr-app}, namely $\rho$ and $2$, cannot be improved beyond 
$1-\nicefrac{1}{\bee}+\eps$ 
and 
${(1-\nicefrac{1}{\bee})}^{-1}+\eps\approx 1.58+\eps$, respectively,  
for any $\eps>0$ and all $\chi\geq 2$ assuming 
$\mbox{P}\neq\NP$. 
\end{remark}
%%%%%%%%%%%%%%%%%%%%%%%%%%%%%%%%%%%%%%%%%%%%%%%%%%%%%%%%%%%%%%%%%%%%%%%%%%%%%%%%%%%%%%%%%%%%%%%%%%%%%%%

%%%%%%%%%%%%%%%%%%%%%%%%%%%%%%%%%%%%%%%%%%%%%%%%%%%%%%%%%%%%%%%%%%%%%%%%%%%%%%%%%%%%%%%%%%%%%%%%%%%%%%%
\begin{remark}
The ``at most $k$ sets'' part of the proof arises in the following steps of the algorithm.
Since we cannot know $k_r$ exactly, we can only assume $\widehat{k_r}\leq k_r$ since it is possible 
that the algorithm for the maximum $k$-set coverage also covers at least $\rho\frac{\opt_\#}{\chi}$
elements for some $k<k_r$.
Secondly, even if we have the guessed the correct value of $k_r$, 
the algorithm for the maximum $k_r$-set coverage 
may cover more than $2\rho\frac{\opt_\#}{\chi}$ elements, and thus we have to ``un-select'' some of the selected
sets to get the desired bounds (the proof shows that sometimes we may have to un-select all but one set).
The following example shows that a solution that insists on selecting exactly $k$ sets may need to select sets all of which 
are not in our solution. 
Consider the following instance of
unweighted \fmcg{1}{\ell}: 
$\cU=\{u_1,\dotsc,u_n\}$, 
$\ell=\nicefrac{n}{2}$,
$\cS_1=\{u_1,\dotsc,u_{n/2}\}$,
and 
$\cS_{j+1}=\{u_{(n/2)+j}\}$ for $j=1,\dots,\nicefrac{n}{2}$.
Our algorithm will select the set $\cS_1$ whereas any 
solution that selects exactly $\ell$ sets must selects the sets 
$\cS_2,\dots,\cS_{(n/2)+1}$.
\end{remark}
%%%%%%%%%%%%%%%%%%%%%%%%%%%%%%%%%%%%%%%%%%%%%%%%%%%%%%%%%%%%%%%%%%%%%%%%%%%%%%%%%%%%%%%%%%%%%%%%%%%%%%%

%%%%%%%%%%%%%%%%%%%%%%%%%%%%%%%%%%%%%%%%%%%%%%%%%%%%%%%%%%%%%%%%%%%%%%%%%%%%%%%%%%%%%%%%%%%%%%%%%%%%%%%
\begin{proof}
We reuse the notations, terminologies and bounds shown in the proof of Theorem~\ref{thm-main} as needed.
Let $\cU_1,\dots,\cU_\chi$ be the partition of the universe based on the color of the elements, \IE, 
$\cU_r=\{ u_\ell \,|\, \cC(u_\ell)=r \}$ for $r\in\{1,\dots,\chi\}$. 
By the definition of 
\sfmc
every set contains elements from exactly one such partition and thus, after renaming the sets and elements for notational 
convenience, we may set assume that our collection $\cS_1,\dots,\cS_m$ of $m$ sets 
is partitioned into $\chi$ collection of sets, where the $r\tx$ collection (for $r\in\{1,\dots,\chi\}$) 
contains the sets $\cS_1^r,\dots,\cS_{m_r}^r$ over the universe $\cU_r=\{u_1,\dots,u_{n_r}\}$ 
of $n_r$ elements such that $\sum_{r=1}^{\chi} m_r=m$ and 
$\sum_{r=1}^{\chi} n_r=n$.
For $r\in\{1,\dots,\chi\}$ and any $\ell$
let \fmc$\!\!_r(1,\ell)$
be the unweighted \fmcg{1}{\ell} problem defined over the universe $\cU_r$ and the collection of sets 
$\cS_1^r,\dots,\cS_{m_r}^r$.
The following observation holds trivially.

%%%%%%%%%%%%%%%%%%%%%%%%%%%%%%%%%%%%%%%%%%%%%%%%%%%%%%%%%%%%%%%%%%%%%%%%%%%%%%%%%%%%%%%%%%%%%%%%%%%%%%%
\begin{quote}
Unweighted \sfmc($\chi,k$)
has a valid solution covering 
$\ell\in\{\chi,2\chi,\dots,\left\lfloor\nicefrac{n}{\chi}\right\rfloor\chi\}$ elements if and only if
%%%%%%%%%%%%%%%%%%%%%%%%%%%%%%%%%%%%%%%%%%%%%%%%%%%%%%%%%%%%%%%%%%%%%%%%%%%%%%%%%%%%%%%%%%%%%%%%%%%%%%%
\begin{enumerate*}[label=(\emph{\roman*})]
\item 
for each 
$r\in\{1,\dots,\chi\}$,
\fmc$\!\!_r(1,k_r)$
has a valid solution covering $\nicefrac{\ell}{\chi}$ elements for some $k_r>0$, and 
%%%%%%%%%%%%%%%%%%%%%%%%%%%%%%%%%%%%%%%%%%%%%%%%%%%%%%%%%%%%%%%%%%%%%%%%%%%%%%%%%%%%%%%%%%%%%%%%%%%%%%%
\item 
$\sum_{r=1}^\chi k_r=k$.
\end{enumerate*}
\end{quote}
%%%%%%%%%%%%%%%%%%%%%%%%%%%%%%%%%%%%%%%%%%%%%%%%%%%%%%%%%%%%%%%%%%%%%%%%%%%%%%%%%%%%%%%%%%%%%%%%%%%%%%%

The above observation suggests that we can guess the value of 
$\opt_\#$
by trying out all possible values of $\ell$ just like the algorithms in Theorem~\ref{thm-main}, 
and for each such value of $\ell$ 
we can solve $\chi$ \emph{independent} \fmc instances and combine them to get a solution of the
original \sfmc instance. 
Although we cannot possibly solve 
the \fmc$\!\!_r(1,k_r)$ problems exactly, appropriate approximate solutions of these problems do correspond to a similar
approximate solution of 
\sfmc($\chi,k$) as stated in the following observation:

%%%%%%%%%%%%%%%%%%%%%%%%%%%%%%%%%%%%%%%%%%%%%%%%%%%%%%%%%%%%%%%%%%%%%%%%%%%%%%%%%%%%%%%%%%%%%%%%%%%%%%%
\begin{quote}
Suppose that for each $r\in\{1,\dots,\chi\}$ we have a solution 
$\cS_{i_1}^r,\dots,\cS_{i_{\widehat{k_r}}}^r\subseteq \cU_r$
of \fmc~$\!_r(1,k_r)$
with the following properties (for some $\eta_1 \leq 1$ and $\eta_2 \geq 1$): 
%%%%%%%%%%%%%%%%%%%%%%%%%%%%%%%%%%%%%%%%%%%%%%%%%%%%%%%%%%%%%%%%%%%%%%%%%%%%%%%%%%%%%%%%%%%%%%%%%%%%%%%
\begin{enumerate*}[label=(\emph{\roman*})]
\item
$\eta_1 (\nicefrac{\ell}{\chi} ) \leq | \cup_{p=1}^{\widehat{k_r}} \cS_{i_p}^r | \leq \eta_2 (\nicefrac{\ell}{\chi})$,
and
%%%%%%%%%%%%%%%%%%%%%%%%%%%%%%%%%%%%%%%%%%%%%%%%%%%%%%%%%%%%%%%%%%%%%%%%%%%%%%%%%%%%%%%%%%%%%%%%%%%%%%%
\item
$\widehat{k_r}\leq k_r$.
\end{enumerate*}
%%%%%%%%%%%%%%%%%%%%%%%%%%%%%%%%%%%%%%%%%%%%%%%%%%%%%%%%%%%%%%%%%%%%%%%%%%%%%%%%%%%%%%%%%%%%%%%%%%%%%%%
Then, the collection of
sets 
$\big\{ 
\cS_{i_{\ell_r}}^r \,|\, 
\ell_r\in \{1,\dots,\widehat{k_r}\} ,\, 
r\in\{1,\dots,\chi\}
\big\}
$
outputs a solution of 
\sfmc($\chi,k$) 
with the following properties: 
%%%%%%%%%%%%%%%%%%%%%%%%%%%%%%%%%%%%%%%%%%%%%%%%%%%%%%%%%%%%%%%%%%%%%%%%%%%%%%%%%%%%%%%%%%%%%%%%%%%%%%%
\begin{enumerate*}[label=(\emph{\alph*})]
\item
the number of selected sets is at most $k$, 
%%%%%%%%%%%%%%%%%%%%%%%%%%%%%%%%%%%%%%%%%%%%%%%%%%%%%%%%%%%%%%%%%%%%%%%%%%%%%%%%%%%%%%%%%%%%%%%%%%%%%%%
\item
the number of elements covered is at least $\eta_1\ell$,
and 
%%%%%%%%%%%%%%%%%%%%%%%%%%%%%%%%%%%%%%%%%%%%%%%%%%%%%%%%%%%%%%%%%%%%%%%%%%%%%%%%%%%%%%%%%%%%%%%%%%%%%%%
\item
for any pair 
$i,j\in\{1,\dots,\chi\}$, 
$\nicefrac{p_i}{p_j} \leq \nicefrac{\eta_2}{\eta_1}$.
%%%%%%%%%%%%%%%%%%%%%%%%%%%%%%%%%%%%%%%%%%%%%%%%%%%%%%%%%%%%%%%%%%%%%%%%%%%%%%%%%%%%%%%%%%%%%%%%%%%%%%%
\end{enumerate*}
%%%%%%%%%%%%%%%%%%%%%%%%%%%%%%%%%%%%%%%%%%%%%%%%%%%%%%%%%%%%%%%%%%%%%%%%%%%%%%%%%%%%%%%%%%%%%%%%%%%%%%%
\end{quote}
%%%%%%%%%%%%%%%%%%%%%%%%%%%%%%%%%%%%%%%%%%%%%%%%%%%%%%%%%%%%%%%%%%%%%%%%%%%%%%%%%%%%%%%%%%%%%%%%%%%%%%%

By the above observation, to prove our claim it suffices if we can find a solution for 
\fmc~$\!_r(1,k_r)$ for any $r$ 
with $\ell=\opt_\#$, $\eta_1=\varrho$ and $\eta_2 = 2\,\varrho$. 
For convenience, we will omit the \emph{superscript} $r$ from the set labels while dealing with 
\fmc$\!_r(1,k_r)$.
Remove from consideration any sets from 
$\cS_1,\dots,\cS_{m_r}$ that contains more than $\nicefrac{\ell}{\chi}$ elements, and 
consider the standard (unweighted) maximum $k$-set coverage problem,
that ignores constraint~(\emph{i}) of the above observation, 
on these remaining collection of sets $\cT$ over the universe $\cU_r$. 
Since we have guessed the correct value of $\ell$, 
there is at least one valid solution and thus the following assertions hold: 
%%%%%%%%%%%%%%%%%%%%%%%%%%%%%%%%%%%%%%%%%%%%%%%%%%%%%%%%%%%%%%%%%%%%%%%%%%%%%%%%%%%%%%%%%%%%%%%%%%%%%%%
\begin{enumerate*}[label=({\Roman*})]
%%%%%%%%%%%%%%%%%%%%%%%%%%%%%%%%%%%%%%%%%%%%%%%%%%%%%%%%%%%%%%%%%%%%%%%%%%%%%%%%%%%%%%%%%%%%%%%%%%%%%%%
\item
there exists a set of $k_r$ sets 
that covers 
$\frac{\opt_\#}{\chi}$
elements, and 
%%%%%%%%%%%%%%%%%%%%%%%%%%%%%%%%%%%%%%%%%%%%%%%%%%%%%%%%%%%%%%%%%%%%%%%%%%%%%%%%%%%%%%%%%%%%%%%%%%%%%%%
\item
$|\cT|\geq k_r$.
%%%%%%%%%%%%%%%%%%%%%%%%%%%%%%%%%%%%%%%%%%%%%%%%%%%%%%%%%%%%%%%%%%%%%%%%%%%%%%%%%%%%%%%%%%%%%%%%%%%%%%%
\end{enumerate*}
%%%%%%%%%%%%%%%%%%%%%%%%%%%%%%%%%%%%%%%%%%%%%%%%%%%%%%%%%%%%%%%%%%%%%%%%%%%%%%%%%%%%%%%%%%%%%%%%%%%%%%%
%
Let $\nu_k$ denote the \emph{maximum} number of elements that can be covered by selecting 
$k$ sets from $\cT$.
There are the following two well-known algorithm algorithms for the maximum $k$-set coverage problem 
both of which select exactly $k$ sets:
the greedy algorithm covers at least 
$\varrho(k)\nu_k$ elements~\cite[Proposition 5.1]{F98},  
where the pipage-rounding algorithm (based on the $\LP$-relaxation in \FI{f1})
covers at least 
$\varrho(f)\nu_k$ elements~\cite{AS04}.  
Note that we do \emph{not} know the exact value of $k_r$ and we \emph{cannot} guess by enumerating every possible 
$k_r$ values for every $r\in\{1,\dots,\chi\}$ in polynomial time. 
To overcome this obstacle, we use the following steps.
%%%%%%%%%%%%%%%%%%%%%%%%%%%%%%%%%%%%%%%%%%%%%%%%%%%%%%%%%%%%%%%%%%%%%%%%%%%%%%%%%%%%%%%%%%%%%%%%%%%%%%%
\begin{enumerate}[label=$\triangleright$]
%%%%%%%%%%%%%%%%%%%%%%%%%%%%%%%%%%%%%%%%%%%%%%%%%%%%%%%%%%%%%%%%%%%%%%%%%%%%%%%%%%%%%%%%%%%%%%%%%%%%%%%
\item
We run both the algorithms for maximum $k$-set coverage for $k=1,2,\dots$ 
until we find the first (\emph{smallest}) index 
$\widehat{k_r}\leq k_r$
such that the better of the two algorithms cover \emph{at least} 
$\max\{ \varrho( \widehat{k_r}), \, \varrho(f) \}\frac{\opt_\#}{\chi}\geq \rho \frac{\opt_\#}{\chi}$ elements.
%%%%%%%%%%%%%%%%%%%%%%%%%%%%%%%%%%%%%%%%%%%%%%%%%%%%%%%%%%%%%%%%%%%%%%%%%%%%%%%%%%%%%%%%%%%%%%%%%%%%%%%
\item
Suppose that this algorithm 
selects the 
$\widehat{k_r}$
sets (after possible re-numbering of set indices) 
$\cS_1,\dots,\cS_{\widehat{k_r}}$, where 
we have ordered the sets such that 
for every $j\in\{2,3,\dots,\widehat{k_r}\}$ 
the number of elements covered by $\cS_j$ and \emph{not} covered by any of the sets $\cS_1,\dots,\cS_{j-1}$
is \emph{at least} as many as
the number of elements covered by $\cS_\ell$ and not covered by any of the sets $\cS_1,\dots,\cS_{j-1}$
for \emph{any} $\ell>j$. 
Remember that $\max_{j\in\{1,\dots,\widehat{k_r}\} } \{ |\cS_j| \} \leq \nicefrac{\opt_\#}{\chi}$.
Let $j$ be the \emph{smallest} index such 
$
| \cup_{\ell=1}^{j-1} \cS_j |
<
\varrho \frac{\opt_\#}{\chi}
$
but 
$
| \cup_{\ell=1}^{j} \cS_j |
\geq
\varrho \frac{\opt_\#}{\chi}
$.
We have the following cases.
%%%%%%%%%%%%%%%%%%%%%%%%%%%%%%%%%%%%%%%%%%%%%%%%%%%%%%%%%%%%%%%%%%%%%%%%%%%%%%%%%%%%%%%%%%%%%%%%%%%%%%%
\begin{enumerate}[label=$\triangleright$]
%%%%%%%%%%%%%%%%%%%%%%%%%%%%%%%%%%%%%%%%%%%%%%%%%%%%%%%%%%%%%%%%%%%%%%%%%%%%%%%%%%%%%%%%%%%%%%%%%%%%%%%
\item
If $|\cS_j|\geq \varrho \frac{\opt_\#}{\chi}$ then we select $\cS_j$ as our solution since 
$\varrho \frac{\opt_\#}{\chi} \leq |\cS_j|\leq \frac{\opt_\#}{\chi}<2 \, \varrho \frac{\opt_\#}{\chi}$.
%%%%%%%%%%%%%%%%%%%%%%%%%%%%%%%%%%%%%%%%%%%%%%%%%%%%%%%%%%%%%%%%%%%%%%%%%%%%%%%%%%%%%%%%%%%%%%%%%%%%%%%
\item
Otherwise $|\cS_j|<\varrho \frac{\opt_\#}{\chi}$ and in this case 
we select the $j\leq \widehat{k_r} \leq k_r$ sets 
$\cS_1,\dots,\cS_{j}$ in our solution since 
$
\varrho \frac{\opt_\#}{\chi}
\leq
| \cup_{\ell=1}^{j} \cS_j |
\leq
2\,\varrho \frac{\opt_\#}{\chi}
$.
%%%%%%%%%%%%%%%%%%%%%%%%%%%%%%%%%%%%%%%%%%%%%%%%%%%%%%%%%%%%%%%%%%%%%%%%%%%%%%%%%%%%%%%%%%%%%%%%%%%%%%%
\end{enumerate}
%%%%%%%%%%%%%%%%%%%%%%%%%%%%%%%%%%%%%%%%%%%%%%%%%%%%%%%%%%%%%%%%%%%%%%%%%%%%%%%%%%%%%%%%%%%%%%%%%%%%%%%
\end{enumerate}
%%%%%%%%%%%%%%%%%%%%%%%%%%%%%%%%%%%%%%%%%%%%%%%%%%%%%%%%%%%%%%%%%%%%%%%%%%%%%%%%%%%%%%%%%%%%%%%%%%%%%%%
\end{proof}
%%%%%%%%%%%%%%%%%%%%%%%%%%%%%%%%%%%%%%%%%%%%%%%%%%%%%%%%%%%%%%%%%%%%%%%%%%%%%%%%%%%%%%%%%%%%%%%%%%%%%%%

%%%%%%%%%%%%%%%%%%%%%%%%%%%%%%%%%%%%%%%%%%%%%%%%%%%%%%%%%%%%%%%%%%%%%%%%%%%%%%%%%%%%%%%%%%%%%%%%%%%%%%%
\subsection{\bfmc: improved deterministic approximation}
\label{sec-bala}
%%%%%%%%%%%%%%%%%%%%%%%%%%%%%%%%%%%%%%%%%%%%%%%%%%%%%%%%%%%%%%%%%%%%%%%%%%%%%%%%%%%%%%%%%%%%%%%%%%%%%%%

%%%%%%%%%%%%%%%%%%%%%%%%%%%%%%%%%%%%%%%%%%%%%%%%%%%%%%%%%%%%%%%%%%%%%%%%%%%%%%%%%%%%%%%%%%%%%%%%%%%%%%%
\begin{proposition}\label{prop-bala-app}
There exists a polynomial-time deterministic algorithm 
\alggreedsmpl
that, given an instance of unweighted 
\bfmc$\!\!(\chi,k)$ 
outputs a solution with the following properties: 
%%%%%%%%%%%%%%%%%%%%%%%%%%%%%%%%%%%%%%%%%%%%%%%%%%%%%%%%%%%%%%%%%%%%%%%%%%%%%%%%%%%%%%%%%%%%%%%%%%%%%%%
\begin{description}
\item[(\emph{a})]
The number of selected sets is $($exactly$)$ $k$. 
%%%%%%%%%%%%%%%%%%%%%%%%%%%%%%%%%%%%%%%%%%%%%%%%%%%%%%%%%%%%%%%%%%%%%%%%%%%%%%%%%%%%%%%%%%%%%%%%%%%%%%%
\item[(\emph{b})]
The approximation ratio is at least 
$\varrho > 1 - \nicefrac{1}{\bee}$.
%%%%%%%%%%%%%%%%%%%%%%%%%%%%%%%%%%%%%%%%%%%%%%%%%%%%%%%%%%%%%%%%%%%%%%%%%%%%%%%%%%%%%%%%%%%%%%%%%%%%%%%
\item[(\emph{c})]
The coloring constraints are  
$O(\Delta f)$-approximately 
satisfied $($cf.~\emph{\eqref{eq1}}$)$, \emph{\IE}, 
$
\forall \,  i,j\in\{1,\dots,\chi\}: \, 
\nicefrac{p_i}{p_j} \leq (2+2\Delta)f 
$.
%%%%%%%%%%%%%%%%%%%%%%%%%%%%%%%%%%%%%%%%%%%%%%%%%%%%%%%%%%%%%%%%%%%%%%%%%%%%%%%%%%%%%%%%%%%%%%%%%%%%%%%
\end{description}
%%%%%%%%%%%%%%%%%%%%%%%%%%%%%%%%%%%%%%%%%%%%%%%%%%%%%%%%%%%%%%%%%%%%%%%%%%%%%%%%%%%%%%%%%%%%%%%%%%%%%%%
\end{proposition}
%%%%%%%%%%%%%%%%%%%%%%%%%%%%%%%%%%%%%%%%%%%%%%%%%%%%%%%%%%%%%%%%%%%%%%%%%%%%%%%%%%%%%%%%%%%%%%%%%%%%%%%

%%%%%%%%%%%%%%%%%%%%%%%%%%%%%%%%%%%%%%%%%%%%%%%%%%%%%%%%%%%%%%%%%%%%%%%%%%%%%%%%%%%%%%%%%%%%%%%%%%%%%%%
\begin{proof}
As already mentioned in the proof of Theorem~\ref{thm-segr-app} and elsewhere, 
there is a deterministic polynomial-time algorithm for the maximum $k$-set coverage problem with 
an approximation ratio of $\varrho$. 
For the given instance of \bfmc$\!\!(\chi,k)$,  
we run this algorithms (ignoring element colors) selecting $k$ sets, say $\cS_1,\dots,\cS_k$. 
Obviously, the total weight of all the elements covered in the selected solution is at least $\varrho\,\opt$.
Let 
$\alpha^+=\sum_{i=1}^k \lceil \nicefrac{|\cS_i|}{\chi} \rceil + \Delta$
and 
$\alpha^-=\sum_{i=1}^k \max \big\{ 1, \, \lfloor \nicefrac{|\cS_i|}{\chi} \rfloor - \Delta \big\}$. 
Note that $k\leq \alpha^- \leq \alpha^+ \leq \alpha^- + (2\Delta+1)k$.
Since each of the sets in the solution is balanced, 
an upper bound for the number $p_i$ of elements of color $i$ in the solution 
is given by 
$p_i \leq \alpha^+$.
Also note that by definition of $f$ 
we have 
$p_i \geq \frac{\alpha^-}{f}$.
It thus follows that for any $i$ and $j$ 
we have 
$
\nicefrac{p_i}{p_j} \leq f \times \frac{\alpha^+}{\alpha^-}
\leq (2+2\Delta) f
$.
\end{proof}
%%%%%%%%%%%%%%%%%%%%%%%%%%%%%%%%%%%%%%%%%%%%%%%%%%%%%%%%%%%%%%%%%%%%%%%%%%%%%%%%%%%%%%%%%%%%%%%%%%%%%%%

%%%%%%%%%%%%%%%%%%%%%%%%%%%%%%%%%%%%%%%%%%%%%%%%%%%%%%%%%%%%%%%%%%%%%%%%%%%%%%%%%%%%%%%%%%%%%%%%%%%%%%%
\section{Approximating \gfmc via randomized shifting}
\label{sec-gfmc}
%%%%%%%%%%%%%%%%%%%%%%%%%%%%%%%%%%%%%%%%%%%%%%%%%%%%%%%%%%%%%%%%%%%%%%%%%%%%%%%%%%%%%%%%%%%%%%%%%%%%%%%

We refer the reader to textbooks such as~\cite{V01} for a general overview of the randomized shifting 
technique (textbook~\cite{V01} illustrates the technique in the context of Euclidean travelling 
salesperson problem). 

%%%%%%%%%%%%%%%%%%%%%%%%%%%%%%%%%%%%%%%%%%%%%%%%%%%%%%%%%%%%%%%%%%%%%%%%%%%%%%%%%%%%%%%%%%%%%%%%%%%%%%%
\begin{theorem}\label{thm-gfmc}
For any constant $0<\eps<1$, 
we can design a randomized algorithm \alggeom for \gfmc with the following properties: 
%%%%%%%%%%%%%%%%%%%%%%%%%%%%%%%%%%%%%%%%%%%%%%%%%%%%%%%%%%%%%%%%%%%%%%%%%%%%%%%%%%%%%%%%%%%%%%%%%%%%%%%
\begin{enumerate}[label=\textbf{\emph{(}\alph*\emph{)}},leftmargin=*]
%%%%%%%%%%%%%%%%%%%%%%%%%%%%%%%%%%%%%%%%%%%%%%%%%%%%%%%%%%%%%%%%%%%%%%%%%%%%%%%%%%%%%%%%%%%%%%%%%%%%%%%
\item
\alggeom runs in $O( (\Delta /d)^d 2^{(Cd/\eps)^{O(d)}} k^d )$ time. 
%%%%%%%%%%%%%%%%%%%%%%%%%%%%%%%%%%%%%%%%%%%%%%%%%%%%%%%%%%%%%%%%%%%%%%%%%%%%%%%%%%%%%%%%%%%%%%%%%%%%%%%
\item
\alggeom satisfies the following properties with probability $1-o(1)$ $($cf.\ Inequality~\eqref{eq1}$'$--\eqref{eq1}$'')$:
%%%%%%%%%%%%%%%%%%%%%%%%%%%%%%%%%%%%%%%%%%%%%%%%%%%%%%%%%%%%%%%%%%%%%%%%%%%%%%%%%%%%%%%%%%%%%%%%%%%%%%%
%%%%%%%%%%%%%%%%%%%%%%%%%%%%%%%%%%%%%%%%%%%%%%%%%%%%%%%%%%%%%%%%%%%%%%%%%%%%%%%%%%%%%%%%%%%%%%%%%%%%%%%
\begin{enumerate}[label=$\triangleright$]
\item
The algorithm
covers at least 
$(1-O(\eps)) (\opt - \eps\chi)$ points.
%%%%%%%%%%%%%%%%%%%%%%%%%%%%%%%%%%%%%%%%%%%%%%%%%%%%%%%%%%%%%%%%%%%%%%%%%%%%%%%%%%%%%%%%%%%%%%%%%%%%%%%
\item
The algorithm
satisfies
the $(1+\eps)$-approximate coloring constraints $($cf.\ Inequality~\eqref{eq1}$'')$, \emph{\IE}, 
%%%%%%%%%%%%%%%%%%%%%%%%%%%%%%%%%%%%%%%%%%%%%%%%%%%%%%%%%%%%%%%%%%%%%%%%%%%%%%%%%%%%%%%%%%%%%%%%%%%%%%%
for all 
$i,j\in\{1,\dots,\chi\}$, 
$p_i \leq (1+\eps)p_j$.
%%%%%%%%%%%%%%%%%%%%%%%%%%%%%%%%%%%%%%%%%%%%%%%%%%%%%%%%%%%%%%%%%%%%%%%%%%%%%%%%%%%%%%%%%%%%%%%%%%%%%%%
\end{enumerate}
%%%%%%%%%%%%%%%%%%%%%%%%%%%%%%%%%%%%%%%%%%%%%%%%%%%%%%%%%%%%%%%%%%%%%%%%%%%%%%%%%%%%%%%%%%%%%%%%%%%%%%%
%%%%%%%%%%%%%%%%%%%%%%%%%%%%%%%%%%%%%%%%%%%%%%%%%%%%%%%%%%%%%%%%%%%%%%%%%%%%%%%%%%%%%%%%%%%%%%%%%%%%%%%
\end{enumerate}
%%%%%%%%%%%%%%%%%%%%%%%%%%%%%%%%%%%%%%%%%%%%%%%%%%%%%%%%%%%%%%%%%%%%%%%%%%%%%%%%%%%%%%%%%%%%%%%%%%%%%%%
\end{theorem}
%%%%%%%%%%%%%%%%%%%%%%%%%%%%%%%%%%%%%%%%%%%%%%%%%%%%%%%%%%%%%%%%%%%%%%%%%%%%%%%%%%%%%%%%%%%%%%%%%%%%%%%

%%%%%%%%%%%%%%%%%%%%%%%%%%%%%%%%%%%%%%%%%%%%%%%%%%%%%%%%%%%%%%%%%%%%%%%%%%%%%%%%%%%%%%%%%%%%%%%%%%%%%%%
\begin{remark}
In many geometric applications, the dimension parameter $d$ is a \emph{fixed} constant. 
For this case, \alggeom runs in polynomial time, and moreover, under the \emph{mild} assumption of 
$\frac{\opt}{\chi} \geq \eta$ for \emph{some} constant $\eta>1$, 
\alggeom covers at least $(1-O(\eps))\opt$ points, \IE, 
under these conditions \alggeom behaves like a randomized polynomial-time approximation scheme.
\end{remark}
%%%%%%%%%%%%%%%%%%%%%%%%%%%%%%%%%%%%%%%%%%%%%%%%%%%%%%%%%%%%%%%%%%%%%%%%%%%%%%%%%%%%%%%%%%%%%%%%%%%%%%%

\begin{proof}
Fix an optimal solution having 
$k$ unit balls 
$\cB_1^\ast,\ldots,\cB_k^\ast \subset \R^d$,
such that 
for all $i,j\in \{1,\dots,\chi\}$,
$\mu_i(\cB^\ast)=\mu_j(\cB^\ast)$, where
$
\cB^\ast = \bigcup_{i=1}^k \cB_i^\ast
$.
Thus, we need to show that 
our algorithm \alggeom 
computes in $(\Delta /d)^d 2^{(Cd/\eps)^{O(d)}} k^d$ time
a set of unit balls $B_1,\ldots,B_k \subset \mathbb{R}^d$
such that the following assertions hold 
with probability $1-o(1)$ (where 
$B = \bigcup_{i=1}^k B_i$):
%%%%%%%%%%%%%%%%%%%%%%%%%%%%%%%%%%%%%%%%%%%%%%%%%%%%%%%%%%%%%%%%%%%%%%%%%%%%%%%%%%%%%%%%%%%%%%%%%%%%%%%
\begin{gather*}
%%%%%%%%%%%%%%%%%%%%%%%%%%%%%%%%%%%%%%%%%%%%%%%%%%%%%%%%%%%%%%%%%%%%%%%%%%%%%%%%%%%%%%%%%%%%%%%%%%%%%%%
\sum_{i=1}^\chi \mu_i(B) > (1-O(\eps)) \sum_{i=1}^\chi (\mu_i(\cB^{\ast}) - \eps)
%%%%%%%%%%%%%%%%%%%%%%%%%%%%%%%%%%%%%%%%%%%%%%%%%%%%%%%%%%%%%%%%%%%%%%%%%%%%%%%%%%%%%%%%%%%%%%%%%%%%%%%
\\
%%%%%%%%%%%%%%%%%%%%%%%%%%%%%%%%%%%%%%%%%%%%%%%%%%%%%%%%%%%%%%%%%%%%%%%%%%%%%%%%%%%%%%%%%%%%%%%%%%%%%%%
\forall \, i,j\in \{1,\ldots,\chi\}: \, 
\mu_i(B) \leq (1+\eps) \mu_j(B)
%%%%%%%%%%%%%%%%%%%%%%%%%%%%%%%%%%%%%%%%%%%%%%%%%%%%%%%%%%%%%%%%%%%%%%%%%%%%%%%%%%%%%%%%%%%%%%%%%%%%%%%
\end{gather*}
%%%%%%%%%%%%%%%%%%%%%%%%%%%%%%%%%%%%%%%%%%%%%%%%%%%%%%%%%%%%%%%%%%%%%%%%%%%%%%%%%%%%%%%%%%%%%%%%%%%%%%%
%%Let $L=8d/\eps$ and $\delta=2^{-\Theta(d)}\eps /C$.
Set $L=8d/\eps$.
Let $G\subset \mathbb{R}^d$ be an axis-parallel grid such that every connected component of 
$\R^d \setminus G$ is an open $d$-dimensional hypercube isometric to $(0,L)^d$.
In other words, $G$ is the union of $d$ infinite families of axis-parallel 
$(d-1)$-dimensional 
%%hyperplanes, spaced apart by $\Delta$ in each orthonormal direction.
hyperplanes, spaced apart by $L$ in each orthonormal direction.
%%Let $\alpha\in [0,R)$ be chosen uniformly at random, and let
Let $\alpha\in [0,L)$ be chosen uniformly at random, and let
$
G' = G + \alpha
$
be the \emph{random} translation of $G$ by $\alpha$, \IE,
\[
(p_1+\alpha,\dots,p_d+\alpha)\in G' 
\,\,
\equiv 
\,\,
(p_1,\dots,p_d)\in G
\]
Let 
$F$ is the set of indices of all balls $\cB_i^\ast$ that has a non-empty intersection
with the randomly shifted grid $G'$, \IE, 
\[
F = \{i \in \{1,\dots,k\} : \cB^\ast_i\cap G' \neq \emptyset\}
\]
Let
$
\cB^{\ast,F} = \bigcup_{i\in F} \cB_i^\ast
$.
Any point $p\in \R^d$ is contained in $\cB^{\ast,F}$ only if it is contained in some unit ball intersecting $G'$.
%%Therefore, $p\in \cB^{*,F}$ only if it is at distance at most $p$ from $G'$; in other words, $\cB^{*,F}$ is 
Therefore, $p\in \cB^{\ast,F}$ only if it is at distance at most $2$ from $G'$ (in other words, $\cB^{\ast,F}$ is 
contained in the $2$-neighborhood of $G'$).
The probability that any particular point $p$ is at distance at most $2$ from any family of parallel randomly shifted 
%%hyperplanes in $G'$ is exactly $4/R$.
hyperplanes in $G'$ is exactly $4/L$.
By the union bound over all dimensions,
%%$\Pr[p\in \cB^{*,F}]\leq 4d/R$.
$\Pr[p\in \cB^{\ast,F}]\leq 
4d/L$.
Therefore, by the linearity of expectation,
%%%%
\[
\Ave{\mu_i(\cB^{\ast,F})} 
%%%%%%%%%%%%%%%%
=
\sum_{p\in \cup_{i=1}^k \cB_i^\ast} \!\!\!\! \prob{p\in \cB^{\ast,F}\,} \,\mu_i(p)
%%%%%%%%%%%%%%%%
\leq \frac{4d}{L} \mu_i(\cB^\ast)
\]
%%%%
Consequently, by  Markov's inequality, we get
%%%%
\[
\prob{\textstyle\mu_i(\cB^{\ast,F}) \geq \frac{8d}{R} \mu_i(\cB^\ast)} \leq \nicefrac{1}{2}
\]
%%%%
Set $\delta=2^{-\Theta(d)}\eps /C$.
Let $B_1^{\ast\ast},\ldots,B_k^{\ast\ast}$ be the collection of unit balls in $\R^d$ obtained as follows.
For each $i\in \{1,\dots,k\}\setminus F$,
obtain a unit ball by translating $\cB_i^\ast$ such that its center has coordinates that are 
\emph{integer multiples} of $\delta$, \IE, it is an element of the dilated integer lattice 
$\delta \cdot \mathbb{Z}^d$.
For every $i\in F$, we obtain a unit ball by picking an arbitrary ball obtained for some $j\in \{1,\dots,k\}\setminus F$ as 
described above.
Essentially, the new solution $B^{\ast\ast}_1,\dots,B^{\ast\ast}_k$ is missing all the balls that intersect $G'$, and 
rounds every other ball so that its center is contained in some integer lattice.
In this construction, each ball $\cB_i^\ast$, with $i\notin F$, gets translated by at most some distance $\sqrt{d} \delta$.
Since each for all $j\in \{1,\ldots,\chi\}$, $\mu_j$ is $C$-Lipschitz, it follows that 
\[
| \mu_j(\cB^{\ast}_i) - \mu_j(B^{\ast\ast}_i) | \leq \mathrm{vol}(\cB^\ast_i) \, \sqrt{d} \, \delta \, C \leq 2^{\Theta(d)} \delta C 
\]
Letting $B^{\ast\ast}=\bigcup_{i=1}^k B^{\ast\ast}_i$,
we get
that for all $i\in \{1,\dots,\chi\}$,
\[
| \mu_j(\cB^{\ast}) - \mu_j(B^{\ast\ast}) | \leq k 2^{\Theta(d)} \delta C 
\]
Let ${\cal I}$ be the set of connected components of $[0,\Delta]^d \setminus G'$.
We refer to the elements of ${\cal I}$ as \emph{cells}.
For each $A\in {\cal I}$, we enumerate the set, ${\cal S}_A$, of all possible subsets of at 
%%most $k$ unit balls with centers in $A\cap \delta\times \mathbb{Z}^d$.
most $k$ unit balls with centers in $A\cap ( \delta\cdot \mathbb{Z}^d )$.
%%There are at most $(R/\delta)^d$ lattice points in $A$, and thus there are at most $2^{(R/\delta)^d}$ such subsets of unit balls.
There are at most $(L/\delta)^d$ lattice points in $A$, and thus there are at most $2^{(L/\delta)^d}$ such subsets of unit balls.
%%Since there are at most $|{\cal I}\leq (\lceil \Delta/R \rceil)^d$, it follows that this enumeration takes time
Since $|{\cal I}|\leq (\lceil \Delta/L \rceil)^d$, it follows that this enumeration takes
%%$(\lceil \Delta/R \rceil)^d \cdot 2^{(R/\delta)^d}$ time.
$O((\Delta/L)^d \, 2^{(L/\delta)^d})$ time.

For each enumerated subset 
$J\in {\cal S}_A$ of unit balls, we record the vector 
\[
\left(|J|, \frac{\eps}{k} \left\lfloor \mu_1(X) \frac{k}{\eps}  \right\rfloor,\dots, 
\frac{\eps}{k} \left\lfloor \mu_k(X) \frac{k}{\eps}  \right\rfloor\right)
\]
where $X = \bigcup_{Y\in J} Y$.
There are at most $(2^{O(d)}k/\eps)^d$ such vectors for each cell in ${\cal I}$.
Via standard dynamic programming, we can inductively compute all possible sums of vectors 
such that we pick at most one vector from each cell, and the total sum of the first coordinate, \IE, the number of unit balls, 
%%is at most $i\in \{1,\ldots,k\}$.
is at most $k$.
%%This can be done in $(\Delta/R)^d 2^{(R/\delta)^d} (2^{O(d)}k/\eps)^d$ time.
This can be done in $O( (\Delta/L)^d 2^{(L/\delta)^d} (2^{O(d)}k/\eps)^d )$ time.
For the correct choice of vectors that corresponds to the solution $B^{\ast\ast}$, the sum 
of the vectors we compute is correct up to an additive factor of $\eps$ on each coordinate.
This means that we compute a solution $B_1,\ldots,B_k$, with 
the following property:
%%%
\begin{multline*}
\sum_{i=1}^k \mu_i(B) \geq (1-\eps) \sum_{i=1}^\chi \mu_i(B^{\ast\ast}) 
 \geq (1-\eps) \Big[ \sum_{i=1}^\chi \big( \, \mu_i(\cB^{\ast}) - 2^{\Theta(d)} \delta C - \mu_i(\cB^{\ast,F}) \, \big)  \Big]
\\
 \geq \big( \, 1-\eps- (8d/L) \, \big) \Big[ \sum_{i=1}^\chi \big( \, \mu_i(\cB^{\ast}) - 2^{\Theta(d)} \delta C\, \big) \Big]
 =
 \big( \, 1-2\eps \, \big) \Big[ \sum_{i=1}^\chi \big( \, \mu_i(\cB^{\ast}) - \eps \, \big) \Big]
\end{multline*}
%%%
with probability at least $\nicefrac{1}{2}$.
Repeating the algorithm $O(\log n)$ times and returning the best solution found, results in the high-probability assertion,
which concludes the proof.
\end{proof}

\section{Conclusion and open problems}

In this paper we formulated a natural combinatorial optimization framework for incorporating fairness issues in 
coverage problems and provided a set of approximation algorithms for the general version 
of the problem as well as its special cases. Of course, it is possible to design other optimization frameworks
depending on the particular application in hand, and we encourage researchers to do that. 
Below we list some future research questions related to our framework: 
%%%%%%%%%%%%%%%%%%%%%%%%%%%%%%%%%%%%%%%%%%%%%%%%%%%%%%%%%%%%%%%%%%%%%%%%%%%%%%%%%%%%%%%%%%%%%%%%%%%%%%%
\begin{description}[leftmargin=0.1in]
\item[Eliminating the gap of factor $f$ in $\LP$-relaxation:]
As noted in Section~\ref{sec-limit-lp}, all of our $\LP$-relaxations incur a gap of factor $f$ in the coloring
constraints while rounding. It seems non-trivial to close the gap using additional linear inequalities while 
preserving the same approximation ratio. However, it may be possible to improve the gap using 
$\SDP$-relaxations. 
%%%%%%%%%%%%%%%%%%%%%%%%%%%%%%%%%%%%%%%%%%%%%%%%%%%%%%%%%%%%%%%%%%%%%%%%%%%%%%%%%%%%%%%%%%%%%%%%%%%%%%%
\item[Primal-dual schema:]
Another line of attack for the \fmc problems is via the primal-dual approach~\cite{V01}.
For example, can the primal-dual approach for partial coverage problem by 
Gandhi, Khuller and Srinivasan~\cite{GKS04} be extended to \fmc? 
A key technical obstacle seems to center around effective interpretation of the dual 
of the coloring constraints. 
Our iterated rounding approach was able to go around this obstacle but 
the case when $\chi=\omega(1)$ may be improvable.
%%%%%%%%%%%%%%%%%%%%%%%%%%%%%%%%%%%%%%%%%%%%%%%%%%%%%%%%%%%%%%%%%%%%%%%%%%%%%%%%%%%%%%%%%%%%%%%%%%%%%%%
\item[Fixed parameter tractability:]
As mentioned in Section~\ref{sec-prior-coverage} fixed-parameter tractability issues for $k$-node coverage 
have been investigated by prior researchers such as 
Marx~\cite{M08} and Gupta, Lee and Li in~\cite{GLL18a,GLL18b}. It would be interesting to extend 
these results to \nfmc.
%%%%%%%%%%%%%%%%%%%%%%%%%%%%%%%%%%%%%%%%%%%%%%%%%%%%%%%%%%%%%%%%%%%%%%%%%%%%%%%%%%%%%%%%%%%%%%%%%%%%%%%
\item[{Generalizing to non-decreasing submodular set objective functions:}]
\added[comment={added}]{
The proofs and proof techniques in this paper do not generalize to the case when 
the objective function for our 
\fmc problems is a (more general) non-decreasing submodular set function. 
It would be interesting to devise new algorithmic techniques and proofs for this more
general case.
Approximation algorithms for such generalizations for the standard 
}
{
maximum $k$-set coverage problem 
(see Section~\ref{sec-prior-coverage})
were provided in~\cite{KST09}.
}
\end{description}
%%%%%%%%%%%%%%%%%%%%%%%%%%%%%%%%%%%%%%%%%%%%%%%%%%%%%%%%%%%%%%%%%%%%%%%%%%%%%%%%%%%%%%%%%%%%%%%%%%%%%%%

%%%%%%%%%%%%%%%%%%%%%%%%%%%%%%%%%%%%%%%%%%%%%%%%%%%%%%%%%%%%%%%%%%%%%%%%%%%%%%%%%%%%%%%%%%%%%%%%%%%%%%%

%%%%%%%%%%%%%%%%%%%%%%%%%%%%%%%%%%%%%%%%%%%%%%%%%%%%%%%%%%%%%%%%%%%%%%%%%%%%%%%%%%%%%%%%%%%%%%%%%%%%%%%

\appendix

\bigskip

\centerline{\bf\Large Appendix}

\section*{Proof of Lemma~\ref{lem1}}

%%\begin{proof}
%%
%%\smallskip
\noindent
{(\emph{a})}
We describe the proof for $\chi=2$; generalization to $\chi>2$ is obvious.
The reduction is from the \emph{Exact Cover by 3-sets} (X3C) problem which is defined as follows.
We are given an universe $\cU'=\{u_1,\dots,u_{n'}\}$ of $n'$ elements for some $n'$ that is a 
multiple of $3$, and a collection of 
$n'$ subsets $\cS_1,\dots,\cS_{n'}$ of $\cU$ such that $\bigcup_{j=1}^{n'} \cS_j=\cU$, 
every element of $\cU'$ occurs in exactly $3$ sets and $|\cS_j|=3$ for $j=1,\dots,n'$. 
The goal is to decide if there exists a collection of $\nicefrac{n'}{3}$ (disjoint) sets 
whose union is $\cU'$. X3C is known to be $\NP$-complete~\cite{GJ79}.
Given an instance $\langle \cU',\cS_1,\dots,\cS_{n'}\rangle$ of X3C as described, we create the
following instance $\langle \cU,\cS_1,\dots,\cS_{n'+1},k\rangle$ of \fmcg{2}{k}: 
%%%%%%%%%%%%%%%%%%%%%%%%%%%%%%%%%%%%%%%%%%%%%%%%%%%%%%%%%%%%%%%%%%%%%%%%%%%%%%%%%%%%%%%%%%%%%%%%%%%%%%%
\begin{enumerate}[label=(\emph{\roman*})]
\item 
The universe is 
$\cU=\{u_1,\dots,u_{n'}\}\cup \{u_{n'+1},\dots,u_{2n'}\}$ (and thus $n=2n'$),
%%%%%%%%%%%%%%%%%%%%%%%%%%%%%%%%%%%%%%%%%%%%%%%%%%%%%%%%%%%%%%%%%%%%%%%%%%%%%%%%%%%%%%%%%%%%%%%%%%%%%%%
\item 
$w(u_j)=1$ for $j=1,\dots,2n'$,
%%%%%%%%%%%%%%%%%%%%%%%%%%%%%%%%%%%%%%%%%%%%%%%%%%%%%%%%%%%%%%%%%%%%%%%%%%%%%%%%%%%%%%%%%%%%%%%%%%%%%%%
\item 
the sets are 
$\cS_1,\dots,\cS_{n'}$ and a new set $\cS_{n'+1}=\{u_{n'+1},\dots,u_{2n'}\}$, 
%%%%%%%%%%%%%%%%%%%%%%%%%%%%%%%%%%%%%%%%%%%%%%%%%%%%%%%%%%%%%%%%%%%%%%%%%%%%%%%%%%%%%%%%%%%%%%%%%%%%%%%
\item 
the coloring function is given by
$
\cC(u_j) = \left\{
\begin{array}{r l}
1, & \mbox{if $1\leq j\leq n'$} 
\\
2, & \mbox{otherwise} 
\end{array}
\right.
$,
%%%%%%%%%%%%%%%%%%%%%%%%%%%%%%%%%%%%%%%%%%%%%%%%%%%%%%%%%%%%%%%%%%%%%%%%%%%%%%%%%%%%%%%%%%%%%%%%%%%%%%%
and 
%%%%%%%%%%%%%%%%%%%%%%%%%%%%%%%%%%%%%%%%%%%%%%%%%%%%%%%%%%%%%%%%%%%%%%%%%%%%%%%%%%%%%%%%%%%%%%%%%%%%%%%
\item 
$k=\frac{n'}{3}+1 = \frac{n}{6}+1$. 
\end{enumerate}
%%%%%%%%%%%%%%%%%%%%%%%%%%%%%%%%%%%%%%%%%%%%%%%%%%%%%%%%%%%%%%%%%%%%%%%%%%%%%%%%%%%%%%%%%%%%%%%%%%%%%%%
%
Clearly, 
every element of $\cU$ occurs in no more than $3$ sets and 
all but the set $\cS_{n'+1}$ contains exactly $3$ elements.
The proof is completed once the following is shown:  

\begin{quote}
($\ast$) 
\emph{the given instance of X3C has a solution if and only if 
the transformed instance of} \fmcg{2}{1+\nicefrac{n}{6}} \emph{has a solution}.
\end{quote}

\noindent
A proof of 
($\ast$) 
is easy: since the set 
$\cS_{n'+1}$ must appear in any valid solution of \fmc, 
a solution 
$\cS_{i_1},\dots,\cS_{i_{n'/3}}$ of X3C corresponds to a solution 
$\cS_{i_1},\dots,\cS_{i_{n'/3}},\cS_{n'+1}$ of \fmcg{2}{k} and \emph{vice versa}.

\bigskip
\noindent
(\emph{b})
The proof is similar to that in 
(\emph{a}) 
but now
instead of X3C
we reduce the node cover problem for cubic (\IE, $3$-regular) graphs (VC$_3$)
which is defined as follows: 
given a cubic graph $G=(V,E)$ of $n'$ nodes and $3n'/2$ edges and an integer $k'$,
determine if there is a set of $k'$ nodes that cover all the edges. 
$VC_3$ is known to be $\NP$-complete even if $G$ is planar~\cite{GJ79}. 
For the translation to an instance of 
\fmcg{2}{k}, 
edges of $G$ are colored with color $1$, 
we add a \emph{new} connected component 
$\cK_{(3n'/2)+1}$ 
to $G$ that is a complete graph of 
$(3n'/2)+1$ nodes with every edge having color $2$, 
transform this to the set-theoretic version of \fmc using the standard 
transformation from node cover to set cover and set 
$k=k'+1$; note that $n=3n'/2 + \binom{(3n'/2)+1}{2}=\Theta((n')^2)$ and $a=3n'/2=O(\sqrt{n}\,)$. 
To complete the proof, note that any feasible solution for the
\fmcg{2}{k} instance 
must contain \emph{exactly} one node from 
$\cK_{(3n'/2)+1}$ 
covering 
$3n'/2$ edges 
and therefore the solution for the edges with color $1$ must correspond to a node cover in $G$ (and \emph{vice versa}).

\bigskip
\noindent
(\emph{c})
We given a different reduction from X3C.
Given an instance $\langle \cU',\cS_1,\dots,\cS_{n'}\rangle$ of X3C as in 
(\emph{a}), 
we create the
following instance $\langle \cU,\cT_1,\dots,\cT_{n'},k\rangle$ of \fmcg{n'}{k}: 
%%%%%%%%%%%%%%%%%%%%%%%%%%%%%%%%%%%%%%%%%%%%%%%%%%%%%%%%%%%%%%%%%%%%%%%%%%%%%%%%%%%%%%%%%%%%%%%%%%%%%%%
\begin{enumerate}[label=(\emph{\roman*})]
%%%%%%%%%%%%%%%%%%%%%%%%%%%%%%%%%%%%%%%%%%%%%%%%%%%%%%%%%%%%%%%%%%%%%%%%%%%%%%%%%%%%%%%%%%%%%%%%%%%%%%%
\item 
For every set $\cS_i=\big\{u_{i_1},u_{i_2},u_{i_3}\big\}$ of X3C we have three elements 
$u_{i_1}^i,u_{i_2}^i,u_{i_3}^i$ and a set 
$\cT_i=\big\{ u_{i_1}^i,u_{i_2}^i,u_{i_3}^i\big\}$ in \fmc
(and thus $n=3n'$, $a=3$ and $f=1$),
%%%%%%%%%%%%%%%%%%%%%%%%%%%%%%%%%%%%%%%%%%%%%%%%%%%%%%%%%%%%%%%%%%%%%%%%%%%%%%%%%%%%%%%%%%%%%%%%%%%%%%%
\item 
$w(u_{i_j}^i)=1$ 
and 
$\cC(u_{i_j}^i) = i_j$ 
for $i\in \{1,\dots,n'\},j\in\{1,2,3\}$
(and thus $\chi=n'=\nicefrac{n}{3}$),
%%%%%%%%%%%%%%%%%%%%%%%%%%%%%%%%%%%%%%%%%%%%%%%%%%%%%%%%%%%%%%%%%%%%%%%%%%%%%%%%%%%%%%%%%%%%%%%%%%%%%%%
\item 
$k=\nicefrac{n'}{3} = \nicefrac{n}{9}$. 
\end{enumerate}
%%%%%%%%%%%%%%%%%%%%%%%%%%%%%%%%%%%%%%%%%%%%%%%%%%%%%%%%%%%%%%%%%%%%%%%%%%%%%%%%%%%%%%%%%%%%%%%%%%%%%%%
%
The proof is completed by showing 
the given instance of X3C has a solution if and only if 
the transformed instance of \fmcg{\nicefrac{n}{3}}{\nicefrac{n}{9}} has a solution.
This can be shown as follows.
We include the set $\cT_i$ in the solution for \fmc 
if and only if the set $\cS_i$ is in the solution for X3C. 
For any valid solution of X3C and every $j\in \{1,\dotsc,n'\}$ 
the element $u_j\in\cU'$ appears in \emph{exactly one} set, say $\cS_\ell=\big\{u_{\ell_1},u_{\ell_2},u_{\ell_3}\big\}$,
of X3C where one of the elements, say $u_{\ell_1}$, is $u_j$. 
Then, the solution of \fmc contains exactly one element, namely the element $u_{\ell_1}^\ell$, of color 
$\ell_1=j$. 
Conversely, given a feasible solution of \fmc 
with at most $k\leq \nicefrac{n}{9}$ sets, 
first note that if $k<\nicefrac{n}{9}$ then the total number of colors of various elements in the solution is 
$3k<n'$ and thus the given solution is not valid. Thus, $k=\nicefrac{n}{9}$
and therefore the solution of X3C contains $\nicefrac{n}{9}=\nicefrac{n'}{3}$ sets. 
Now, for every color $j$ the solution of \fmc contains a set, 
say $\cT_\ell=\big\{u_{\ell_1}^\ell,u_{\ell_2}^\ell,u_{\ell_3}^\ell\big\}$, 
containing an element of color $j$, say the element 
$u_{\ell_1}^\ell$. Then $\ell_1=j$ and the element $u_j$ appears in a set in the solution of X3C. 
To see that 
remaining claims about the reduction, 
there is no solution of \fmc that includes at least one element of every color and that is not a solution 
of X3C.
%%\end{proof}

%
\end{document}